\documentclass{lmcs}
\pdfoutput=1

\usepackage{lastpage}
\lmcsdoi{22}{4}{14}
\lmcsheading{}{\pageref{LastPage}}{}{}%
{Oct.~15,~2025}{Nov.~16,~2026}{}

\usepackage[utf8]{inputenc}
\usepackage[T1]{fontenc}

\usepackage{amsmath,amssymb,amsfonts,amsthm}
\usepackage{mathtools}
\usepackage{hyperref}
\usepackage{bm}
\usepackage{thm-restate}
\usepackage{xspace}
\usepackage{booktabs}
\newcommand{\envalias}[2]{\newenvironment{#1}{\begin{#2}}{\end{#2}}}
\envalias{theorem}{thm}
\envalias{corollary}{cor}
\envalias{lemma}{lem}
\envalias{proposition}{prop}
\envalias{remark}{rem}
\envalias{example}{exa}
\envalias{definition}{defi}
\newtheorem{claim}[thm]{Claim}
\envalias{conjecture}{conj}

\usepackage{cleveref}
\crefname{thm}{theorem}{theorems}
\crefname{cor}{corollary}{corollaries}
\crefname{lem}{lemma}{lemmas}
\crefname{prop}{proposition}{propositions}
\crefname{rem}{remark}{remarks}
\crefname{exa}{example}{examples}
\crefname{defi}{definition}{definitions}
\crefname{claim}{claim}{claims}
\crefname{conj}{conjecture}{conjectures}

\usepackage{tikz}
\usetikzlibrary{decorations.pathreplacing,calc,arrows}
\usetikzlibrary{patterns}

\usepackage{algorithm}
\usepackage[noend]{algpseudocode}

\algnewcommand\algorithmicfixed{\textbf{Fixed:}}
\algnewcommand\Fixed{\item[\algorithmicfixed]}
\algnewcommand\algorithmicswitch{\textbf{switch}}
\algnewcommand\algorithmiccase{\textbf{case}}
\algnewcommand\algorithmicforeach{\textbf{foreach}}
\algnewcommand\algorithmicnondet{\textbf{nondet}}
\algnewcommand\algorithmicor{\textbf{or}}
\algnewcommand\algorithmicassert{\textbf{assert}}
\algnewcommand\algorithmiclet{\textbf{let}}
\algnewcommand\Assert[1]{\State \algorithmicassert(#1)}
\algdef{SE}[SWITCH]{Switch}{EndSwitch}[1]{\algorithmicswitch\ #1\ \algorithmicdo}{\algorithmicend\ \algorithmicswitch}%
\algtext*{EndSwitch}%
\algdef{SE}[CASE]{Case}{EndCase}[1]{\algorithmiccase\ #1\ \textbf{:}}{\algorithmicend\ \algorithmiccase}%
\algtext*{EndCase}%
\algdef{S}[FOR]{ForEach}[1]{\algorithmicforeach\ #1\ \algorithmicdo}%
\algdef{SE}[NONDET]{Nondet}{EndNondet}{\algorithmicnondet\ }{\algorithmicend\ \algorithmicnondet}%
\algtext*{EndNondet}%
\algdef{SE}[Or]{Or}{EndOr}{\algorithmicor\ }{\algorithmicend\ \algorithmicor}%
\algtext*{EndOr}%
\algdef{SE}[LET]{Let}{EndLet}{\algorithmiclet\ }{\algorithmicend\ \algorithmiclet}%
\algtext*{EndLet}%

\AtBeginEnvironment{algorithm}{\linespread{1.05}}

\usepackage{todonotes}
\colorlet{jorge}{red!50!black!80}
\colorlet{alessio}{blue!70!black!80}



\newcommand{\newextmathcommand}[2]{%
    \newcommand{#1}{\ensuremath{#2}\xspace}
}

\newcommand{\renewextmathcommand}[2]{%
    \renewcommand{#1}{\ensuremath{#2}\xspace}
}

\makeatletter
\newcommand{\labeltext}[2]{%
  #1%
  \@bsphack%
  \csname phantomsection\endcsname 
  \def\@currentlabel{#1}{\label{#2}}%
  \@esphack%
}
\makeatother

\makeatletter
               {\list{}{\leftmargin=0pt
                        \labelwidth\z@ \itemindent-\leftmargin
                        }}%
               {\endlist}
\makeatother

\newextmathcommand{\N}{\mathbb{N}}
\newextmathcommand{\Np}{\Nat_+}
\newextmathcommand{\Z}{\mathbb{Z}}
\newextmathcommand{\Q}{\mathbb{Q}}
\newextmathcommand{\R}{\mathbb{R}}
\newextmathcommand{\B}{\mathbb{B}}
\newextmathcommand{\D}{\mathbb{D}}
\newextmathcommand{\FP}{\mathbb{D}_{\textit{fp}}}

\newextmathcommand{\fptime}{\textup{\textsc{FP}}}
\newextmathcommand{\ptime}{\textup{\textsc{P}}}
\newextmathcommand{\np}{\textup{\textsc{NP}}}
\newextmathcommand{\pspace}{\textup{\textsc{PSpace}}}
\newextmathcommand{\nexptime}{\textup{\textsc{NExpTime}}}
\newextmathcommand{\exptime}{\textup{\textsc{ExpTime}}}
\newextmathcommand{\twoexptime}{\textup{\textsc{2Exp}}}
\newextmathcommand{\threeexptime}{\textup{\textsc{3Exp}}}
\newextmathcommand{\expspace}{\textup{\textsc{ExpSpace}}}
\newextmathcommand{\twoexpspace}{\textup{\textsc{2ExpSpace}}}
\newextmathcommand{\tower}{\textup{\textsc{Tower}}}
\newextmathcommand{\cclass}{\textup{\textsc{C}}}

\newextmathcommand{\poly}{\textup{poly}}

\renewextmathcommand{\phi}{\varphi}
\renewcommand{\vec}{\bm}

\newextmathcommand{\lcm}{{\rm lcm}}

\newcommand{\abs}[1]{\ensuremath{\left|#1\right|}\xspace}
\newcommand{\ceil}[1]{\ensuremath{\left\lceil#1\right\rceil}\xspace}
\newcommand{\floor}[1]{\ensuremath{\left\lfloor#1\right\rfloor}\xspace}

\newextmathcommand{\dom}{\textup{dom}}

\DeclareFontEncoding{LS2}{}{\noaccents@}
  \DeclareFontSubstitution{LS2}{stix}{m}{n}
  \DeclareSymbolFont{stix@largesymbols}{LS2}{stixex}{m}{n}
  \SetSymbolFont{stix@largesymbols}{bold}{LS2}{stixex}{b}{n}
  \DeclareMathDelimiter{\lBrace}{\mathopen} {stix@largesymbols}{"E8}%
                                            {stix@largesymbols}{"0E}
  \DeclareMathDelimiter{\rBrace}{\mathclose}{stix@largesymbols}{"E9}%
                                            {stix@largesymbols}{"0F}

\newextmathcommand{\defeq}{\coloneqq}
\newextmathcommand{\eqdef}{\defeq}

\newcommand{\sub}[2]{\ensuremath{[#1\,/\,#2]}\xspace}

\definecolor{light-gray}{gray}{0.95}

\newcommand{\myguess}{\textbf{guess}\xspace}
\newcommand{\myreturn}{\textbf{return}\xspace}

\algnewcommand\algorithmicndbranchoutput{\textbf{Output of each branch ($\beta$):}}
\algnewcommand\NDBranchOutput{\item[\algorithmicndbranchoutput]}
\algnewcommand\algorithmicglobalspec{\textbf{Ensuring:}}
\algnewcommand\GlobalSpec{\item[\algorithmicglobalspec]}

\newextmathcommand{\Cmap}{\textup{\textsc{C}}}

\newextmathcommand{\SIGN}{\textup{\textsc{Sign}}}
\newextmathcommand{\ERisk}{\textup{\textsc{ERisk}}}

\newextmathcommand{\cn}{\xi} 
\newextmathcommand{\alg}{\alpha}

\newextmathcommand{\dotb}{\mathbin{\scalebox{0.7}{$\bullet$}}}

\newextmathcommand{\height}{\textup{h}}
\newextmathcommand{\size}{\textup{size}}

\newcommand{\ipow}[1]{\ensuremath{{#1}^{\Z}}\xspace}

\begin{document}

\title[On $\exists \mathbb{R}$ with Integer Powers of a Computable Number]{On the Existential Theory of the Reals Enriched with Integer Powers of a Computable Number}
\titlecomment{{\lsuper*}This is an extended version of the conference paper~\cite{GallegoM25}.}

\specialissue{Selected Papers of the 42nd International Symposium on Theoretical Aspects of Computer Science (STACS 2025) - Track B}

\author{Jorge Gallego-Hernández\lmcsorcid{0009-0002-2240-1107}}[a]	
\author{Alessio Mansutti\lmcsorcid{0000-0002-1104-7299}}[b]	

\address{IMDEA Software Institute and Universidad Politécnica de Madrid, Madrid, Spain}	
\email{jorge.gallego@imdea.org, j.gallegoh@alumnos.upm.es}  

\address{IMDEA Software Institute, Madrid, Spain}	
\email{alessio.mansutti@imdea.org}  

\keywords{Theory of the reals with exponentiation, decision procedures, computability} 

\begin{abstract}
  This paper investigates the extension of the existential theory of 
  the reals by an additional unary predicate~${\cn}^{\Z}$ for the integer 
  powers of a fixed computable real
  number~$\cn > 0$. 
  We denote this extension by $\exists\R(\ipow{\cn})$.
  If all we have access to is a Turing machine computing $\cn$,
  it is not possible to decide whether an input formula from this theory
  is satisfiable. However, we show an algorithm to decide this problem when 
  \begin{itemize}
    \item $\cn$ is known to be transcendental, or
    \item $\cn$ is a root of some given integer polynomial (that is, $\cn$ is
    algebraic).
  \end{itemize}
  In other words, knowing the algebraicity of $\cn$ suffices to circumvent
  undecidability. Furthermore, we establish complexity results under the proviso
  that~$\cn$ enjoys what we call a \emph{polynomial root barrier}. Using this
  notion, we show that the satisfiability problem of~$\exists\R(\ipow{\cn})$
  \begin{itemize}
    \item is in~\nexptime if $\cn$ is a natural number,
    \item is in~\expspace if~$\cn$ is an algebraic number, and
    \item is in~\threeexptime if~$\cn$ belongs to a family of transcendental numbers including~$\pi$ and~Euler's $e$.
  \end{itemize}
  \noindent 
  To establish our results, we first observe that the satisfiability problem
  of~$\exists\R(\ipow{\cn})$ reduces in exponential time to the problem of
  solving quantifier-free instances of the theory of the reals where variables
  range over~${\cn}^{\Z}$. We then prove that these instances have a \emph{small
  witness property}: only finitely many integer powers of $\cn$ must be
  considered to find whether a formula is satisfiable. Our complexity results
  are shown by relying on well-established machinery from Diophantine
  approximation and transcendental number theory, such as bounds for the
  transcendence measure of numbers.
  
  As a by-product of our results, we are able to remove the appeal to Schanuel's
  conjecture from the proof of decidability of the entropic risk
  threshold problem for stochastic games with rational probabilities, rewards, and
  threshold~[Baier et al., \textit{MFCS}, 2023]: when the base of the entropic risk is $e$
  and the aversion factor is a fixed algebraic number, the problem is
  (unconditionally) in~\exptime.
\end{abstract}

\maketitle

\noindent
\begin{minipage}{0.85\linewidth}
\footnotesize{\paragraph*{Funding.}
This work is part of a project co-funded by the European Union (GA 101154447) 
and by MCIN/AEI (GA CEX2024-001471-M and PID2022-138072OB-I00).
Views and opinions expressed are however those of the authors only and do not necessarily reflect those of the European Union or European Commission. Neither the European Union nor the granting authority can be held responsible for them.}
\end{minipage}%
\begin{minipage}{0.15\linewidth}
\flushright
\includegraphics[scale=0.25]{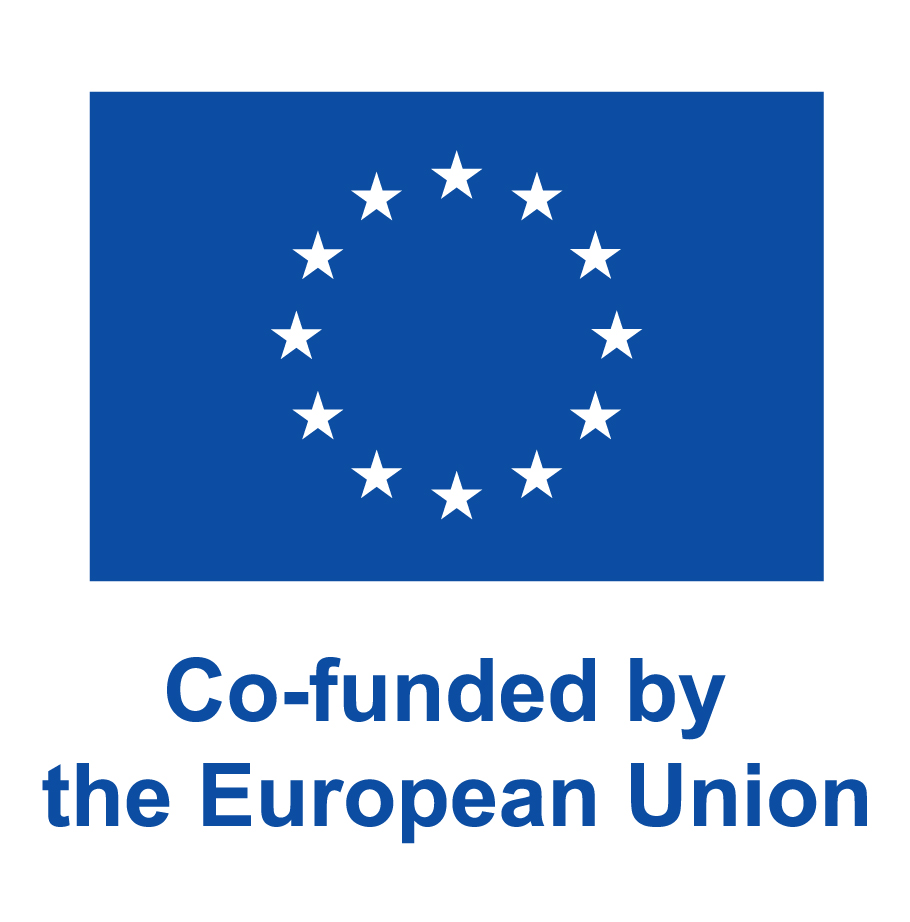}
\end{minipage}

\section{Introduction}
\label{sec:introduction}

Tarski's exponential function problem asks to determine the decidability of the
validity problem from the first-order (FO) theory of the structure $(\R; 0, 1,
+, \cdot, e^x, <, =)$. This theory, hereinafter denoted $\R(e^x)$, extends the
FO theory of the reals (a.k.a.~Tarski arithmetic) with the exponential
function~${x \mapsto e^x}$. A celebrated result by Macintyre and Wilkie
establishes an affirmative answer to Tarski's problem
conditionally to the truth of Schanuel's conjecture, a profound conjecture from
transcendental number theory~\cite{MacWilkie96}. Recent years have seen this
result being used as a black-box to establish conditional decidability results
for numerous problems stemming from dynamical systems~\cite{Dantam21,Almagor22}
automata theory~\cite{Daviaud21,ChistikovKMP22}, neural networks
verification~\cite{HankalaHKV24,IsacZBK23}, the theory of stochastic
games~\cite{BaierCMP23}, and differential privacy~\cite{BartheCKS021}. 

As it is often the case when appealing to a result as a black-box, some of the
computational tasks resolved by relying on the work in~\cite{MacWilkie96} do not
require the full power of~$\R(e^x)$. Consequently, it is natural to ask whether
some of these tasks can be tackled without relying on unproven conjectures,
perhaps by reduction to tame fragments or variants of~$\R(e^x)$. A~few~results
align with this question:
\begin{itemize}
  \item In the papers~\cite{AnaiW00,AchatzMW08,McCallumW12}, Achatz,
    Anai, McCallum and Weispfenning introduce a procedure to decide sentences of
    the form $\exists x \exists y : y = \text{trans}(x) \land \phi(x,y)$, where~$\phi$ 
    is a formula from Tarski arithmetic, and $x \mapsto \text{trans}(x)$
    is any analytic and strongly transcendental function
    (see~\cite[Section~2]{McCallumW12} for the precise definition). 
    Since $x \mapsto e^x$ enjoys such
    properties, this result shows a non-trivial fragment of $\R(e^x)$ that is
    unconditionally decidable. The procedure is implemented in the tool
    Redlog~\cite{DolzmannS97}. No complexity bound is known.

  \item In~\cite{Dries1986}, van den Dries proves decidability of the extension
    of Tarski arithmetic with the unary predicate~$\ipow{2}$ interpreted as the
    set $\{2^i : i \in \Z\}$, i.e., the set of all integer powers of $2$. While
    this result is achieved by model-theoretic arguments, an effective
    quantifier elimination procedure was later given by Avigad and
    Yin~\cite{AvigadY07}. Their procedure runs in~\tower, and in fact it
    requires non-elementary time already for the elimination of a single
    quantified variable. The choice of the base $2$ for the integer powers is
    arbitrary: in~\cite{DriesG06}, the decidability is extended to 
    any fixed computable real number (see also~\cite{miller2005tameness}), 
    and Avigad and Yin procedure extends to any algebraic number $\eta > 1$~\cite[Section~4]{AvigadY07} 
    (i.e., $\eta$ is a root of some polynomial equation;
    see~\Cref{sec:preliminaries} for background knowledge on computable, algebraic and transcendental numbers).
    Considering any two $\alpha,\beta \in \R$
    satisfying~$\ipow{\alpha} \cap \ipow{\beta} = \{1\}$ yields undecidability,
    as shown by Hieronymi in~\cite{Hieronymi10}.
\end{itemize}

\noindent 
When comparing the two lines of work discussed above, it becomes apparent that
there is a balance to be struck between reasoning about transcendental numbers,
the path followed by the first set of works, and developing algorithms that are
well-behaved from a complexity standpoint, the path taken in particular in~\cite{AvigadY07}. Our aim with this paper is to somewhat bridge this gap: we add
to the second line of work by studying predicates for integer powers of
bases that may be transcendental, all the while maintaining complexity upper bounds. 

From now on, we write $\exists\R(\ipow{\cn})$ to denote the existential fragment
of the FO theory of the structure $(\R; 0, 1, \cn, +, \cdot, \ipow{\cn}, <,=)$,
where $\cn > 0$ is a fixed real number. In this paper, we examine the complexity
of deciding the satisfiability problem of $\exists\R(\ipow{\cn})$ for different
choices of the number $\cn$. The following theorem summarises our results.

\vbox{%
\begin{restatable}{theorem}{TheoremRootBarrier}
  \label{theorem:result-root-barrier}
    Fix a real number $\cn > 0$. The satisfiability problem for
    $\exists\R(\ipow{\cn})$ is
    \begin{enumerate}
      \item\label{theorem:result-root-barrier:point0} in \nexptime whenever
      $\cn$ is a natural number;
      \item\label{theorem:result-root-barrier:point1} in \expspace whenever
      $\cn$ is an algebraic number;
      \item\label{theorem:result-root-barrier:point2} in \threeexptime if
        $\cn \in \{\pi,\, e^\pi,\, e^\eta,\, \alg^\eta,\, \ln(\alg),\, 
        \frac{\ln(\alg)}{\ln(\beta)} : \alg,\beta,\eta \text{ algebraic with } \alg, \beta > 1\}$;
      \item\label{theorem:result-root-barrier:point3} decidable whenever $\cn$
      is a computable transcendental number.
    \end{enumerate}
\end{restatable}%
}


\noindent
\Cref{theorem:result-root-barrier} has a catch, however.
To be effective, the algorithm for deciding~$\exists\R(\ipow{\cn})$
requires: 
\begin{itemize}
  \item For~\Cref{theorem:result-root-barrier}(\ref{theorem:result-root-barrier:point1}),
  to have access to a canonical representation (see~\Cref{sec:preliminaries}) of $\cn$.
  \item In the cases covered
  by~\Cref{theorem:result-root-barrier}(\ref{theorem:result-root-barrier:point2}),
  to have access to representations of~$\alg$,~$\beta$, and~$\eta$.
  \item In the case of $\cn$ computable transcendental number
  (\Cref{theorem:result-root-barrier}(\ref{theorem:result-root-barrier:point3})),
  to have access to a Turing machine $T$ that computes $\cn$ (that is,~given an
  input~$n \in \N$ written in unary, $T$ returns a rational number $T_n$ such
  that $\abs{\cn-T_n} \leq 2^{-n}$).
\end{itemize}
That is, \Cref{theorem:result-root-barrier} 
gives an effective procedure for~$\exists\R(\ipow{\cn})$
for every fixed computable real number $\cn > 0$ given with a canonical 
representation, 
and in many cases the procedure comes with a complexity guarantee.


The results in~\Cref{theorem:result-root-barrier} are obtained by
\textbf{(i)} reducing the satisfiability problem for~$\exists\R(\ipow{\cn})$ to the
problem of solving instances of~$\exists\R(\ipow{\cn})$ where all variables
range over~$\ipow{\cn}$, and \textbf{(ii)}~showing that a solution over $\ipow{\cn}$ can
be found by only looking at a ``small'' set of integer powers of~$\cn$ (a
\emph{small witness property}). In proving Step~(ii), we also obtain a
quantifier elimination procedure for \emph{sentences}
of~$\exists\R(\ipow{\cn})$, that is formulae where no variable occurs free. This
procedure provides a partial answer to the question raised in~\cite{AvigadY07}
regarding the complexity of removing a single existential variable
in Tarski arithmetic extended with $\ipow{2}$:
within sentences of the existential fragment, such an elimination step 
can be performed in elementary time.



Coming back to our initial question on identifying computational tasks that
might not need the full power of $\R(e^x)$, as a by-product of our results we
show that the entropic risk threshold problem for stochastic games
studied by Baier, Chatterjee, Meggendorfer and Piribauer~\cite{BaierCMP23} is
unconditionally decidable in~\exptime even when the base of the entropic risk is
$e$ (or algebraic) and the aversion factor is any (fixed) algebraic number.

\section{Approaching complexity bounds with root barriers}\label{section:intro-to-root-barriers}
Theorems~\ref{theorem:result-root-barrier}(\ref{theorem:result-root-barrier:point1})
and \ref{theorem:result-root-barrier}(\ref{theorem:result-root-barrier:point2})
are instances of a more general result concerning classes of computable real
numbers. To properly introduce this result, it is beneficial to go back to
Macintyre and Wilkie's work on $\R(e^x)$. The exact statement made
in~\cite{MacWilkie96} is that~$\R(e^x)$ is decidable as soon as the following
computational problem, implied by Schanuel's conjecture, is established:

\begin{conjecture}
  \label{conjecture:WSC}
  There is a procedure that for input $f_1,\dots,f_n,g \in
  \Z[x_1,\dots,x_n,e^{x_1}\!,\dots,e^{x_n}]$, with $n \geq 1$, returns a positive
  integer $t$ with the following property: for every non-singular\footnote{A
  solution~$\vec \alpha$ of $\bigwedge_{i=1}^nf_i(\vec x) = 0$ is said to be
  non-singular whenever the determinant of the $n \times n$ Jacobian matrix
  $\frac{\partial(f_1,\dots,f_n)}{\partial(x_1,\dots,x_n)}$ is, once evaluated
  at $\vec \alpha$, non-zero. We give this definition only for completeness of
  the discussion on~\Cref{conjecture:WSC}. It is not used in this paper.}
  solution $\vec \alpha \in \R^n$ of the system of equalities
  $\bigwedge_{i=1}^nf_i(\vec x) = 0$, either $g(\vec \alpha) = 0$ or
  $\abs{g(\vec \alpha)} > t^{-1}$.
\end{conjecture}

\noindent
Above, $\Z[x_1,\dots,x_n,e^{x_1}\!,\dots,e^{x_n}]$ is the set of all $n$-variate
exponential-polynomials with integer coefficients. As remarked
in~\cite{MacWilkie96}, $t$ is guaranteed to exist by Khovanskii's
theorem~\cite{Khovanskii91}, hence the crux of the problem concerns how
to effectively compute such a number starting from $f_1,\dots,f_n$ and $g$. 
The purpose of the dichotomy ``either $g(\vec \alpha) = 0$ or
${\abs{g(\vec{\alpha})} > t^{-1}}$'' is in part to resolve what is a
fundamental problem when working with computable real numbers. Let
$\vec{\alpha}$ be a vector of computable numbers. Consider the problem of
establishing, given in input a polynomial $p$ with integer coefficients, whether
$p(\vec{\alpha})$ is positive, negative, or zero. This \emph{polynomial sign
evaluation} task is a well-known undecidable problem~\cite[Section~1.3]{Bournez24}. Intuitively, the
undecidability arises from the possibility that any approximation
$\vec{\alpha}^*$ of $\vec \alpha$ might yield $p(\vec{\alpha}^*) \neq 0$, even
though $p(\vec \alpha) = 0$. However, when working under the hypothesis that
either $p(\vec \alpha) = 0$ or $\abs{p(\vec \alpha)} > t^{-1}$, the problem
becomes decidable: it suffices to compute an approximation $\vec \alpha^*$
enjoying $|p(\vec \alpha) - p(\vec{\alpha}^*)| < (2t)^{-1}$, and then check
whether $|p(\vec \alpha^*)| \leq (2t)^{-1}$. If the answer is positive, then
$p(\vec{\alpha}) = 0$, otherwise $p(\vec{\alpha})$ and $p(\vec \alpha^*)$ have
the same sign.

The same issue occurs in~$\exists\R(\ipow{\cn})$: under the sole hypothesis that
$\cn$ is computable, we cannot even check if $\cn = 2$ holds. However, what we
can do is to draw some inspiration from~\Cref{conjecture:WSC}, and introduce as
a further assumption the existence of what we call a \emph{root barrier} of
$\cn$. Below, $\N_{\geq 1} = \{1,2,3,\dots\}$, and given a polynomial $p$ we
write $\deg(p)$ for its \emph{degree} and $\height(p)$ for its \emph{height}
(i.e., the maximum absolute value of a coefficient~of~$p$).

\begin{definition}
  \label{definition:intro:root-barrier}
  A function $\sigma \colon (\N_{\geq 1})^2 \to \N$ is a \emph{root barrier} of $\cn
  \in \R$ if for every non-constant polynomial $p(x)$ with integer coefficients, $p(\cn) =
  0$ or $\ln \abs{p(\cn)} \geq {-\sigma(\deg(p), \height(p))}$.
\end{definition}

To avoid non-elementary bounds on the runtime of our algorithms, we focus on
computable numbers having root barriers $\sigma(d,h)$ that are polynomial
expressions of the form ${c \cdot (d + \ceil{\ln h})^k}$, where $c,k \in \N$
are some positive constants and $\ceil{\cdot}$ is the ceiling function. We call such
functions~\emph{polynomial root barriers},  
highlighting the fact that then $\sigma(\deg(p),\height(p))$
in~\Cref{definition:intro:root-barrier} is bounded by a polynomial in the bit
size of $p$. The
\Cref{theorem:result-root-barrier}(\ref{theorem:result-root-barrier:point2}) is obtained by instantiating the following \Cref{theorem:general-result-root-barrier}(\ref{theorem:general-result-root-barrier:point2}) to natural choices of~$\cn$.

\begin{restatable}{theorem}{TheoremGeneralResultRootBarrier}
  \label{theorem:general-result-root-barrier}
    Let $\cn > 0$ be a real number computable by a polynomial-time Turing machine, 
    and let~$\sigma(d,h) \coloneqq {c
    \cdot (d + \ceil{\ln h})^k}$ be a root barrier of~$\cn$, 
    for some $c,k \in \N_{\geq 1}$.
    \begin{enumerate}
      \item\label{theorem:general-result-root-barrier:point1} If $k = 1$, then
      the satisfiability problem for $\exists\R(\ipow{\cn})$ is in \twoexptime.
      \item\label{theorem:general-result-root-barrier:point2} If $k > 1$, then
      the satisfiability problem for $\exists\R(\ipow{\cn})$ is in
      \threeexptime.
    \end{enumerate}
\end{restatable}

\noindent
As we will see in~\Cref{sec:poly-evaluation}, whenever algebraic, the base~$\cn$ has a root barrier with exponent~${k =
1}$, and the related satisfiability problem for~$\exists(\ipow{\cn})$ thus lies in~\twoexptime.
However, a small trick will allow us to further improve this result to
\expspace,
establishing~\Cref{theorem:result-root-barrier}(\ref{theorem:result-root-barrier:point1}). 
When $\cn \in \N$, this result can be refined to the \nexptime bound given in~\Cref{theorem:result-root-barrier}(\ref{theorem:result-root-barrier:point0}).

\section{Preliminaries}
\label{sec:preliminaries}

In this section, we fix our notation, introduce background knowledge on
computable, algebraic and transcendental numbers, and define the
existential theory $\exists\R(\ipow{\cn})$.

\paragraph{Sets, vectors, and basic functions.}
Given a finite set $S$, we write $\abs{S}$ for its cardinality. Given $a,b \in
\R$, we write $[a,b]$ for the closed interval $\{ c \in \R : a \leq c \leq b
\}$. We use parenthesis \mbox{$($ and $)$} for open intervals, hence writing,
e.g., $[a,b)$ for the set $\{ c \in \R : a \leq c < b\}$. We write $[a..b]$ for
the set of integers $[a,b] \cap \Z$. Given $A \subseteq \R$, $c \in \R$, and a
binary relation $\sim$ (e.g., $\geq$), we define $A_{\sim c} \coloneqq {\{a \in
A : a \sim c\}}$. The \emph{endpoints} of $A$ are its supremum and infimum, if
they exist. For instance, the endpoints of the interval $[a,b)$ are the numbers
$a$ and $b$, while the endpoints of $[a..b]$ are the numbers $\ceil{a}$ and
$\floor{b}$, where $\floor{\cdot}$ stands for the floor function. 

Given a positive real number $b$ with $b \neq 1$, we write $\log_b(\cdot)$ for
the logarithm function of base $b$. We abbreviate $\log_2(\cdot)$ and
$\log_e(\cdot)$ as $\log(\cdot)$ and $\ln(\cdot)$, respectively. 

Unless stated explicitly, all integers encountered by our algorithms are encoded
in binary; note that $n \in \Z$ can be represented using $1+\ceil{\log(n+1)}$
bits. Similarly, each rational is encoded as a ratio  $\frac{n}{d}$ of two
coprime integers $n$ and $d$ encoded in binary, with $d \geq 1$.

\paragraph{Integer polynomials.}
An \emph{integer polynomial} in variables $\vec x = (x_1,\dots,x_n)$ is an
expression $p(\vec x) \coloneqq \sum_{j = 1}^m (a_j \cdot \prod_{i=1}^n
x_i^{d_{j,i}})$, where $a_j \in \Z$ and $d_{j,i} \in \N$ for every $j \in
[1..m]$ and $i \in [1..n]$. In the context of algorithms, we assume the
coefficients $a_j$ to be given in binary encoding, and the exponents $d_{j,i}$
to be given in unary encoding. We rely on the following notions:
\begin{itemize}
  \item The \emph{height} of $p$, denoted $\height(p)$, is defined as
  $\max\{|a_j| : j \in [1..m]\}$. 
  \item  The \emph{degree} of $p$, denoted $\deg(p)$, is defined as
  $\max\{\sum_{i=1}^n d_{j,i} : j \in [1..m]\}$.
  \item Given $i \in [1..n]$, the \emph{partial degree of $p$ in} $x_i$, denoted
  $\deg(x_i,p)$, is $\max\{ d_{j,i} : j \in
  [1..m] \}$. 
  \item The \emph{bit size} of $p$, denoted $\size(p)$, is defined as $m \cdot
  (\ceil{\log(\height(p)+1)} + n \cdot \deg(p))$.
\end{itemize}

\paragraph{Computable numbers, and algebraic and transcendental numbers.}
A real number $\cn \in \R$ is said to be \emph{computable} whenever there is a
(deterministic) Turing machine $T \colon \N \to \Q$ that given in input $n \in
\N$ written in unary (e.g., over the alphabet $\{1\}^*$) returns a rational
number~$T_n$ (represented as described above) such that $\abs{\cn-T_n} \leq
2^{-n}$. We thus have $\cn = \lim_{n \to \infty} T_n$, and for this reason $\cn$
is said to be \emph{computed by} $T$ (or $T$ \emph{computes} $\cn$).
The computable numbers form a field~\cite{Rice54}; 
we will later need the following two statements regarding
their closure under product and reciprocal.

\begin{restatable}{lemma}{LemmaTuringMachineProducts}
  \label{lemma:turing-machine-products}
  Given Turing machines $T$ and $T'$ computing reals $a$ and $b$,
  one can construct a Turing machine $T''$ computing $a \cdot b$.
  If $T$ and $T'$ run in polynomial time, then so does~$T''$.
\end{restatable}

\begin{proof}
  Let $\ell \coloneqq \ceil{\log(\abs{T_0}+\abs{T_0'}+3)}$.
  We define $T''$ as the Turing machine that on input $n$ returns the 
  rational number $T_{n+\ell} \cdot T_{n+\ell}'$. 
  Clearly, $T''$ runs in time polynomial in $n$.
  We show that $\abs{a \cdot b - T''} \leq \frac{1}{2^n}$ for every $n \in \N$, 
  i.e., $T''$ computes $a \cdot b$. 
  Let $\epsilon_1 \coloneqq \abs{T_{n+\ell}-a}$ and $\epsilon_2 \coloneqq \abs{T_{n+\ell}'-b}$. Recall that $\epsilon_1,\epsilon_2 \leq \frac{1}{2^{n+\ell}}$. 
  Then, 
  \begin{align*}
    \abs{a \cdot b - T''} &=
    \abs{a \cdot b - T_{n+\ell} \cdot T_{n+\ell}'}\\
    &\leq \abs{a} \cdot \epsilon_2 + \abs{b} \cdot \epsilon_1 + \epsilon_1 \cdot \epsilon_2\\
    &< \frac{\abs{a}+\abs{b}+1}{2^{n+\ceil{\log(\abs{T_0}+\abs{T_0'}+3)}}}
    &\tag*{since $\epsilon_1,\epsilon_2 \leq \frac{1}{2^{n+\ell}}$
    and def.~of~$\ell$}\\
    &\leq \frac{1}{2^{n}}\frac{\abs{a}+\abs{b}+1}{\abs{a}+\abs{b}+1}\leq \frac{1}{2^n}
    &\tag*{since $\abs{a} \leq \abs{T_0}+1$ and $\abs{b} \leq \abs{T_0'}+1$\qedhere}
    &&&\hfill
  \end{align*}
  \end{proof}

\vspace{-8pt}
\begin{restatable}{lemma}{LemmaTuringMachineReciprocal}
  \label{lemma:turing-machine-reciprocal}
  Given a Turing machine $T$ computing a non-zero real number $r$,
  one can construct a Turing machine $T'$ computing $\frac{1}{r}$. 
  If $T$ runs in polynomial time, then so does $T'$.
\end{restatable}

\begin{proof}
  Compute the smallest $k \geq 2$ such that
  $\frac{1}{2^k} < \abs{T_k}$; its existence follows from the fact
  that $\lim_{n \to \infty} T_n = r \neq 0$, whereas $\lim_{n \to \infty}
  \frac{1}{2^n} = 0$. Since $\abs{r - T_k} \leq \frac{1}{2^k}$, we have 
  that $T_k$ and $r$ have the same sign, and $0 < \abs{T_k}-\frac{1}{2^k} \leq \abs{r}$. 
  Let $T_k = \frac{p}{q}$, where $p \in \Z \setminus \{0\}$ and $q \geq 1$, 
  and let ${\ell \coloneqq 2(k + \ceil{\log(q)})}$.
  We first give a construction of $T'$ that depends on the~sign~of~$T_k$.

  \begin{description}
    \item[case: $T_k > 0$] 
      We define $T'$ as the Turing machine that on input $n$ returns the
      rational $\frac{1}{\max(\abs{T_{n+\ell}},T_k - 2^{-k})}$. Clearly, if $T$
      runs in time polynomial in $n$, so does $T'$. We
      prove that $T'$ computes~$\frac{1}{r}$. First, observe that
      $\abs{r-T_{n+\ell}} \leq \frac{1}{2^{n+\ell}}$ and $r > 0$ imply
      $\abs{r-\abs{T_{n+\ell}}} \leq \frac{1}{2^{n+\ell}}$. Then, because $0 <
      T_k - 2^{-k} \leq r$, we have $\abs{r-\max(\abs{T_{n+\ell}},T_{k}-2^{-k})}
      \leq \frac{1}{2^{n+\ell}}$.

      For
      every $n \in \N$,
      \begin{align*}
        \abs{\frac{1}{r} - T_n'}
        &= \abs{\frac{r - \max(\abs{T_{n+\ell}},T_{k}-2^{-k})}{r \cdot \max(\abs{T_{n+\ell}},T_{k}-2^{-k})}}
        \leq \frac{1}{2^{n+\ell}} 
        \cdot \frac{1}{r \cdot \max(\abs{T_{n+\ell}},T_{k}-2^{-k})}\\
        &\leq \frac{1}{2^{n+\ell} \cdot (T_k-2^{-k})^2} 
        \leq \frac{1}{2^{n+\ell+2\log(T_k - 2^{-k})}}
        &\hspace{-10cm}\tag*{$\text{since } 0 < T_k - 2^{-k} \leq r$}
      \end{align*}
      To conclude the proof it suffices to show $\ell + 2\log(T_k-2^{-k}) \geq 0$:
      \begin{align*}
        & \ell + 2\log(T_k-2^{-k})\\
        ={}& \ell + 2\log((2^{k}T_k-1)2^{-k})
        = \ell + 2\log\Big(\Big(\frac{2^{k}p-q}{q}\Big)2^{-k}\Big)\\
        ={}& \ell + 2\log(2^{k}p-q) -2\log(q) -2k\\
        \geq{}& \ell -2\log(q) -2k 
        & \tag*{since $2^kp-q$ is an integer,}\\
        && \hspace{-1.8cm}\tag*{from $\frac{1}{2^k} < T_k$ we get $\log(2^{k}p-q) \geq 0$}\\
        ={}& 2(k+\ceil{\log(q)}-\log(q)-k)\geq{} 0.
        &\tag*{by def.~of~$\ell$}\\
      \end{align*}
    \item[case: $T_k < 0$] 
      Since $\abs{r}$ is computed by the machine that on input $n$ returns $\abs{T_n}$, 
      by following the previous case of the proof
      we conclude that $\frac{1}{\abs{r}}$ 
      is computed by the Turing machine that on input $n$
      returns the positive rational $\frac{1}{\max(\abs{T_{n+\ell}},\abs{T_k} - 2^{-k})}$.
      Then, the Turing machine that on input $n$
      returns the negative rational $\frac{-1}{\max(\abs{T_{n+\ell}},\abs{T_k} - 2^{-k})}$
      computes $\frac{1}{r}$.
  \end{description}
  Putting the two cases together we conclude that $\frac{1}{r}$ 
  is computed by the Turing machine 
  that on input $n$ returns the non-zero rational number $\frac{s}{\max(\abs{T_{n+\ell}},\abs{T_k} - 2^{-k})}$, 
  where $s = +1$ if $T_k > 0$, and otherwise $s = -1$.
\end{proof}

A real number $\cn$ is \emph{algebraic} if it is a root of some univariate
non-zero integer polynomial. Otherwise, $\cn$ is \emph{transcendental}. We often
denote algebraic numbers by $\alg,\beta,\eta,\dots$\,. Throughout the paper, we
consider the following canonical representation: an algebraic number $\alg$ is
represented by a triple $(q,\ell,u)$ where $q$ is a non-zero integer polynomial
and $\ell,u$ are (representations of) rational numbers such that $\alg$ is the
only root of $q$ belonging to~$[\ell,u]$. 

\paragraph{The theory~$\exists\R(\ipow{\cn})$.}
Let $\cn > 0$ be a computable real number. We consider the structure $(\R; 0, 1,
\cn, +, \cdot, \ipow{\cn}, <,=)$ extending the signature of the FO theory of the
reals with the constant $\cn$ and the unary \emph{integer power} predicate
$\ipow{\cn}$ interpreted as $\{ \cn^{i} : i \in \Z \}$. Formulae from the
existential theory of this structure, denoted $\exists \R(\ipow{\cn})$, are
built from the grammar
\begin{align*} 
  \phi,\psi &\,\Coloneqq\,  p(\cn,\vec x) \sim 0  \,\mid\, \ipow{\cn}(x) \,\mid\, \top \,\mid\, \bot \,\mid\, \phi \lor \psi \,\mid\, \phi \land \psi \,\mid\, \exists x \, \phi \,,
\end{align*}
where $\sim$ belongs to $\{<,=\}$,  the argument $x$ of the predicate
$\ipow{\cn}(x)$ is a variable, and $p$ is an integer polynomial involving $\cn$
and variables $\vec{x}$. For convenience of notation, $\cn$ is in this context
seen as a variable of the polynomial $p$, so that we can rely on the previously
defined notions of height, degree and bit size. We remark that, then, $\height(p)$
is independent of $\cn$ whereas $\deg(p)$ depends on the integers occurring as
powers of $\cn$. The bit size of a formula $\phi$, denoted as $\size(\phi)$, is
the number of bits required to write down $\phi$ (where~$\cn$ is stored
symbolically, using a constant number of~symbols). Similarly, we write
$\deg(\phi)$ and $\height(\phi)$ for the maximum degree and height of
polynomials occurring in $\phi$, respectively.

The semantics of formulae from $\exists \R(\ipow{\cn})$ is standard; 
it is the one of the FO theory of the reals, plus a rule stating
that $\ipow{\cn}(x)$ is true whenever $x \in \R$ belongs to the set $\ipow{\cn}$.
The grammar above features disjunctions~($\lor$), conjunctions ($\land$), true
($\top$) and false~($\bot$), but it does not feature negation ($\lnot$) on top
of atomic formulae. 
This restriction is w.l.o.g.: $\lnot \ipow{\cn}(x)$ is
equivalent to the formula $x \leq 0 \lor \exists y : \ipow{\cn}(y) \land y < x
\land x < \cn \cdot y$ stating that $x$ is either non-positive or strictly
between two successive integer powers of~$\cn$, whereas $\lnot (p(\cn,\vec x) <
0)$ and $\lnot (p(\cn,\vec x) = 0)$ are equivalent to $p(\cn, \vec x) = 0 \lor
-p(\cn,\vec x) < 0$, and $p(\cn,\vec x) < 0 \lor -p(\cn,\vec x) < 0$,
respectively. We still sometimes write negations in formulae, but these
occurrences should be seen as shortcuts. The grammar also avoids polynomials in
the scope of~$\ipow{\cn}(\cdot)$, since
$\ipow{\cn}(p(\cn,\vec{x}))$ is equivalent to $\exists y : y = p(\cn,\vec{x})
\land \ipow{\cn}(y)$. We write $\phi \models \psi$ whenever $\phi$
\emph{entails} $\psi$.

For simplicity of exposition, we sometimes consider 
formulae from the full first-order theory $\R(\ipow{\cn})$, rather than just
$\exists\R(\ipow{\cn})$. The grammar of $\R(\ipow{\cn})$ is obtained from 
that of $\exists\R(\ipow{\cn})$ by allowing arbitrary negations. 
This is convenient, for instance, when dealing with implications 
between formulae of $\exists\R(\ipow{\cn})$: 
given two such formulae $\phi$ and $\psi$, 
the formula $\phi \implies \psi$ belongs to $\R(\ipow{\cn})$ 
(and, after a simple rewriting, to $\exists\R(\ipow{\cn})$ if $\phi$ is quantifier-free). 
We also use $\iff$ to denote the double implication.

\paragraph{Refined representation for fixed algebraic numbers}Since, in our case, we
are fixing the base~$\cn$ of $\exists\R(\ipow{\cn})$, sometimes it is convenient to
improve how \emph{fixed} algebraic numbers (i.e.,~those that
do not depend on the input of our procedures) are represented. 
In these cases, we impose the following
restriction on $\ell$ and~$u$: either $\ell = u$, or $\alpha \in
(\ell,u)$ and $(\ell,u) \cap \Z = \emptyset$ (note: this is in addition to the
property that $\alg$ is the only root of $q$ in $[\ell,u]$). This restriction is
without loss of generality. Indeed, given a triple $(q,\ell,u)$ not satisfying
it, we can apply dichotomy search to refine the interval $[\ell,u]$ to an
interval $[\ell',u']$ such that $u'-\ell' < 1$. This refinement is done by tests
of the form $\exists x : q(x) = 0 \land \ell < x \leq \frac{u-\ell}{2}$, which
can be performed (in fact, in polynomial time) by, e.g., \cite[Theorem 1.3.1]{BasuPR96},
whose statement we recall in \autoref{theorem:basu}.
After computing $[\ell',u']$, we reason as follows: 
\begin{itemize}
  \item if $[\ell',u']$ does not contain an integer, $(q,\ell',u')$ is the
  required representation of $\alpha$.
  \item if $[\ell',u']$ contains $k \in \Z$ and $q(k) = 0$, then $(q,k,k)$ the
  required representation of~$\alpha$.
  \item if $[\ell',u']$ contains $k \in \Z$ and $q(k) \neq 0$, then one among
  $(q,\ell,k)$ and $(q,k,u)$ is the required representation of~$\alpha$. It then
  suffices to check where $\alpha$ lies, which can be done by testing $\exists x
  : q(x) = 0 \land \ell < x \leq k$, again with, e.g., the algorithm
  in~\cite[Theorem 1.3.1]{BasuPR96}.
\end{itemize}
Once more, we stress the fact that this representation is only used for algebraic numbers that are \emph{fixed}, and so the above refinement of $\ell$ and $u$ takes constant time.


\section{An algorithm for deciding \texorpdfstring{$\exists\R(\ipow{\cn})$}{the existential theory}}
\label{sec:the-algorithm}

\emph{Fix a computable number $\cn > 0$ that is either transcendental or has a
polynomial root barrier}. In this section, we discuss our procedure for deciding
the satisfiability of formulae in~$\exists\R(\ipow{\cn})$. For simplicity, we
assume for now $\cn > 1$. The general case of $\cn > 0$ is handled
in~\Cref{subsection:small-bases}.

The pseudocode of the procedure is given in~\Cref{algo:main-procedure}. To keep
it as simple as possible, we use nondeterminism in line~\ref{algo:line8} instead
of implementing, e.g., a routine backtracking algorithm. The procedure assumes
the input formula~$\phi(x_1,\dots,x_n)$ to be quantifier-free (this is without
loss of generality, since~$\exists\R(\ipow{\cn})$ is an existential theory), 
and it is split into three steps, which we discuss in the forthcoming three
subsections.

\begin{algorithm}[t]
  \caption{A procedure deciding the satisfiability problem for $\exists\R(\ipow{\cn})$.}\label{algo:main-procedure}
  \begin{algorithmic}[1]
    \Fixed $\cn > 1$ computable number that is transcendental or has a polynomial root barrier.
    \Require $\phi(x_1,\dots,x_n)$ : quantifier-free formula from
    $\exists\R(\ipow{\cn})$.
    \Ensure True $(\top)$ if $\phi$ is satisfiable, and otherwise False
    $(\bot)$.
    \medskip
    \For{$i \in [1..n]$}
    \label{algo:line1}
    \Let $u_i$ and $v_i$ be two fresh variables
    \label{algo:line2}
    \EndLet
    \State update $\phi$: replace every occurrence of $\ipow{\cn}(x_i)$ with $v_i=1$
    \label{algo:line3}
    \State update $\phi$: replace every occurrence of $x_i$ with $u_i\cdot v_i$
    \label{algo:line4}
    \State $\phi \gets \phi \land (v_i = 0 \lor 1 \leq \abs{v_i} < \cn)$
    \label{algo:line5}
    \EndFor
    \State $\psi(u_1,\dots,u_n)\gets{}$ \textsc{RealQE}$(\,\exists v_1 \dots
      \exists v_n : \phi\,)$
    \Comment{eliminate $v_1,\dots,v_n$ (see~\autoref{theorem:basu})}
    \label{algo:line6}
    \For{$i \in [1..n]$}
    \label{algo:line7}
    \Comment{$g_i$ below is encoded in unary}
    \State \myguess $g_i \gets{}$an element of $P_{\psi}$  
    \Comment{$P_{\psi} \subseteq \Z$ is the set from~\Cref{theorem:small-model-property}}
    \label{algo:line8}
    \EndFor
    \State \myreturn evaluate whether the assignment $(u_1 = \cn^{g_1}, \dots, u_n = \cn^{g_n})$ is a solution to $\psi$
    \label{algo:line9}
  \end{algorithmic}
\end{algorithm}

\subsection{Step I (lines~\ref{algo:line1}--\ref{algo:line6}): reducing the
variables to integer powers of
\texorpdfstring{$\xi$}{the base}}\label{subsection:reduction-to-substructure}

The first step reduces the problem of finding a solution over $\R$ to the
problem of finding a solution over $\ipow{\cn}$. Below, we denote by
$\exists\ipow{\cn}$ the existential theory of the structure $(\ipow{\cn}; 0, 1,
\cn, +, \cdot, <,=)$. Formulae from this theory are built from the grammar of
$\exists\R(\ipow{\cn})$, except they do not feature predicates $\ipow{\cn}(x)$,
as they are now trivially true.

For reducing $\exists\R(\ipow{\cn})$ to $\exists\ipow{\cn}$, we observe that
every $x \in \R$ can be factored as $u \cdot v$ where $u$ belongs to
$\ipow{\cn}$ and $v$ is either $0$ (if $x = 0$) or it belongs, in absolute
value, to the interval~$[1,\cn)$. In the case of $x \neq 0$, this factorisation
is unique, and $u$ corresponds to the largest element of $\ipow{\cn}$ that is
less or equal to the absolute value of~$x$, i.e., $u \leq \abs{x} < \cn \cdot
u$. The procedure uses this fact to replace every occurrence of a variable $x_i$
in the input formula $\phi(x_1,\dots,x_n)$ with two fresh variables $u_i$ and
$v_i$ (see the \textbf{for} loop of line~\ref{algo:line1}), where $v_i$ is set
to satisfy either $v_i = 0$ or $1 \leq \abs{v_i} < \cn$ (the latter is short for
the formula $(1 \leq v_i < \cn) \lor (-\cn < v_i \leq -1)$), and $u_i$ is
(implicitly) assumed to belong to $\ipow{\cn}$. This allows to replace all
occurrences of the predicate $\ipow{\cn}(x_i)$ with $v_i = 1$
(line~\ref{algo:line3}). We obtain in this way an equivalent formula from the
existential theory of the reals, but where the variables $u_1,\dots,u_n$ are
assumed to range over~$\ipow{\cn}$.

After the updates performed by the \textbf{for} loop, the procedure eliminates
the variables $v_1,\dots,v_n$ by appealing to a quantifier elimination procedure
for the FO theory of the reals, named \textsc{RealQE} in the pseudocode. We
remind the reader that a quantifier elimination procedure is an algorithm that, 
from an input (quantified) formula, produces an \emph{equivalent} quantifier-free
formula. Since such a procedure preserves formula equivalence, we can use it to
eliminate $v_1,\dots,v_n$ even if $u_1,\dots,u_n$ are assumed to range over
$\ipow{\cn}$.
The constant $\cn$ appearing in the formula is treated as an
additional free variable by \textsc{RealQE}. The output formula
$\psi(u_1,\dots,u_n)$ belongs to $\exists\ipow{\cn}$, as required. This
concludes the first step of the~algorithm.

To perform the quantifier elimination step, we rely on the
quantifier elimination procedure for the (full) FO theory of the reals developed
by Basu, Pollack and Roy~\cite{BasuPR96}. This procedure achieves the
theoretically best-known bounds for the output formula, both for 
arbitrary
quantifier alternation and for the existential
fragment (i.e., setting~$\omega = 1$ below).

\begin{restatable}[{\cite[Theorem 1.3.1]{BasuPR96}}]{thmC}{TheoremBasu} 
  \label{theorem:basu}
  There is an algorithm with the specification:
  \begin{center}
    \begin{minipage}{0.95\linewidth}
      \begin{description}
        \item[\textup{Input}] A formula~$\phi(\vec y)$ from the first-order
          theory of $(\R; 0, 1, +, \cdot, <, =)$.
        \item[\textup{Output}] A quantifier-free formula $\gamma(\vec y) =
          \bigvee_{i = 1}^I \bigwedge_{j = 1}^J p_{i,j}(\vec y) \sim_{i,j} 0$
          equivalent to $\phi$,\\ 
          where every $\sim_{i,j}$ is from $\{<,=\}$.
      \end{description}
    \end{minipage}
  \end{center}
  Suppose the input formula~$\phi$ to be of the form $Q_1 \vec x_1 \in \R^{n_1}
  \dots Q_\omega \vec x_\omega \in \R^{n_\omega} : \psi(\vec y, \vec x_1, \dots,
  \vec x_\omega)$, where $\vec y = (y_1,\dots,y_k)$, every $Q_i$ is $\exists$ or
  $\forall$, and $\psi$ is a quantifier-free formula with $m$ atomic formulae
  $g_i \sim 0$ satisfying $\deg(g_i) \leq d$ and $\height(g_i) \leq h$. Then,
  the output formula~$\gamma$ satisfies
  \begin{center}
      $\begin{aligned} I &\leq (m \cdot d + 1 )^{(k+1) \Pi_{i=1}^\omega
        O(n_i)}\,, &\hspace{1cm} \deg(p_{i,j}) &\leq d^{\Pi_{i=1}^\omega
        O(n_i)}\,, \\
        J &\leq (m \cdot d + 1)^{\Pi_{i=1}^\omega O(n_i)}\,, & \height(p_{i,j})
        &\leq (h+1)^{d^{(k+1) \Pi_{i=1}^\omega O(n_i)}}\,, \end{aligned}$
        \vspace{-3pt}
  \end{center}
  and the algorithm runs in time $\size(\phi)^{O(1)}(m \cdot d
  +1)^{(k+1)\Pi_{i=1}^\omega O(n_i)}$.
\end{restatable}

\subsection{Step II (lines~\ref{algo:line7} and~\ref{algo:line8}):
solving \texorpdfstring{$\exists\ipow{\cn}$}{the existential theory over integer powers of the base}}\label{subsection:algorithm-step-two}

The second step of the procedure searches for a solution to the quantifier-free
formula $\psi$ in line~\ref{algo:line6}. For every variable $u_i$ in $\psi$, the
algorithm guesses an integer $g_i$, encoded in unary, from a finite set $P_{\psi}$.
Implicitly, this guess is setting $u_i = \cn^{g_i}$. The next proposition shows
that $P_{\psi}$ can be computed from $\psi$ and the base~$\cn$, i.e.,
$\exists\ipow{\cn}$ has a \emph{small witness property}.%

\begin{restatable}{proposition}{TheoremSmallModelProperty}
  \label{theorem:small-model-property}
  Fix $\cn > 1$. There is an algorithm with the following specification: 
  \begin{center}
    \begin{minipage}{0.95\linewidth}
      \begin{tabular}{rp{0.8\linewidth}}
        \textbf{Input:}& A quantifier-free formula $\psi(u_1,\dots,u_n)$ from $\exists \ipow{\cn}$.\\ 
        \textbf{Output:}& A finite set $P_{\psi} \subseteq \Z$ such that $\psi$ is satisfiable if and only if
        $\psi$ has a solution in the set $\{(\cn^{j_1},\dots,\cn^{j_n}) : j_1, \dots, j_n \in P_{\psi} \}$.
      \end{tabular}
    \end{minipage}
  \end{center}
  To be effective, the algorithm requires knowing either that $\cn$ is a
  computable transcendental number, or two integers  $c,k \in \N_{\geq 1}$ for
  which~$\sigma(d,h) \coloneqq c \cdot (d+\ceil{\ln(h)})^k$ is a root barrier of
  $\cn$. In the latter case, the elements in $P_{\psi}$ are bounded in absolute
  value by $(2^{c}\ceil{\ln(H)})^{D^{2^5n^2} k^{D^{8n}}}$, where $H \coloneqq
  \max(8,\height(\psi))$ and $D \coloneqq \deg(\psi)+2$.
\end{restatable}

\noindent
We defer the proof of~\Cref{theorem:small-model-property} 
to~\Cref{sec:solving-substructure}. The bound on the elements in~$P_{\psi}$ 
follows from the quantifier elimination procedure we develop in~\Cref{subsection:quantifier-elimination}: 
eliminating a single variable yields a formula of size singly exponential in the input when~$k = 1$, and doubly exponential when $k > 1$. Eliminating several variables adds one further exponential when $k > 1$. 
The dichotomy in~\Cref{theorem:general-result-root-barrier} stem from this distinction.

\subsection{Step III (line~\ref{algo:line9}): polynomial sign
evaluation}\label{subsection:polynomial-sign-evaluation} 
The last step of the procedure checks if the assignment $u_1 =
\cn^{g_1},\dots,u_n = \cn^{g_n}$ is a solution to $\psi(u_1,\dots,u_n)$. Observe
that $\psi(\cn^{g_1},\dots,\cn^{g_n})$ is a Boolean combination of polynomial
(in)equalities $p(\cn) \sim 0$, where $\cn$ may occur with negative powers (as
some $g_i$ may be negative). This is unproblematic, as one can make all powers
non-negative by rewriting each (in)equality $p(\cn) \sim 0$ as $\cn^{-d} \cdot p
\sim 0$, where $d$ is the smallest negative integer occurring as a power of
$\cn$ in~$p$ (or $0$ if such an integer does not exist).
After this small update, line~\ref{algo:line9} boils down to determining the
sign that each polynomial in the formula has when evaluated at $\cn$. This
enables us to simplify all inequalities to either $\top$ or $\bot$, to then
return $\top$ or $\bot$ depending on the Boolean structure of~$\psi$. Let us
thus focus on the required sign evaluation problem, which we denote
by~$\SIGN_{\cn}$. Its specification is the following:

\begin{center}
  \begin{minipage}{0.95\linewidth}
      \begin{tabular}{rp{0.8\linewidth}}
      \textbf{Input:}& A univariate integer polynomial $p(x)$.\\
      \textbf{Output:}& The symbol $\sim$ from $\{<,>,=\}$
      such that $p(\cn) \sim 0$.
      \end{tabular}
  \end{minipage}
\end{center}

\paragraph*{Solving {\rm$\SIGN_{\cn}$} when $\cn \in \R$ is transcendental.}
It is a standard fact that~$\SIGN_{\cn}$ becomes solvable whenever~$\cn$ is any
computable transcendental number. Indeed, in this case $p(\cn)$ must be
different from~$0$, and one can rely on the fast-convergence sequence of
rational numbers~$T_0,T_1,\dots$ to find $n \in \N$ such that $|p(\cn) -
p(T_{n})|$ is guaranteed to be less than $|p(T_{n})|$. The sign of $p(\cn)$ then
agrees with the sign of~$p(T_{n})$, and the latter can be easily computed. In
general, the asymptotic running time of this algorithm cannot be bounded.

\paragraph*{Solving {\rm$\SIGN_{\cn}$} when $\cn \in \R$ has a (polynomial) root barrier.}
A similar algorithm as the one given for transcendental numbers can be defined
for numbers with a polynomial root barrier; and in this case its running time
can be properly analysed. The pseudocode of such a procedure is given
in~\Cref{algo:sign-evaluation}, and it should be self-explanatory. We stress
that running this algorithm requires access to the root barrier $\sigma$ and the
Turing machine $T$.

To prove that \Cref{algo:sign-evaluation} respects its specification 
we need the following lemma.

\begin{algorithm}[t]
  \caption{Algorithm for solving $\SIGN_{\cn}$ when $\cn$ has a root barrier.}
  \label{algo:sign-evaluation}
  \begin{algorithmic}[1]
    \Fixed A number $\cn \in \R$ computed by a Turing machine $T$ and having a root barrier $\sigma$.
    \Require A univariate integer polynomial $p(x)$ of degree $d$ and height $h$.
    \Ensure The symbol $\sim$ from $\{<,>,=\}$ such that $p(\cn) \sim 0$.
    \medskip
    \State $n \gets 1 + 2 \sigma(d,h) + 3 d\ceil{\log(h+4)}$ 
      \label{algo:sign-evaluation:bound-on-n}
    \If{$\abs{p(T_n)} \leq 2^{-2 \sigma(d,h)-1}$ and $\abs{T_n} < h+2$}
      \textbf{return} the symbol $=$
      \label{algo:sign-evaluation:small-pTn}
    \Else \
      \textbf{return} the sign of $p(T_n)$
      \label{algo:sign-evaluation:large-pTn}
    \EndIf 
  \end{algorithmic}
\end{algorithm}%


\begin{lemma}\label{lemma:approx-univ-polynomial}
  Let $p(x)$ be an integer polynomial, 
and let ${r \in \R}$ with $\abs{r} \leq K$ for some $K \geq 1$.
Consider $L,M \in \N$ satisfying 
$M\geq L + \log(\height(p)+1) + 2 \deg(p) \cdot \log(K+1)$.
For every $r^* \in \R$, if $\abs{r-r^*} \leq 2^{-M}$, 
then $\abs{p(r)-p(r^*)} \leq 2^{-L}$.
\end{lemma}
\begin{proof}
  Let $p(x) \coloneqq \sum_{j=0}^d a_i \cdot x^j$, and
  suppose $\abs{r-r^*} \leq 2^{-M}$.
  If $d = 0$, then $p$ is a constant polynomial and $\abs{p(r)-p(r^*)} = 0$, 
  which proves the lemma. 
  Below, we assume $d \geq 1$.

  To show that $\abs{p(r)-p(r^*)} \leq 2^{-L}$, let us start by bounding the maximum of the absolute value
  that the first derivative $p'(x) = \sum_{j=1}^d a_j \cdot j \cdot x^{j-1}$ of
  $p$ takes in the interval $I \coloneqq [-(K+1),K+1]$. For every $x \in \R$, $\abs{p'(x)} \leq
  g(x) \coloneqq \sum_{j=1}^d \abs{a_j \cdot j \cdot x^{j-1}}$. Since the
  function $g$ is monotone over $\R_{\geq 0}$, and $g(y) = g(-y)$ for every $y
  \in \R$, we conclude that for every $x \in I$, $\abs{p'(x)} \leq g(K+1) =
  \sum_{j=1}^d \abs{a_j} \cdot j \cdot (K+1)^{j-1} \leq d^2 \height(p) (K+1)^{d-1}$. 
  
  From $\abs{r} \leq K$ and $\abs{r-r^*} \leq 2^{-M}$, where $M \geq 0$, 
  we have that both $r$ and $r^*$ belong to~$I$.
  This implies $\frac{\abs{p(r) -
  p(r^*)}}{\abs{r-r^*}} \leq \max\{\abs{p'(x)} : x \in I \} \leq
  d^2 \height(p) (K+1)^{d-1}$. So,
  {\allowdisplaybreaks
  \begin{align*}
    & \abs{p(r) - p(r^*)}\\
    \leq{}& d^2 \height(p) (K+1)^{d-1} \abs{r-r^*}\\
    \leq{}& 2^{2 \log(d)+\log(\height(p))+(d-1)\log(K+1)-(L + \log(\height(p)+1) + 2 d \cdot \log(K+1))}
    & \tag*{\text{bound on $M$}}\\
    \leq{}& 2^{2 \log(d)-L-(d+1) \cdot \log(K+1)} \leq 2^{-L}
    &\hspace{-4.5cm}\tag*{\text{since }$2 \log(d) \leq d$ and $\log(K+1) \geq 1$\qedhere}
  \end{align*}}
\end{proof}

\begin{restatable}{lemma}{LemmaSignPolyRootBarrier}
  \label{lemma:sign-poly-root-barrier}
  \Cref{algo:sign-evaluation} respects its specification.
\end{restatable}

\begin{proof}
  Let $p(x) = \sum_{j=0}^d a_j \cdot x^j$ be input integer polynomial, having
  degree $d = \deg(p) \geq 1$ and height $h = \height(p)$. Recall that, from the
  definition of root barrier, whenever $p(\cn) \neq 0$ we have $\abs{p(\cn)}
  \geq e^{-\sigma(d,h)} > 2^{-2 \sigma(d,h)}$, where the last inequality
  follows from $\sigma(d,h) \geq 0$. Following
  line~\ref{algo:sign-evaluation:bound-on-n}, define $n \coloneqq 1 +
  2 \sigma(d,h) + 3d\ceil{\log(h+4)}$. 
  Note that $n \geq 5$.

  Let us first assume that $\abs{T_n} \geq h + 2$. In this case, the algorithm
  returns the sign of $p(T_n)$ (line~\ref{algo:sign-evaluation:large-pTn}). We
  show that $p(T_n)$ and $p(\cn)$ have the same sign. Since $\abs{\cn-T_n}
  \leq 2^{-1}$, we have $\abs{\cn} > h+1$.  
  By a result of Cauchy~\cite[Chapter 8]{Rahman02}, $h+1$ is an upper bound to
  the absolute value of every root of $p$. This implies that there are no roots
  of $p$ in the interval $[\cn,T_n]$, so in particular $\cn$ and $T_n$ are
  not roots of $p$, and $p(\cn)$ and $p(T_n)$ have the same sign.
  
  Let us consider now the case $\abs{T_n} < h + 2$, and so $\abs{\cn} \leq K
  \coloneqq h + 3$. 
  From~\Cref{lemma:approx-univ-polynomial}, the definition of $n$ and the fact that $\abs{\cn-T_n}
  \leq 2^{-n}$, we conclude that
  $\abs{p(\cn) - p(T_n)} \leq 2^{-2 \sigma(d,h)-1}$. This implies that if
  $\abs{p(T_n)} \leq 2^{-2 \sigma(d,h)-1}$ then $p(\cn) = 0$, and
  otherwise $p(T_n)$ and $p(\cn)$ have the same sign;  
  which concludes the proof of the lemma (see
  lines~\ref{algo:sign-evaluation:small-pTn}
  and~\ref{algo:sign-evaluation:large-pTn}). Indeed,  
  \begin{itemize}
    \item If $p(\cn) = 0$ then $\abs{p(T_n)} \leq 2^{-2 \sigma(d,h)-1}$
    (from $\abs{p(\cn) - p(T_n)} \leq 2^{-2 \sigma(d,h)-1}$).
    \item If $p(\cn) \neq 0$, then:
    \begin{align*}
      \abs{p(T_n)} 
      & \geq \abs{p(\cn)} - \abs{p(\cn) - p(T_n)}
      & \tag*{from properties of the absolute value}\\
      & > 2^{-2 \sigma(d,h)} - 2^{-2 \sigma(d,h)-1}
      & \tag*{bounds on $\abs{p(\cn)}$ and $\abs{p(\cn) - p(T_n)}$}\\
      & = 2^{-2 \sigma(d,h)-1}.
    \end{align*}
    Moreover, $\abs{p(\cn) - p(T_n)} \leq 2^{-2 \sigma(d,h)-1}$ and
    $\abs{p(T_n)} > 2^{-2 \sigma(d,h)-1}$ imply that $p(\cn) > 0$ if and
    only if $p(T_n) > 0$.
    \qedhere
  \end{itemize}
\end{proof}

\noindent 
When $\sigma$ is a polynomial root barrier, the integer $n$ from
line~\ref{algo:sign-evaluation:bound-on-n} can be written in unary using
polynomially many digits with respect to~$\size(p)$. This yields the following
lemma.

\begin{restatable}{lemma}{LemmaRuntimeSignPolyRootBarrier}
  \label{lemma:runtime-sign-poly-root-barrier}
  Let $\cn \in \R$ be a number computed by a Turing machine $T$ and having a polynomial
  root barrier~$\sigma$. If $T$ runs in polynomial time, then so
  does~\Cref{algo:sign-evaluation}.
\end{restatable}

\begin{proof}
  When encoded in unary, the number $n$ defined in
  line~\ref{algo:sign-evaluation:bound-on-n} has size polynomial in the size of
  the input polynomial $p$. Then, to compute $T_n$ only requires polynomial time
  in $\size(p)$. Observe that this implies $T_n = \frac{q}{d}$ for some integers
  $q$ and $d$ encoded in binary using polynomially many bits with respect to $\size(p)$. Evaluating a
  polynomial at such a rational point can be done in polynomial time in the size of
  the polynomial and of the bit size of the rational. 
  This means that also lines~\ref{algo:sign-evaluation:small-pTn}
  and~\ref{algo:sign-evaluation:large-pTn} run in polynomial time in $\size(p)$.
\end{proof}

\subsection{Handling small bases}\label{subsection:small-bases}

We now extend our algorithm so that it works assuming $\cn > 0$ instead of just
$\cn > 1$. Let $\cn$ be computable and either
transcendental or with a polynomial root barrier. First, observe that we can
call the procedure for $\SIGN_\cn$ on input $x-1$ in order to check if $\cn \in
(0,1)$, $\cn = 1$ or $\cn > 1$. 

If $\cn = 1$, we
replace in the input formula $\phi$ every occurrence of $\ipow{\cn}(x)$ with $x
= 1$, obtaining a formula from the existential theory of the reals, which we can
solve by~\autoref{theorem:basu}. If~$\cn > 1$, we
call~\Cref{algo:main-procedure}. Suppose then $\cn \in (0,1)$. In this case, we
replace every occurrence of $\ipow{\cn}(x)$ with
$\ipow{\big(\frac{1}{\cn}\big)}(x)$, and opportunely multiply by integer powers of
$\frac{1}{\cn}$ both sides of polynomials inequalities in order to eliminate the
constant $\cn$. In this way, we obtain from $\phi$ an equivalent formula
in~$\exists\R(\ipow{\big(\frac{1}{\cn}\big)})$. Since $\frac{1}{\cn} > 1$, we
can now call~\Cref{algo:main-procedure}; provided we first establish the
properties of $\frac{1}{\cn}$ required to run this algorithm. These properties
indeed hold:
\begin{enumerate}
  \item\label{small-bases:i2} If $\cn$ is transcendental, then so is
  $\frac{1}{\cn}$. This is because the algebraic numbers form a field.
  \item\label{small-bases:i3} If $\cn$ has a polynomial root barrier~$\sigma$,
  then $\sigma$ is also a root barrier of $\frac{1}{\cn}$. Indeed, consider an
  integer polynomial $p(x) = \sum_{i=0}^d a_i \cdot x^i$ with height $h$, and
  assume $p(\frac{1}{\cn}) \neq 0$. Since $\sigma$ is a root barrier of $\cn$,
  we have $\cn^d \cdot |{p(\frac{1}{\cn})}| = |{\sum_{i=0}^d a_i \cdot
  \cn^{d-i}}| \geq e^{-\sigma(h,d)}$, which in turn implies that
  $|{p(\frac{1}{\cn})}| \geq e^{-\sigma(h,d)} \cdot \cn^{-d} \geq
  e^{-\sigma(h,d)}$, where the last inequality uses~$\frac{1}{\cn} \geq 1$.
  \item\label{small-bases:i1} From a Turing machine $T$ computing $\cn$, we can
  construct a Turing machine $T'$ computing~$\frac{1}{\cn}$.
  \Cref{lemma:turing-machine-reciprocal} gives this construction, and shows that
  $T'$ runs in polynomial time if $T$ does. 
\end{enumerate}

\subsection{Correctness of Algorithm~\ref{algo:main-procedure} 
(Proof of \Cref{theorem:result-root-barrier}(\ref{theorem:result-root-barrier:point3}))}
\label{subsection:correctness-algorithm}

Since~lines~\ref{algo:line1}--\ref{algo:line5} preserve the satisfiability the
input formula, by chaining~\autoref{theorem:basu},
\Cref{theorem:small-model-property}, and \Cref{lemma:sign-poly-root-barrier}, we
conclude that~\Cref{algo:main-procedure} is correct. 
\begin{restatable}{lemma}{LemmaCorrectnessAlgorithmOne}
  \label{lemma:correctness-algorithm-1}
  \Cref{algo:main-procedure} respects its specification.
\end{restatable}

\begin{proof}
  Consider an input
  formula $\phi(x_1,\dots,x_n)$, and let $\phi'(u_1,\dots,u_n,v_1,\dots,v_n)$ be
  the formula obtained from it at the completion of the \textbf{for} loop of
  line~\ref{algo:line1}. Note that if $\phi$ and $\phi'$ are equisatisfiable,
  then the lemma follows. Indeed:
  \begin{itemize}
    \item By~\autoref{theorem:basu}, the formula $\psi(u_1,\dots,u_n)$ in
    line~\ref{algo:line6} is equisatisfiable with $\phi'$,
    \item By~\Cref{theorem:small-model-property}, $\psi$ is satisfiable if and
    only if it has a solution from the set $S = \{(\cn^{j_1},\dots,\cn^{j_n}) :
    j_1,\dots,j_n \in P_\psi\}$, where $P_\psi$ is the set from~\Cref{theorem:small-model-property}. 
    Lines~\ref{algo:line7} and~\ref{algo:line8}, 
    search for such an element of $S$.
    \item Following line~\ref{algo:line9}, the algorithm returns $\top$ if and only if $\psi$ evaluates to true
    on a point from the set $S$.
    For this evaluation step, one consider all polynomials inequalities ${p(\cn,\cn^{g_1},\dots,\cn^{g_n}) \sim 0}$ in $\psi(\cn^{g_1},\dots,\cn^{g_n})$, and evaluate its sign using the algorithm for $\SIGN_{\cn}$. As a result of this operation, $\psi(\cn^{g_1},\dots,\cn^{g_n})$ is updated into a Boolean combination of $\top$ and $\bot$, reduces to just $\top$ or $\bot$ after all Boolean connectives are evaluated.
  \end{itemize}
  So, to conclude the proof we just have to formally prove that $\phi$ and $\phi'$ are equisatisfiable.
  Recall that for every real number $r \in \R$ there is a pair of numbers $(u,v)$ 
  such that $x = u \cdot v$, $u \in \ipow{\cn}$ and either $v = 0$ or $1 \leq \abs{v} < \cn$.
  If $r \neq 0$, the pair $(u,v)$ is unique.
  Then, the formula $\phi$ is equisatisfiable with 
  \begin{equation}
    \label{eq:phi-for-line-1} 
    \phi\sub{u_i \cdot v_i}{x_i : i \in [1..n]} \land \bigwedge_{i=1}^n (\ipow{\cn}(u_i) \land (v_i = 0 \lor 1 \leq \abs{v} < \cn))
  \end{equation}
  where $u_1,\dots,u_n,v_1,\dots,v_n$ are fresh variables. 
  The formula $\phi\sub{u_i \cdot v_i}{x_i : i \in [1..n]}$ features atomic formulae $\ipow{\cn}(u_i \cdot v_i)$. Under the assumption that $\ipow{\cn}(u_i) \land (v_i = 0 \lor 1 \leq \abs{v} < \cn)$ holds, note that $\ipow{\cn}(u_i \cdot v_i)$ is equivalent to $v_i = 1$.
  Then, we can replace in Formula~\eqref{eq:phi-for-line-1} every occurrence of $\ipow{\cn}(u_i \cdot v_i)$ with $v_i = 1$, preserving equivalence. 
  The formula we obtain is exactly the formula $\phi'$, which is thus equisatisfiable with $\phi$.
\end{proof}

\subsection{Running time (Proof of~\Cref{theorem:general-result-root-barrier})}\label{subsection:proof-general-root-barrier}

From the discussion in~\Cref{subsection:small-bases}, it suffices to restrict to instances of 
the problem where $\cn > 1$.
Below, we analyse the complexity of Algorithm \ref{algo:main-procedure}, considering the three steps separately. 

Let $\phi(x_1,\dots,x_n)$ be a formula from $\exists\R(\ipow{\cn})$, having $m_1$ occurrences of polynomial (in)equalities $g \sim 0$, all with $\deg(g) \leq d$ and $\height(g) \leq h$, and $m_2$ occurrences of the predicate~$\ipow{\cn}$. 
We run~\Cref{algo:main-procedure} with $\phi$ as input.

\paragraph{\textit{Step~I \textup{(}runtime: exponential in $\size(\phi)$\textup{)}}}
Lines~\ref{algo:line1}--\ref{algo:line5} update $\phi$ by (i) 
replacing the occurrences of $\ipow{\cn}(x_i)$ with $v_i = 1$, 
(ii)~replacing the occurrences of $x_i$ with $u_i \cdot v_i$ and 
(iii)~adding constraints $v_i = 0$ and $1 \leq \abs{v_i} < \cn$. Let $\phi'$ be the formula obtained after these updates. The size of $\phi'$ is polynomial in $\size(\phi)$. Moreover, $\phi'$ has:
\begin{enumerate}
  \item at most $2n$ variables,
  \item at most $m_1+m_2+5n$ polynomial (in)equalities (recall that $1 \leq \abs{v_i} < \cn$ is a shortcut for the formula $-\cn < v_i \leq 1 \lor 1 \leq v_i < \cn$),
  \item and all its polynomials (in)equalities $g \sim 0$ are such that $\deg(g) \leq 2d$ and $\height(g) \leq h$. The increase in the degree is due to the replacements of variables $x_i$ with $u_i \cdot v_i$.
\end{enumerate}
The procedure then eliminates the variables $v_1,\dots,v_n$ by
calling~\textsc{RealQE} (line~\ref{algo:line6}).
Following~\autoref{theorem:basu}, the runtime of~\textsc{RealQE} is exponential
in $\size(\phi)$, and therefore $\psi$ has size exponential in
$\size(\phi)$. More precisely $\psi$ has 
\begin{enumerate}
  \setcounter{enumi}{3}
  \item at most $n$ variables,
  \item\label{rtalgo:it5} at most $((m_1+m_2+5n) \cdot 2 \cdot d + 1)^{O(n^2)}$ polynomial (in)equalities,
  \item and all its (in)equalities $g \sim 0$ are such that~$\deg(g) \leq (2d)^{O(n)}$ and $\height(g) \leq (h+1)^{(2d)^{O(n^2)}}$.
\end{enumerate}

\paragraph{\textit{Step~II \textup{(}runtime: 2-exp.~or 3-exp.~in $\size(\phi)$, depending on the value of~$k$\textup{)}}}
For each variable $u_i$, the algorithm guesses an integer $g_i$ written in unary (lines~\ref{algo:line7} and~\ref{algo:line8}). Let $H \coloneqq \max(8,h(\psi))$ and $D \coloneqq \deg(\psi)+2$.
By~\Cref{theorem:small-model-property}, 
\begin{align*}
  \abs{g_i} 
  & \leq \left(2^c \ceil{\ln H}\right)^{D^{2^5 n^2} k^{D^{8n}}} 
  \leq \Big(2^c {\lceil\ln \big((2(h+1))^{(2d)^{O(n^2)}}\big)\rceil}\Big)^{(2d)^{O(n^3)} k^{(2d)^{O(n^2)}}},
\end{align*}
that is, if $k = 1$ then $\abs{g_i}$ is doubly exponential in $\size(\phi)$, and otherwise, for every $k \geq 2$, $\abs{g_i}$ is triply exponential in $\size(\psi)$.
We can implement lines~\ref{algo:line7}--\ref{algo:line9} deterministically in the following na\"ive way: 
\begin{algorithmic}[1]
  \setcounter{ALG@line}{6}
  \For{$(g_1,\dots,g_n) \in P^n$}
    \If{the assignment $(u_1 = \cn^{g_1}, \dots, u_n = \cn^{g_n})$ is a solution to $\psi$}
      \myreturn $\top$
    \EndIf
  \EndFor
  \State \myreturn $\bot$
\end{algorithmic}
Since each $g_i$ is stored in unary encoding, the number of iterations of
the \textbf{for} loop above is either doubly or triply exponential in
$\size(\phi)$, depending on whether $k = 1$.

\paragraph{\textit{Step~III \textup{(}runtime: $2$-exp.~or $3$-exp.~in $\size(\phi)$, depending
on the value of~$k$\textup{)}}}
The algorithm evaluates whether $(u_1 = \cn^{g_1}, \dots, u_n =
\cn^{g_n})$ is a solution to $\psi$. As discussed in the body of the
paper, $\psi(\cn^{g_1},\dots,\cn^{g_n})$ is a Boolean combination of
polynomial (in)equalities $p(\cn) \sim 0$, where $\cn$ may occur with
negative powers (as some $g_i$ may be negative). We rewrite each
(in)equality $p(\cn) \sim 0$ as $\cn^{-d} \cdot p \sim 0$, where $d$ is
the smallest negative integer occurring as a power of $\cn$ in~$p$ (or $0$
if such an integer does not exist), thus obtaining a formula where all
polynomials have non-negative degrees. Let us denote by $\psi'$ this
formula. This update takes polynomial time in the size of
$\psi(\cn^{g_1},\dots,\cn^{g_n})$; that is doubly or triply exponential
time in~$\size(\phi)$, depending on which case among $k = 1$ or $k \geq 2$
we are considering.

After this update, we determine the sign that each inequality in~$\psi'$. These
inequalities are of the form $p(\cn) \sim 0$, and hence this problem can be
solved with~\Cref{algo:sign-evaluation}. (Note that the degree $p$ depends on
$g_1,\dots,g_n$.) By~\Cref{lemma:runtime-sign-poly-root-barrier}, the runtime of
this algorithm is polynomial in the size of $p$; which again is doubly or triply
exponential in~$\size(\phi)$, depending on $k$. This enables us to simplify all
inequalities to either $\top$ or $\bot$, to then return $\top$ or $\bot$
depending on the Boolean structure of~$\psi'$ Since $\psi$ has size exponential
in $\size(\phi)$, evaluating the Boolean structure of $\psi'$ takes exponential
time.

Putting all together, we conclude that~\Cref{algo:main-procedure} 
runs in doubly exponential time if $k = 1$, 
and in triply exponential time if $k \geq 2$, 
thus establishing~\Cref{theorem:general-result-root-barrier}.

\section{Finding solutions over integer powers of \texorpdfstring{$\xi$}{the base}}
\label{sec:solving-substructure}

In this section we prove
\Cref{theorem:small-model-property}, i.e., we show that $\exists
\ipow{\cn}$ has a~\emph{small witness property}. The proof is split into two parts:
\begin{enumerate}
  \item We first give a quantifier-elimination-like procedure for $\exists
  \ipow{\cn}$. Instead of targeting formula equivalence, we only focus on
  equisatisfiability: given a formula $\exists y\, \phi(y,\vec{x})$, with $\phi$
  quantifier-free, the procedure derives an \emph{equisatisfiable}
  quantifier-free formula $\psi(\vec x)$. Preserving
  equisatisfiability, instead of equivalence, is advantageous complexity-wise.
  (Our procedure
  preserves equivalence for sentences, as
  these are equivalent to~$\top$~or~$\bot$.)
  \item By analysing our quantifier elimination procedure, we derive the bounds
  on the set $P_{\psi}$ from~\Cref{theorem:small-model-property} required to
  complete the proof. This step is similar to the \emph{quantifier
  relativisation} technique for Presburger arithmetic (see, e.g.,~\cite[Theorem
  2.2]{Weispfenning90}).
\end{enumerate}
Some core mechanisms of our quantifier-elimination-like procedure follow
observations done by Avigad and Yin for their (equivalence-preserving) quantifier elimination
procedure~\cite{AvigadY07}. Apart from targeting equisatisfiability, a key property of our procedure is that it does not require
appealing to a quantifier elimination procedure for the theory of the reals. The
procedure in~\cite{AvigadY07} calls such a procedure once for each eliminated
variable instead.

\subsection{Quantifier elimination}
\label{subsection:quantifier-elimination}
Fix a real number $\cn > 1$. In this section, we rely on some auxiliary notation
and definitions:
\begin{itemize}
  \item We often see an integer polynomial $p(\cn, \vec x)$ as a polynomial in
  variables $\vec x = (x_1,\dots,x_m)$ having as coefficients univariate integer
  polynomials on $\cn$, i.e., ${p(\cn, \vec x) = \sum_{i=1}^n q_i(\cn) \cdot
  \vec{x}^{\vec{d}_i}}$, where the notation $\vec{x}^{\vec{d}_i}$ is short for
  the \emph{monomial} $\prod_{j=1}^m x_j^{d_{i,j}}$, with $\vec{d}_{i} =
  (d_{i,1},\dots,d_{i,m})$.
  \item We sometimes write polynomial (in)equalities using Laurent polynomials,
  i.e., polynomials with negative powers. For instance, \Cref{lemma:lambda-close-to-variable} below features equalities with monomials
  $\cn^g \cdot {\vec{x}}^{\vec{d}_i}$ where $g$ may be a negative integer. Laurent polynomials are just a
  shortcut for us, as one can opportunely manipulate the (in)equalities to make
  all powers non-negative (as we did
  in~\Cref{subsection:polynomial-sign-evaluation}): a polynomial
  (in)equality $p(\cn,x_1,\dots,x_m) \sim 0$ is rewritten as
  $p(\cn,x_1,\dots,x_m) \cdot \cn^{-d} \cdot x_1^{-d_1} \cdot {\dots} \cdot
  x_m^{-d_m} \sim 0$, where $d_i$ (resp.~$d$) is the smallest negative integer
  occurring as a power of $x_i$ (resp.~$\cn$) in $p$ (or $0$ if such a negative
  integer does not exist). Observe that this transformation does not change the
  number of monomials nor the height of the polynomial $p$, but it may double the
  degree of each variable and of $\cn$.
  \item Given a formula $\phi$, a variable $x$ and a Laurent polynomial
  $q(\vec{y})$, we write $\phi\sub{q(\vec y)}{x}$ for the formula obtained from
  $\phi$ by replacing every occurrence of $x$ by $q(\vec y)$, and then
  updating all polynomial (in)equalities with negative degrees in the way
  described above.
  \item We write $\lambda \colon \R_{> 0} \to \ipow{\cn}$ for the function
  mapping $a \in \R_{> 0}$ to the largest integer power of $\cn$ that is less or
  equal than $a$, i.e., $\lambda(a)$ is the only element of $\ipow{\cn}$
  satisfying $\lambda(a) \leq a < \cn \cdot \lambda(a)$. 
\end{itemize}
\noindent 
The relation $\lambda(p(\cn,\vec x)) = y$, where $p$ is an integer polynomial,
is definable in $\exists \ipow{\cn}$ as $p(\cn,\vec x) > 0 \land y \leq
p(\cn,\vec x) < \cn \cdot y$. To obtain a quantifier elimination procedure, we
must first understand what values can $y$ take given~$p(\cn,\vec x)$. 
\Cref{lemma:lambda-close-to-variable} answers this question: 


\begin{restatable}{lemma}{LemmaLambdaCloseToVariable}
  \label{lemma:lambda-close-to-variable}
  Let $p(\cn,\vec x) \coloneqq \sum_{i = 1}^n (q_i(\cn) \cdot
  {\vec{x}}^{\vec{d}_i})$, where each~$q_i$ is a univariate integer polynomial.
  In the theory~$\exists \ipow{\cn}$, the formula $p(\cn,\vec x) > 0$ entails the formula $\textstyle\bigvee_{i=1}^n \bigvee_{g \in G} \lambda(p(\cn,\vec x)) = \cn^g \cdot {\vec{x}}^{\vec{d}_i}$\,, for some finite set $G \subseteq \Z$.
  Moreover:
  \begin{enumerate}[label=\Roman*., ref=\Roman*]
    \item\label{lemma:lambda-close-to-variable:i1} If $\cn$ is a computable transcendental number, there is an
          algorithm computing $G$ from $p$.
    \item\label{lemma:lambda-close-to-variable:i2} If $\cn$ has a root barrier $\sigma(d,h) \coloneqq c \cdot
            (d+\ceil{\ln(h)})^k$, for some $c,k \in \N_{\geq 1}$, then
            \vspace{-3pt}
            \begin{equation*}
              G\coloneqq\left[-L..L\right],
              \qquad 
              \text{where }
              L \coloneqq \left(2^{3c}D\ceil{\ln(H)}\right)^{6nk^{3n}},
            \end{equation*}
            with $H \coloneqq \max\{8,\height(q_i) : i \in [1..n]\}$, and $D \coloneqq \max\{\deg(q_i)+2 : i \in [1..n]\}$.
  \end{enumerate}
\end{restatable}

\begin{proof}
  We start by replacing every monomial 
  $\prod_{j=1}^m x_j^{d_{i,j}}$ with a term $\cn^{z_i}$, where 
  $z_i$ is a fresh variable ranging over $\Z$.
  It now suffices to show the following statement instead:
  \begin{itemize}
    \item[]\textit{Let $p(\vec{z}) \coloneqq \sum_{i=1}^n q_i(\cn) \cdot \cn^{z_i}$,
    where $\vec z = (z_1,\dots,z_n)$ and each~$q_i(x)$ is an integer polynomial.
    There is a finite set $G \subseteq \Z$ with the following property: 
    for every~${\vec z^* \in \Z^n}$,
    if $p(\vec z^*) > 0$
    then $\lambda(p(\vec z^*)) = \cn^g \cdot \cn^{z_i^*}$ for some
    $g \in G$ and $i \in [1..n]$.
    Moreover:}
    \begin{enumerate}[label=\Roman*., ref=\Roman*]
      \item \textit{If $\cn$ is computable and transcendental, there is an
            algorithm computing $G$ from $p$.}
      \item \textit{If $\cn$ has a root barrier $\sigma(d,h) \coloneqq c \cdot
              {(d+\ceil{\ln(h)})}^k$, for some $c,k \in \N_{\geq 1}$, then,
              \begin{equation*}
                G\coloneqq{\left[-L..L\right]},
                \qquad 
                \text{where }
                L \coloneqq {\left(2^{3c}D\ceil{\ln(H)}\right)}^{6nk^{3n}},
              \end{equation*}
              with $H \coloneqq \max\{8,\height(q_i) : i \in [1..n]\}$, and $D \coloneqq \max\{\deg(q_i)+2 : i \in [1..n]\}$.}
    \end{enumerate}
  \end{itemize}
  
  \medskip\noindent
  Note that for $n = 0$ we have $p(\vec z^*) = 0$ for every $\vec z^* \in \Z^n$, and we can take $G = \emptyset$. Therefore, throughout the proof, we assume $n \geq 1$.
  We start by showing the existence of the finite set~$G$. 
  To prove this, we first fix a vector $\vec z^* = (z_1^*,\dots,z_n^*) \in \Z^n$ such that $p(\vec z^*) > 0$, 
  and use it to derive a definition for $G$ that does not, in fact, depend on $\vec z^*$.
  Without loss of generality, we work under the additional assumption that
  $z_1^* \geq \cdots \geq z_n^*$.

  The following claim provides an analysis on the value of $\lambda(p(\vec z^*))$.
  
  \begin{claim}\label{claim1:lambda-close-to-variable}
    There is a non-empty interval $[j..\ell]$, with $j,\ell \in [1..n]$,
    and natural numbers ${g_j},\dots,{g_{\ell-1}}$
    with respect to which the recursively defined polynomials
    $Q_j,\dots,Q_{\ell}$ given by
    \begin{align*}
      Q_j(x) & \coloneqq q_j(x),                                  \\
      Q_r(x) & \coloneqq Q_{r-1}(x) \cdot x^{g_{r-1}} + q_{r}(x),
              & \tag*{$\text{for every } r \in [j+1,\ell],$}
    \end{align*}
    satisfy the following properties:
    \begin{enumerate}[label=\Alph*., ref=\Alph*]
      \item\label{claim1:lambda-close-to-variable:A}
      the numbers $Q_j(\cn),\dots,Q_{\ell-1}(\cn)$ are all non-zero, and
      $Q_\ell(\cn)$ is (strictly) positive,
      \item\label{claim1:lambda-close-to-variable:B}
      for every $r \in [j..\ell-1]$,
      the number $\cn^{g_r}$ belongs to the interval
      $\big[1\,,\,\frac{\abs{q_{r+1}(\cn)}+\cdots+\abs{q_n(\cn)}}{\abs{Q_r(\cn)}}\big]$,
      and
      \item\label{claim1:lambda-close-to-variable:C}
      either\, $\lambda(p(\vec z^*)) = \lambda(Q_\ell(\cn)) \cdot
        \cn^{z_\ell^*}$\, or\, $\frac{\lambda(Q_\ell(\cn) \cdot (\cn-1))}{\cn} \cdot
        \cn^{z_\ell^*} \leq \lambda(p(\vec z^*)) \leq
        \frac{\lambda(Q_\ell(\cn) \cdot (\cn+1))}{\cn} \cdot \cn^{z_\ell^*}$.
    \end{enumerate}
  \end{claim}
  
  \begin{proof}
    The proof is by induction on $n$.
    \begin{description}
      \item[base case: $n = 1$]
        In this case, $p(z_1)$ is the expression $q_1(\cn) \cdot \cn^{z_1}$. By
        definition of $\lambda$, $\lambda(p(z_1^*)) = \lambda(q_1(\cn)) \cdot
          \cn^{z_1^*}$. Observe that $p(z_1^*) > 0$ implies $q_1(\cn) > 0$, and
        thus
        $\lambda(q_1(\cn))$ is a defined integer power of $\cn$. Taking the
        interval $[1..1]$ shows~\Cref{claim1:lambda-close-to-variable}.
      \item[induction step: $n \geq 2$]
        Below, we assume $q_1(\cn)$ to be non-zero.
        Indeed, if $q_1(\cn) = 0$,
        we can then apply the induction hypothesis on
        $\hat{p}(z_2,\dots,z_n) \coloneqq \sum_{i=2}^n q_i(\cn) \cdot
          \cn^{z_i}$,
        concluding the proof (since $p(\vec z^*) =
          \hat{p}(z_2^*,\dots,z_n^*)$).

        We split the proof depending on whether $\cn^{z_1^*} \geq
          \frac{\sum_{i=2}^n \abs{q_i(\cn)}}{\abs{q_1(\cn)}} \cdot
          \cn^{z_2^*+1}$ holds.
        \begin{description}
          \item[case: $\cn^{z_1^*} \geq \frac{\sum_{i=2}^n
          \abs{q_i(\cn)}}{\abs{q_1(\cn)}} \cdot \cn^{z_2^*+1}$]
            Observe that in this case, $q_1(\cn)$ must be positive. We show
            that
            $\frac{\lambda(q_1(\cn) \cdot (\cn-1))}{\cn} \cdot \cn^{z_1^*} \leq
              \lambda(p(\vec{z}^*)) \leq \frac{\lambda(q_1(\cn) \cdot (\cn+1))}{\cn}
              \cdot \cn^{z_1^*}$,
            thus establishing that taking the interval $[1..1]$
            proves~\Cref{claim1:lambda-close-to-variable} also in this case.
            For the lower bound:
            \[
              \hspace{1.4cm}
              \begin{aligned} 
                \quad p(\vec z^*) & \geq q_1(\cn) \cdot \cn^{z_1^*} - \sum\nolimits_{i=2}^n
                \abs{q_i(\cn)} \cdot \cn^{z_i^*}
                            & \text{by def.~of~$p$}
                \\
                            & \geq q_1(\cn) \cdot \cn^{z_1^*} - \cn^{z_2^*}\cdot
                \sum\nolimits_{i=2}^n
                \abs{q_i(\cn)}
                            & z_2^* \geq z_i^* \text{ for all $i \in [2..n]$}
                \\
                            & \geq
                q_1(\cn) \cdot \cn^{z_1^*} - q_1(\cn) \cdot \cn^{z_1^*-1}
                            & \text{assumption of this case and $q_1(\cn) > 0$}
                \\
                            & \geq
                q_1(\cn) \cdot (\cn - 1) \cdot \cn^{z_1^*-1}.
              \end{aligned}
            \]
            Since $a \geq b$ implies $\lambda(a) \geq \lambda(b)$, we thus
            obtain
            $\lambda(p(\vec z^*)) \geq \frac{\lambda(q_1(\cn) \cdot
                (\cn-1))}{\cn} \cdot
              \cn^{z_1^*}$.

            For the upper bound:
            \[
              \hspace{1.8cm}
              \begin{aligned}
                p(\vec z^*) & \leq q_1(\cn)\cdot \cn^{z_1^*} + \sum_{i=2}^n
                \abs{q_i(\cn)} \cdot \cn^{z_2^*}
                            & \text{by def.~of $p$, and $z_2^* \geq z_i^*$  for
                  all $i \in
                    [2..n]$}
                \\
                            & \leq q_1(\cn)\cdot \cn^{z_1^*} + q_1(\cn) \cdot
                \cn^{z_1^*-1}
                            & \text{assumption of this case and $q_1(\cn) > 0$}
                \\
                            & \leq q_1(\cn) \cdot (\cn+1) \cdot \cn^{z_1^*-1},
              \end{aligned}
            \]
            and again from the properties of $\lambda$,
            we obtain $\lambda(p(\vec z^*)) \leq \frac{\lambda(q_1(\cn) \cdot
                (\cn
                + 1))}{\cn} \cdot \cn^{z_1^*}$.

          \item[case: $\cn^{z_1^*} < \frac{\sum_{i=2}^n
          \abs{q_i(\cn)}}{\abs{q_1(\cn)}} \cdot \cn^{z_2+1}$] We have
            $\cn^{z_1^*} \leq
              \frac{\sum_{i=2}^n \abs{q_i(\cn)}}{\abs{q_1(\cn)}} \cdot
              \cn^{z_2}$.
            Since $z_1^* \geq z_2^*$,
            there must be $g_1 \in \N$ such that $\cn^{g_1} \in
              \big[1,\frac{\sum_{i=2}^n \abs{q_i(\cn)}}{\abs{q_1(\cn)}}\big]$
            and $\cn^{z_1^*} = \cn^{g_1} \cdot \cn^{z_2^*}$.
            We define
            \[
              \hspace{1.4cm}
              q_2'(x) \coloneqq q_1(x) \cdot x^{g_1} + q_2(x),
              \qquad
              p'(z_2,\dots,z_n) \coloneqq q_2'(\cn) \cdot \cn^{z_2} +
              \sum_{i=3}^n
              q_i(\cn) \cdot \cn^{z_i}.
            \]
            So, $p(\vec z^*) = p'(\vec z_2^*)$,
            where $\vec z_2^* \coloneqq (z_2^*,\dots,z_\ell^*)$.
            By induction hypothesis,
            there is a non-empty interval $[j..\ell]$, with $j,\ell \in
              [2..n]$,
            and natural numbers ${g_j},\dots,{g_{\ell-1}}$
            with respect to which the recursively defined polynomials
            $Q_j,\dots,Q_{\ell}$ given by
            \begin{align*}
              Q_j(x) & \coloneqq
              \begin{cases}
                q_2'(x) & \text{if $j = 2$} \\
                q_j(x)   & \text{otherwise}
              \end{cases}                                \\
              Q_r(x) & \coloneqq Q_{r-1}(x) \cdot x^{g_{r-1}} + q_{r}(x),
                      & \tag*{$\text{for every } r \in [j+1..\ell],$}
            \end{align*}
            satisfy that
            (\labeltext{A$'$}{claim1:lambda-close-to-variable:Ap})
            $Q_j(\cn),\dots,Q_{\ell-1}(\cn)$ are all non-zero
            and~$Q_\ell(\cn)$ is positive,
            (\labeltext{B$'$}{claim1:lambda-close-to-variable:Bp})~for every $r
              \in
              [j,\ell-1]$, the number
            $\cn^{g_r}\in  
            \big[1\,,\,\frac{\abs{q_{r+1}(\cn)}+\cdots+\abs{q_n(\cn)}}{\abs{Q_r(\cn)}}\big]$,
            and
            (\labeltext{C$'$}{claim1:lambda-close-to-variable:Cp})~either
            $\lambda(p'(\vec z_2^*)) = \lambda(Q_\ell(\cn)) \cdot
              \cn^{z_\ell^*}$ or
            $\frac{\lambda(Q_\ell(\cn) \cdot (\cn-1))}{\cn} \cdot \cn^{z_\ell^*} \leq
              \lambda(p'(\vec z_2^*)) \leq
              \frac{\lambda(Q_\ell(\cn) \cdot (\cn+1))}{\cn} \cdot
              \cn^{z_\ell^*}$.

            If $j \neq 2$, then $Q_j(x) = q_j(x)$;
            since $p(\vec z^*) = p'(\vec z_2^*)$
            we conclude that the interval $[j..\ell]$
            and the numbers ${g_j},\dots,{g_{\ell-1}}$
            defined for $p'$ also establish the claim for $p$.

            Otherwise, when $j = 2$ we have $Q_j(x) = q_2'(x) = q_1(x) \cdot
              x^{g_1} + q_2(x)$.
            Recall that $q_1(\cn)$ is non-zero and that,
            by definition of $g_1$, $\cn^{g_1} \in
              \big[1,\frac{\sum_{i=2}^n \abs{q_i(\cn)}}{\abs{q_1(\cn)}}\big]$.
            Therefore,
            taking the interval $[1..\ell]$ and the numbers
            $g_1,\dots,g_{\ell-1}$ proves
            the claim for $p$.
            \qedhere%
        \end{description}
    \end{description}
  \end{proof}

  \noindent 
  With~\Cref{claim1:lambda-close-to-variable} at hand, we now argue that the
  finite set $G \subseteq \Z$ required by the lemma exists.
  The key observation is that the definition of $Q_\ell$ from~\Cref{claim1:lambda-close-to-variable} does not depend on~$\vec z^*$. 
  Hence, a suitable set $G$ can be defined as follows. 
  Let $\mathcal{Q}$ be the set of all polynomials $Q$
  for which there are $j \leq \ell \in [1..n]$, $g_j,\dots,g_{\ell-1} \in \N$, and polynomials $Q_j,\dots,Q_{\ell}$ 
  such that:
    \begin{enumerate}
      \item the polynomial $Q$ is equal to $Q_\ell$,
      \item\label{mcQ-polynomials} the polynomials $Q_j^{},\dots,Q_\ell$ are defined as 
          \begin{align*}
          Q_j(x) & \coloneqq q_j(x),                                  \\
          Q_r(x) & \coloneqq Q_{r-1}(x) \cdot x^{g_{r-1}} + q_{r}(x),
              & \tag*{$\text{for every } r \in [j+1,\ell],$}
          \end{align*}
      \item
      the numbers $Q_j(\cn),\dots,Q_{\ell-1}(\cn)$ are all non-zero, and
      $Q_\ell(\cn)$ is (strictly) positive,
      \item\label{mcQ-finiteness-of-gr}
      for every $r \in [j..\ell-1]$, and
      the number $\cn^{g_r}$ belongs to the interval
      $\big[1\,,\,\frac{\abs{q_{r+1}(\cn)}+\cdots+\abs{q_n(\cn)}}{\abs{Q_r(\cn)}}\big]$.
    \end{enumerate}
  In a nutshell, $\mathcal{Q}$ contains all polynomials $Q_\ell$ that might be considered in~\Cref{claim1:lambda-close-to-variable} as the vector $\vec z^*$ varies. Items~\ref{mcQ-polynomials}--\ref{mcQ-finiteness-of-gr} ensure that $\mathcal{Q}$ is a finite set.
  We define $G \coloneqq [\min B.. \max B]$, where $B$ is defined as the set
  \[ 
    B \coloneqq \left\{ \beta \in \Z : \text{there is } Q \in \mathcal{Q} \text{ such that } \cn^{\beta} \in \Big\{{\textstyle\lambda(Q(\cn)), \frac{\lambda(Q(\cn) \cdot (\cn-1))}{\cn}, \frac{\lambda(Q(\cn) \cdot (\cn+1))}{\cn}}\Big\}
    \right\}.
  \]
  Since $\mathcal{Q}$ is finite, then so are $B$ and $G$.

  It is now simple to see that $G$ satisfies the property required by the first statement of the lemma. 
  Indeed, consider a vector $\vec z^* = (z_1^*,\dots,z_n^*) \in \Z^n$ such that $p(\vec z^*) > 0$ (this is not necessarily the vector we have fixed at the beginning of the proof). 
  By definition of $\mathcal{Q}$ and by~\Cref{claim1:lambda-close-to-variable},
  there is a polynomial $Q$ in $\mathcal{Q}$ such that:
  \begin{itemize} 
      \item $Q(\cn)$ is strictly positive (and so $\lambda(Q(\cn))$ is well-defined). Since $\cn > 1$, observe that this means that also $Q(\cn) \cdot (\cn-1)$ and $Q(\cn) \cdot (\cn+1)$ are strictly positive.
      \item Either $\lambda(p(\vec z^*)) = \lambda(Q(\cn)) \cdot \cn^{z_i^*}$ or $\frac{\lambda(Q(\cn) \cdot (\cn-1))}{\cn} \cdot \cn^{z_i^*} \leq \lambda(p(\vec z^*)) \leq \frac{\lambda(Q(\cn) \cdot (\cn+1))}{\cn} \cdot \cn^{z_i^*}$, for some $i \in [1..n]$ (this follows by Property~\ref{claim1:lambda-close-to-variable:C} of~\Cref{claim1:lambda-close-to-variable}). 

      In the latter case of $\frac{\lambda(Q(\cn) \cdot (\cn-1))}{\cn} \cdot \cn^{z_i^*} \leq \lambda(p(\vec z^*)) \leq \frac{\lambda(Q(\cn) \cdot (\cn+1))}{\cn} \cdot \cn^{z_i^*}$, observe that $\lambda(p(\vec z^*)) = \cn^\beta \cdot \cn^{z_i^*}$, for some $\cn^\beta \in [\frac{\lambda(Q(\cn) \cdot (\cn-1))}{\cn},\frac{\lambda(Q(\cn) \cdot (\cn+1))}{\cn}]$.
  \end{itemize}
  By definition of $B$ and $G$, 
  we conclude that $\lambda(p(\vec z^*)) = \cn^g \cdot \cn^{z_i^*}$ 
  for some $g \in G$ and $i \in [1..n]$. 
  This concludes the proof of the first statement of the lemma.

  We move to the second part of the lemma, which adds further assumptions
  on~$\cn$. This part still relies on the definitions of the sets $\mathcal{Q}$, $B$ and $G$ above.

  \paragraph{\textit{Case: $\cn$ is a computable transcendental number \textup{(}Item~\eqref{lemma:lambda-close-to-variable:i1}\textup{)}}}
  Assume $\cn$ is a transcendental number computed by a Turing machine $T$. We provide an algorithm for computing a superset of the set $G$. Here is a high-level pseudocode of the algorithm:

  \begin{algorithmic}[1]
      \State\label{algo-lambda-trans-1} compute a finite set of polynomials $\mathcal{Q}'$ that includes all polynomials in $\mathcal{Q}$
      \State\label{algo-lambda-trans-2} remove from $\mathcal{Q}'$ all polynomials $Q$ such that $Q(\cn) \leq 0$
      \State\label{algo-lambda-trans-3} compute rationals $\ell,u > 0$ such that~$\ell \leq \frac{Q(\cn)\cdot (\cn-1)}{\cn^2}$ and $Q(\cn) \cdot (\cn+1) \leq u$, for all $Q$ in~$\mathcal{Q}'$
      \State\label{algo-lambda-trans-4} \textbf{return} a superset of $\{ \beta \in \Z : \ell \leq \cn^\beta \leq u \}$
  \end{algorithmic}
  The correctness of this algorithm is immediate from the definition of the sets $\mathcal{Q}$, $B$ and $G$. 
  In particular, note that $\{ \cn^\beta : \beta \in B\} \subseteq [\ell..u]$, 
  because for every $\alpha > 0$ we have $\frac{\alpha}{\cn} < \lambda(\alpha) \leq \alpha$ (by definition of $\lambda$), and moreover 
  $\frac{Q(\cn) \cdot (\cn-1)}{\cn^2} \leq \frac{Q(\cn)}{\cn} \leq Q(\cn) \leq Q(\cn)\cdot (\cn+1)$ (recall that $Q(\cn) > 0$ and $\cn > 1$). So, $G$ is a subset of the set in output of the algorithm, as required. Below, we give more information on how to implement each line of the algorithm (starting for simplicity with line~\ref{algo-lambda-trans-2}), showing its effectiveness. We will often rely on the following claim:

  \begin{claim}\label{claim:lu-p}
    Given an integer polynomial $p(x)$, one can compute 
    \begin{enumerate}
      \item\label{claim:lu-p:i1} a rational number $\ell'$ such that $0 <
      \ell' \leq \abs{p(\cn)}$;
      \item\label{claim:lu-p:i2} a rational number $u'$ such that $\abs{p(\cn)}
      \leq u'$.
    \end{enumerate}
  \end{claim}

  \begin{proof}
    Recall that $\abs{T_0}+1$ is an upper bound to the transcendental number $\cn > 1$. 
    By iterating over the natural numbers, we find the smallest $L \in \N$ 
    such that 
    $\abs{p(T_M)} > \frac{1}{2^L} \geq \abs{p(\cn) - p(T_M)}$, 
    where $M \coloneqq L + \ceil{\log(\height(p)+1)} + 2 \deg(p) \cdot \ceil{\log(\abs{T_0}+2)}$.
    The existence of such an $L$ follows from~\Cref{lemma:approx-univ-polynomial} (for the second inequality) together with the fact that $p(\cn) \neq 0$, and so $\lim_{n \to \infty} \abs{p(T_n)} \neq 0$ whereas $\lim_{m \to \infty} \frac{1}{2^m} = 0$
    (which implies the first inequality).
    For~\Cref{claim:lu-p:i1}, take~$\ell'$ to be $\abs{p(T_M)} - \frac{1}{2^L}$.
    For~\Cref{claim:lu-p:i2}, take~$u'$ to be $\abs{p(T_M)} + \frac{1}{2^L}$.
  \end{proof}
  
  \noindent 
  Here is the argument for the effectiveness of the algorithm:
  \begin{itemize}
      \item \textit{line~\ref{algo-lambda-trans-2}.} 
      In general, to evaluate the sign of a polynomial~$p$ at $\cn$, one
      relies on the fact that $p(\cn)$ must be different from~$0$ (because
      $\cn$ is transcendental). Then, we can rely on the fast-convergence
      sequence of rational numbers~$T_0,T_1,\dots$ to find $n \in \N$ such that
      $|p(\cn) - p(T_{n})|$ is guaranteed to be less than $|p(T_{n})|$. The
      sign of $p(\cn)$ then agrees with the sign of~$p(T_{n})$, and the latter
      can be easily computed.
      \item \textit{line~\ref{algo-lambda-trans-1}.} By definition of $\mathcal{Q}$, the fact that such a set $\mathcal{Q}'$ can be computed follows from the fact that we can compute an upper bound, for every $j \leq \ell \in [1..n]$ and $r \in [j..\ell-1]$, to the maximum $g_r$ such that $\cn^{g_r} \in \big[1\,,\,\frac{\abs{q_{r+1}(\cn)}+\cdots+\abs{q_n(\cn)}}{\abs{Q_r(\cn)}}\big]$, where $Q_r$ is any polynomial that can be defined in terms of $g_{j},\dots,g_{r-1}$ following the recursive definition of Item~\ref{mcQ-polynomials}.
      It suffices to find a positive lower bound $\ell' \in \Q$ to $\abs{Q_r(\cn)}$, 
      as well as upper bounds $u_i' \in \Q$ to every $\abs{q_i(\cn)}$, with $i \in [r+1..n]$. The rationals $\ell',u_{r+1}',\dots,u_n'$ are computed following~\Cref{claim:lu-p}. Then, 
      $\frac{\abs{q_{r+1}(\cn)}+\cdots+\abs{q_n(\cn)}}{\abs{Q_r(\cn)}} \leq \frac{u_{r+1}'+\cdots+u_{n}'}{\ell'} \leq \ceil{\frac{u_{r+1}'+\cdots+u_{n}'}{\ell'}} \eqqcolon D \in \N$.
      To bound $g_r$ it now suffices to find the largest integer power of $\cn$ that is less or equal to~$D$. This can be done using the algorithm for the sign evaluation problem described for line~\ref{algo-lambda-trans-2}: iteratively, starting at $i = 0$, we test whether $\cn^i - D$ is non-positive; we increase $i$ by $1$ if this test is successful, and return $i-1$ otherwise. 
      \item \textit{line~\ref{algo-lambda-trans-3}.} 
      Recall that $Q(\cn)$ is positive and $\cn > 1$.
      Following~\Cref{claim:lu-p}, we can find positive rationals $\ell',u_1',u_2'$ 
      such that $\ell' < Q(\cn) \cdot (\cn-1)$, $\cn^2 \leq u_1'$ 
      and $Q(\cn) \cdot (\cn+1) \leq u_2'$.
      The first two inequalities imply $0 < \frac{\ell'}{u_1'} < \frac{Q(\cn) \cdot (\cn-1)}{\cn^2}$.
      We can then take $\ell \coloneqq \frac{\ell'}{u_1'}$ and $u \coloneqq u_2'$. 
      Note that we have $\frac{Q(\cn) \cdot (\cn-1)}{\cn^2} < Q(\cn) \cdot (\cn + 1)$, 
      and therefore $\ell < u$.
      \item \textit{line~\ref{algo-lambda-trans-4}.} Given $\ell$ and $u$, we can compute a superset of those $\beta \in \Z$ such that $\ell \leq \cn^\beta \leq u$ by iterated calls to the algorithm for the sign evaluation problem. First, we can extend the interval $[\ell,u]$ to always include $1$: if $\ell > 1$, update $\ell$ to $1$; if $u < 1$, update $u$ to $1$.
      This ensures $\cn^0 \in [\ell,u]$. We can then find the largest~$\cn^i$ that is less or equal to~$u$ by testing whether $\cn^i - u$ is non-positive for increasing $i$ starting at $0$, as we did in line~\ref{algo-lambda-trans-1} for finding the largest integer powers less or equal to~$D$. Similarly, we can find the smallest integer power~$\cn^{-i}$ that is greater or equal than $\ell$ by testing whether $1 - \ell \cdot \cn^{i}$ is non-negative for increasing $i$ starting at $0$.
  \end{itemize}

  \paragraph{\textit{{Case: $\cn$ has a polynomial root barrier \textup{(}Item~\eqref{lemma:lambda-close-to-variable:i2}\textup{)}}}}
  Assume now $\cn$ to have a polynomial root barrier 
  $\sigma(d,h) \coloneqq c \cdot {(d+\ceil{\ln(h)})}^k$, 
  with $c,k \in \N_{\geq 1}$.
  In this case, we need to provide an explicit set~$G$.
  We do so by analysing the polynomials 
  $Q_j,\dots,Q_\ell$
  and the natural numbers 
  $g_j,\dots,g_{\ell-1}$ 
  introduced in~\Cref{claim1:lambda-close-to-variable}
  and used in the definition of the set $\mathcal{Q}$,
  and by providing both lower and upper bounds for the positive numbers
  $Q_\ell(\cn)$,
  $Q_{\ell}(\cn) \cdot (\cn-1)$ and 
  $Q_{\ell}(\cn) \cdot (\cn+1)$.
  These bounds entail bounds on the integers occurring in the set $B$
  introduced at the end of the proof of the first statement of the lemma.

  We start by providing a bound on the degrees and heights
  of $Q_j,\dots,Q_\ell$:
  \begin{claim}\label{claim2:lambda-close-to-variable}
    For every $r \in [j..\ell]$,
    $\deg(Q_r) \leq D + \sum_{s=j}^{r-1} g_s$
    and $\height(Q_r) \leq (r-j+1) \cdot H$.
  \end{claim}
  \begin{proof}
    By a straightforward induction on $r$, using the definitions of
    $Q_j,\dots,Q_\ell$.
  \end{proof}
  In~\Cref{claim2:lambda-close-to-variable},
  note that $(r-j+1) \cdot H \leq n \cdot H$, and therefore we obtain a bound on $\height(Q_\ell)$ that does not depend on the previous $Q_r$. 
  Below, we prove a similar bound for $\deg(Q_\ell)$. 
  \begin{claim}\label{claim3:lambda-close-to-variable}
    The degree of $Q_\ell$ is bounded as follows:
    \[ 
      \deg(Q_\ell) \leq {\left(\frac{2c \cdot D \cdot \ln(H)}{\ln(1 + \frac{1}{e^c})}\right)}^{5nk^{n+1}}.
    \]
  \end{claim}
  \begin{proof}
    By Property~\ref{claim1:lambda-close-to-variable:A},
    $Q_j(\cn),\dots,Q_\ell(\cn)$ are non-zero.
    Then, \Cref{claim2:lambda-close-to-variable}
    and the fact that $\sigma$ is a root barrier for $\cn$ entail
    \begin{equation}
      \label{inequality1:lambda-close-to-variable}
      \ln \abs{Q_r(\cn)} \geq - c \cdot {\Big(D + \ceil{\ln(n \cdot H)} +
      \sum\nolimits_{s=j}^{r-1} g_s \Big)}^k.
    \end{equation}
    Analogously, since $\cn > 1$, we can consider the polynomial $x - 1$ in order to obtain a lower bound on $\cn$, via the root barrier $\sigma$. We obtain 
    \begin{equation}
      \label{inequality1:lower-bound-cn}
      \cn \geq 1 + \frac{1}{e^c}.
    \end{equation}
    Given $r \in [1..n]$, we also have
    \begin{equation}
      \label{inequality2:lambda-close-to-variable}
      \abs{q_r(\cn)} \leq 
        H \cdot \sum_{i=0}^d \cn^i \leq H \cdot D \cdot \cn^{D}.
    \end{equation}
    We use Inequalities~\eqref{inequality1:lambda-close-to-variable}
    and~\eqref{inequality2:lambda-close-to-variable} to bound the values
    of $g_j,\dots,g_{\ell-1}$. 
    By Property~\ref{claim1:lambda-close-to-variable:B}, 
    $\cn^{g_r} \leq
      \frac{\abs{q_{r+1}(\cn)}+\cdots+\abs{q_n(\cn)}}{\abs{Q_r(\cn)}}$, 
    and therefore
    {\allowdisplaybreaks
    \begin{align*}
      g_r\leq{}&
        \log_{\cn}(\abs{q_{r+1}(\cn)}+\cdots+\abs{q_n(\cn)})
        - \log_{\cn}(\abs{Q_j(\cn)})
      \\
      \leq{}& \frac{1}{\ln(\cn)} 
        (\ln(\abs{q_{r+1}(\cn)}+\cdots+\abs{q_n(\cn)})
        - \ln(\abs{Q_r(\cn)}))
      &
      \tag*{change of base}
      \\
      \leq{}& \frac{1}{\ln(\cn)} 
        (\ln(H \cdot D \cdot \cn^{D} \cdot n) 
        - \ln(\abs{Q_r(\cn)}))
      &
      \tag*{by Inequality~\eqref{inequality2:lambda-close-to-variable}}
      \\
      \leq{}& \frac{1}{\ln(\cn)}  
        \Big(\ln( H \cdot D \cdot \cn^{D} \cdot n) 
        + c \cdot \big(D + \ceil{\ln(n H)} 
        + \sum_{s=j}^{r-1} g_s \big)^k \Big)
      &
      \tag*{by~Inequality~\eqref{inequality1:lambda-close-to-variable}}
      \\
      \leq{}& \frac{1}{\ln(\cn)} 
        \Big( D \cdot \ln(\cn) 
        + \ln(nH D)
        + c \big(D + \ceil{\ln(n H)} 
        + \sum_{s=j}^{r-1} g_s \big)^k \Big)\\
      \leq{}& \frac{1}{\ln(\cn)} 
        \Big(D \cdot \ln(\cn) 
        + 2c \big(D + \ceil{\ln(n H)} 
        + \sum_{s=j}^{r-1} g_s \big)^k \Big)
      &
      \hspace{-3.5cm}
      \tag*{as $D + \ceil{\ln(nH)} \geq \ln(nH D)$}
      \\
      \leq{}& \frac{1}{\ln(\cn)}  
        \Big(D \frac{\ln(\cn)}{\ln(1 + \frac{1}{e^c})} 
        + 2 c \cdot \big(D + \ceil{\ln(n H)} 
        + \sum_{s=j}^{r-1} g_s \big)^k \Big)
      &
      \tag*{as $\frac{1}{\ln(1+\frac{1}{e^c})} > 1$}\\
      \leq{}& \frac{1}{\ln(\cn)}  
        \Big(D \frac{\ln(\cn)}{\ln(1 + \frac{1}{e^c})} 
        \cdot 2 c \cdot \big(D + \ceil{\ln(n H)} 
        + \sum_{s=j}^{r-1} g_s \big)^k \Big)
      &
      \tag*{as $\frac{\ln(\cn)}{\ln(1 + \frac{1}{e^c})} \geq 1$}\\[-5pt]
      &&
      \tag*{by~\eqref{inequality1:lower-bound-cn}, and $D \geq 2$}
      \\
      \leq{}& \frac{2cD}{\ln(1 + \frac{1}{e^c})} 
        \big(D + \ceil{\ln(n H)} + \sum_{s=j}^{r-1} g_s \big)^k.
    \end{align*}
    }
      Let us inductively define the following numbers $B_j,\dots,B_\ell$:
      \begin{align*} 
        B_j &\coloneqq 
          \frac{2cD}{\ln(1 + \frac{1}{e^c})} 
          \big(D + \ceil{\ln(n H)} \big)^k\\
        B_r &\coloneqq 
          \frac{2cD}{\ln(1 + \frac{1}{e^c})} 
          \big(D + \ceil{\ln(n H)} + \sum_{s=j}^{r-1} B_s \big)^k
        &\tag*{$\text{for } r \in [j+1..\ell].$}
      \end{align*}
      From the previous inequalities, $g_r \leq B_r$ for every $r \in [j..\ell]$.
      Moreover, observe that, since $\frac{1}{\ln(1+\frac{1}{e^c})} > 1$,  
      for every $r \in [j+1..\ell]$ we have $B_r \geq D + \ceil{\ln(n H)} + \sum_{s=j}^{r-1} B_s$, and therefore $B_\ell \geq \deg(Q_\ell)$.
      We proceed by bounding $B_r$ with respect to 
      $B_{r-1}$:
      \begin{align*}
          B_r & = \frac{2cD}{\ln(1 + \frac{1}{e^c})} 
            \big(D + \ceil{\ln(n H)} + \sum_{s=j}^{r-1} B_s \big)^k
          \\
          & = \frac{2cD}{\ln(1 + \frac{1}{e^c})} 
            \big(B_{r-1} + D + \ceil{\ln(n H)} + \sum_{s=j}^{r-2} B_s \big)^k\\
          &\leq \frac{2cD}{\ln(1 + \frac{1}{e^c})} \big(2 \cdot B_{r-1} \big)^k
          \leq \frac{2^{k+1}cD}{\ln(1 + \frac{1}{e^c})}(B_{r-1})^k.
      \end{align*}
      Let $A \coloneqq \frac{2^{k+1}cD}{\ln(1 + \frac{1}{e^c})}$. 
      Hence, $B_r \leq A \cdot (B_{r-1})^k$ for every $r \in [j+1..\ell]$.
      We show by induction that $B_r \leq A^{\max(r-j,k^{r-j}-1)}B_j^{k^{r-j}}$
      for every $r \in [j..\ell]$.
      \begin{description}
        \item[base case: $r = j$] In this case the inequality is trivially satisfied.
        \item[induction step: $r > j$] We divide the proof depending on whether $k = 1$.
        \begin{itemize}
        \item If $k = 1$, then $\max(r-j,k^{r-j}-1) = r-j$ and we need to prove that $B_r \leq A^{r-j}B_j$. 
        Because $k = 1$, the induction hypothesis simplifies to $B_{r-1} \leq A^{r-1-j}B_j$, and the bound $B_r \leq A \cdot (B_{r-1})^k$ becomes $B_r \leq A \cdot B_{r-1}$. Hence, $B_r \leq A^{r-j}B_j$  follows.

        \item If $k \geq 2$, then $\max(r-j,k^{r-j}-1) = k^{r-j}-1$ and therefore we need to prove that $B_r \leq A^{k^{r-j}-1}B_j^{k^{r-j}}$.
        By induction hypothesis $B_{r-1} \leq A^{\max(r-1-j,k^{r-1-j}-1)}B_j^{k^{r-1-j}}$. Here note that if $r-1=j$ then $r-1-j = 0 = k^{r-1-j}-1$, and otherwise $\max(r-j,k^{r-j}-1) = k^{r-j}-1$; 
        hence $B_{r-1} \leq A^{k^{r-1-j}-1}B_j^{k^{r-1-j}}$.
        Then, 
        \begin{align*}
          B_r &\leq A \cdot (B_{r-1})^k\\
          & \leq A \cdot (A^{k^{r-1-j}-1}B_j^{k^{r-1-j}})^k
          &\tag*{by induction hypothesis}\\
          & = A^{k^{r-j}-k+1}B_j^{k^{r-j}}\\
          & \leq A^{k^{r-j}-1}B_j^{k^{r-j}}
          &\tag*{since $k \geq 2$.}
        \end{align*} 
        \end{itemize}
      \end{description}
      We can now compute the aforementioned bound on $\deg(Q_\ell)$:
      \begin{align*}
          \hspace{-0.23cm} \deg(Q_\ell) \leq{}& B_\ell
          \leq{} A^{\max(n,k^{n}-1)}B_j^{k^{n}}
          &\tag*{remark: $\ell - j < n$}
          \\
          \leq{}& 
            \left(\frac{2^{k+1}cD}{\ln(1 + \frac{1}{e^c})}\right)^{\max(n,k^{n}-1)}
          \cdot
            \left(\frac{2cD}{\ln(1 + \frac{1}{e^c})} 
            \big(D + \ceil{\ln(n H)} \big)^k\right)^{k^n}
          &\tag*{def.~of $A$ and $B_j$}
          \\
          \leq{}&  
          2^{(k+1)(n+k^n)+k^n}\left(\frac{cD}{\ln(1 + \frac{1}{e^c})}\right)^{n+2k^n} \hspace{-7pt}\big(D + \ceil{\ln(n H)} \big)^{k^{n+1}}\\
          \leq{}& 
          2^{(k+1)(n+k^n)+k^n}\left(\frac{c}{\ln(1 + \frac{1}{e^c})}\right)^{n+2k^n} \hspace{-9pt} D^{n+2k^n+k^{n+1}}\ln(n H)^{1+k^{n+1}}
          \\
          &&\tag*{since $D \geq 2$ and $H \geq 8$}\\
          \leq{}& 
          \left(\frac{2cD\ln(H)}{\ln(1 + \frac{1}{e^c})}\right)^{5nk^{n+1}} 
          &
          \tag*{since $\ln(nH) \leq \ln(H)^{2n}$,}\\[-7pt]
          &&\tag*{and then all exponents are bounded by $5nk^{n+1}$.}
      \end{align*}
      This concludes the proof of the claim.
  \end{proof}

    We are now ready to derive an explicit characterisation for the set $G$.
    Consider the sets $\mathcal{Q}$ and $B$ defined during the proof of the first statement of the lemma. In particular,
    \[ 
        B \coloneqq \left\{ \beta \in \Z : \text{there is } Q \in \mathcal{Q} \text{ such that } \cn^{\beta} \in \Big\{{\textstyle\lambda(Q(\cn)), \frac{\lambda(Q(\cn) \cdot (\cn-1))}{\cn}, \frac{\lambda(Q(\cn) \cdot (\cn+1))}{\cn}}\Big\}
        \right\}.
    \]
    and $G$ can be set to be any finite set satisfying $[\min B..\max B] \subseteq G$. We also recall that every polynomial $Q$ in the set $\mathcal{Q}$ is such that the numbers $Q(\cn)$, $Q(\cn) \cdot (\cn-1)$ and $Q(\cn) \cdot (\cn+1)$ are all strictly positive; and so, in particular, for these numbers $\lambda$ is well-defined.
    By definition of $\mathcal{Q}$ and from Claims~\ref{claim2:lambda-close-to-variable} 
    and~\ref{claim3:lambda-close-to-variable}, we deduce that the heights and degrees of the univariate polynomials $Q$, $Q \cdot (x-1)$ and $Q \cdot (x+1)$ are bounded as follows: 
    \begin{align*}
      \height(Q) \leq n \cdot H,
      &&\height(Q(x-1)) \leq 2n \cdot H, 
      &&\height(Q(x+1)) \leq 2n \cdot H,\\ 
      \deg(Q) \leq E,
      &&\deg(Q(x-1)) \leq E+1, 
      &&\deg(Q(x+1)) \leq E+1,
    \end{align*}
    where $E \coloneqq \left(\frac{2cD\ln(H)}{\ln(1 + \frac{1}{e^c})}\right)^{5nk^{n+1}}$. Let $P$ be a number among $Q(\cn)$, $Q(\cn) \cdot (\cn-1)$ and $Q(\cn) \cdot (\cn+1)$.
    An upper bound to $P$ is given by 
    \begin{equation*} 
      P \leq 2n \cdot H \cdot (E+2) \cdot \cn^{E+2},
    \end{equation*}
    whereas a lower bound follows by relying on the root barrier $\sigma$: 
    \begin{equation*} 
      P \geq \frac{1}{e^{c(E+1+\ceil{\ln(2nH)})^k}}.
    \end{equation*}
    Recall that, for every $\alpha > 0$, the definition of $\lambda$ implies $\frac{\alpha}{\cn} < \lambda(\alpha) \leq \alpha$.
    We conclude that, for every integer $\beta \in B$, 
    \begin{equation*} 
      \frac{1}{\cn^2 \cdot e^{c(E+1+\ceil{\ln(2nH)})^k}} \leq \cn^{\beta} \text{ \ and \ } \cn^{\beta} \leq 2n \cdot H \cdot (E+2) \cdot \cn^{E+2}.
    \end{equation*}
    Applying the logarithm base $e$ to both inequalities shows: 
    \begin{equation*} 
      -\ln(\cn^2 \cdot e^{c(E+1+\ceil{\ln(2nH)})^k}) \leq \beta \cdot \ln(\cn) \text{ \ and \ } \beta \cdot \ln(\cn) \leq \ln(2n \cdot H \cdot (E+2) \cdot \cn^{E+2}).
    \end{equation*}
    This implies that taking $G$ to be the interval 
    $[-\frac{\ln(\cn^2 \cdot e^{c(E+1+\ceil{\ln(2nH)})^k})}{\ln(\cn)}..
      \frac{\ln(2n \cdot H \cdot (E+2) \cdot \cn^{E+2})}{\ln(\cn)}]$
    suffices. In the statement of the lemma we provide however a slightly larger set with an easier-to-digest bound, that is, $[-L..L]$, where $L \coloneqq \left(2^{3c}D\ceil{\ln(H)}\right)^{6nk^{3n}}$.  
    To conclude the proof, below we show that
    $[-\frac{\ln(\cn^2 \cdot e^{c(E+1+\ceil{\ln(2nH)})^k})}{\ln(\cn)}..
    \frac{\ln(2n \cdot H \cdot (E+2) \cdot \cn^{E+2})}{\ln(\cn)}] \subseteq [-L..L]$.

    \begin{description}
      \item[upper bound] We show that $\frac{\ln(2n \cdot H \cdot (E+2) \cdot \cn^{E+2})}{\ln(\cn)} \leq L$:
      {\allowdisplaybreaks
      \begin{align*} 
        &  \frac{\ln(2n \cdot H \cdot (E+2) \cdot \cn^{E+2})}{\ln(\cn)}\\ 
        \leq{}& \frac{\ln(2n \cdot H \cdot (E+2))}{\ln(1+\frac{1}{e^c})} + 2 + E
        & \tag*{by properties of $\ln$, and $\cn \geq 1 + \frac{1}{e^c} $}\\
        \leq{}& \frac{2 \cdot \ln(2n \cdot H)}{\ln(1+\frac{1}{e^c})} + \frac{\ln(E+2)}{\ln(1+\frac{1}{e^c})} + E
        & \tag*{we have $\frac{\ln(2n \cdot H)}{\ln(1+\frac{1}{e^c})} \geq 2$}\\
        \leq{}& 2 \cdot E + \frac{\ln(E+2)}{\ln(1+\frac{1}{e^c})}
        & \tag*{we have $\frac{2 \cdot \ln(2n \cdot H)}{\ln(1+\frac{1}{e^c})} \leq E$}\\
        \leq{}& 2 \cdot E + \frac{E}{\ln(1+\frac{1}{e^c})}
        & \tag*{we have $E \geq \ln(E+2)$ since $E \geq 2$}\\
        \leq{}& \frac{3 \cdot E}{\ln(1+\frac{1}{e^c})}
        & \tag*{as $\frac{1}{\ln(1+\frac{1}{e^c})} \geq 1$}\\
        \leq{}& \frac{3}{\ln(1+\frac{1}{e^c})} \left(\frac{2cD\ln(H)}{\ln(1 + \frac{1}{e^c})}\right)^{5nk^{n+1}}
        & \tag*{def.~of~$E$}\\
        \leq{}& \left(\frac{2cD\ln(H)}{\ln(1 + \frac{1}{e^c})}\right)^{6nk^{n+1}}
        & \tag*{as $\frac{2cD\ln(H)}{\ln(1 + \frac{1}{e^c})} \geq \frac{3}{\ln(1+\frac{1}{e^c})}$}\\
        \leq{}& \left(2c \cdot 2^{2c}D\ln(H)\right)^{6nk^{n+1}}
        & \tag*{as $\frac{1}{\ln(1 + \frac{1}{e^c})} \leq 2^{2c}$}\\
        \leq{}& \left(2^{3c}D\ln(H)\right)^{6nk^{n+1}} \leq{} L
        & \tag*{since $2c \leq 2^c$ and by def.~of $L$}
      \end{align*}
      }
      {\allowdisplaybreaks
      \item[lower bound] We show that $\frac{\ln(\cn^2 \cdot e^{c(E+1+\ceil{\ln(2nH)})^k})}{\ln(\cn)} \leq L$ (so, $-L \leq -\frac{\ln(\cn^2 \cdot e^{c(E+1+\ceil{\ln(2nH)})^k})}{\ln(\cn)}$):
      \begin{align*} 
        & \frac{\ln(\cn^2 \cdot e^{c(E+1+\ceil{\ln(2nH)})^k})}{\ln(\cn)}\\
        \leq{}& 2 + \frac{c(E+1+\ceil{\ln(2nH)})^k}{\ln(\cn)} 
        & \tag*{by properties of $\ln$}\\
        \leq{}& 2 + \frac{c(E+1+\ceil{\ln(2nH)})^k}{\ln(1+\frac{1}{e^c})} 
        & \tag*{since $\cn \geq 1 + \frac{1}{e^c}$}\\
        \leq{}& 2 + \frac{c(E+2+\ln(H)^{4n})^k}{\ln(1+\frac{1}{e^c})} 
        & \tag*{as $\ceil{\ln(2nH)} \leq 1 + \ln(H)^{4n}$}\\
        \leq{}& 2 + \frac{c(2E)^k}{\ln(1+\frac{1}{e^c})} 
        & \tag*{since $E \geq 2 + \ln(H)^{4n}$}\\
        \leq{}& \frac{2^{k+1}c}{\ln(1+\frac{1}{e^c})} \cdot E^k 
        & \tag*{since $\frac{c(2E)^k}{\ln(1+\frac{1}{e^c})} \geq 2$}\\
        \leq{}& \frac{2^{k+1}c}{\ln(1+\frac{1}{e^c})}\left(\frac{2cD\ln(H)}{\ln(1 + \frac{1}{e^c})}\right)^{5nk^{n+2}}
        & \tag*{def.~of~$E$}\\
        \leq{}& \left(\frac{2cD\ln(H)}{\ln(1 + \frac{1}{e^c})}\right)^{6nk^{n+2}}
        & \tag*{note: $D \geq 2$}\\
        \leq{}& \left(2^{3c}D\ln(H)\right)^{6nk^{n+2}} \leq{} L
        & 
        \tag*{$\text{as in the previous case, } \frac{2c}{\ln(1+\frac{1}{e^c})} \leq 2^{3c}$}
      \end{align*}
      }
    \end{description}
    This concludes the proof of \Cref{lemma:lambda-close-to-variable}.
  \end{proof}

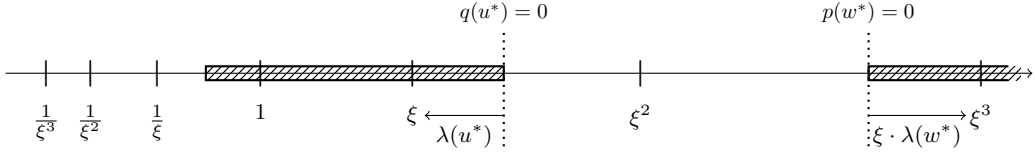
\begin{figure}
  \begin{tikzpicture}[scale=4,minimum size=2pt]
    \node[label=-90:$\frac{1}{\cn^3}$] (m3) at (0.295,0) {\textbf{|}};
    \node[label=-90:$\frac{1}{\cn^2}$] (m2) at (0.44,0) {\textbf{|}};
    \node[label=-90:$\frac{1}{\cn}$] (m1) at (0.66,0) {\textbf{|}};
    \node[label=-90:{\footnotesize$1$}] (o) at (1,0) {\textbf{|}};
    \node[label=-90:{\footnotesize$\cn$}] (p1) at (1.5,0) {\textbf{|}};
    \node[label=-90:{\footnotesize$\cn^2$}] (p2) at (2.25,0) {\textbf{|}};
    \node[label=-90:{\footnotesize$\cn^3$}] (p3) at (3.37,0) {\textbf{|}};

    \node (l) at (0.13,0) {}; \node (r) at (3.57,0) {}; \draw[->] (l) -- (r);

    \fill[draw=black,thick,pattern=north east lines] (0.82,-0.023) rectangle
    (1.8,0.023); \fill[thick,pattern=north east lines] (3,-0.023) rectangle
    (3.5,0.023); \draw[thick] (3.46,0.023) -- (3,0.023) -- (3,-0.023) --
    (3.46,-0.023);

    \node (z1) at (1.8,0.2) {\scalebox{0.8}{$q(u^*) = 0$}}; 
    \node (z2) at (3,0.2) {\scalebox{0.8}{$p(w^*) = 0$}}; 
    \draw[dotted,thick] (z1) -- (1.8,-0.25); 
    \draw[dotted,thick] (z2) -- (3,-0.25);

    \draw (1.8,-0.14) edge[->] node[below]{\footnotesize{$\lambda(u^*)$}}
    (1.54,-0.14); \draw (3,-0.14) edge[->] node[below]{\footnotesize{$\cn \cdot
    \lambda(w^*)$}} (3.32,-0.14);

    \node at (1.5,-0.32) {};
  \end{tikzpicture}
  \caption{High-level idea of the quantifier elimination procedure. Dashed
  rectangles are intervals corresponding to the set of solutions over $\R$ of a
  (univariate) formula~$\phi$. To search for a solution over~$\ipow{\cn}$, it
  suffices to look for elements of $\ipow{\cn}$ that are close to the endpoints
  of these intervals. At each endpoint, a polynomial in $\phi$ must evaluate to
  zero (since around endpoints the truth of $\phi$ changes), so it suffices to
  look for integer powers of $\cn$ that are close to roots of polynomials in
  $\phi$.}
  \label{figure:qe-idea}

\end{figure}

We now give the high-level idea of the quantifier elimination procedure, which
is also depicted in~\Cref{figure:qe-idea}. Let $\psi(u,\vec y)$ be a
quantifier-free formula of $\exists\ipow{\cn}$, and $u$ be the variable we
want to eliminate. Suppose to evaluate the variables $\vec y$ with elements in
$\ipow{\cn}$, hence obtaining a univariate formula $\phi(u)$. The set of all
solutions \emph{over the reals} of $\phi(u)$ can be decomposed into a finite set
of disjoint intervals. (This follows from the o-minimality of the FO
theory of the reals~\cite[Chapter~3.3]{marker2002model}.) \Cref{figure:qe-idea}
shows these intervals as dashed rectangles. Around the endpoints of these
intervals the truth of $\phi$ changes, and therefore for each such
endpoint~$u^*$ there must be a non-constant polynomial in $\phi$ such that
$q(u^*) = 0$. If an interval with endpoint $u^* \in \R_{>0}$ contains an element
of $\ipow{\cn}$, then it contains one that is ``close'' to $u^*$:
\begin{itemize} 
  \item If $u^* \in \R_{>0}$ is the \emph{right endpoint} of an interval, at
  least one among $\lambda(u^*)$ and $\cn^{-1} \cdot \lambda(u^*)$ belongs to
  the interval. The first case is depicted in \Cref{figure:qe-idea}. The latter
  case occurs when~$u^*$ belongs to $\ipow{\cn}$ but not to the interval.
  \item If $u^*$ is the \emph{left endpoint} of an interval, then $\cn
  \cdot \lambda(u^*)$ or $\lambda(u^*)$ belongs to the interval.
  The latter case occurs when $u^*$ belongs to $\ipow{\cn}$ and also to the interval.
\end{itemize}
Note that we have restricted the endpoint $u^*$ to be
positive, so that $\lambda(u^*)$ is well-defined. The only case where we may not
find such an endpoint is when $\phi(u)$ is true for every $u > 0$.
But finding an element of $\ipow{\cn}$ is in this case simple: we
can just pick~$1 \in \ipow{\cn}$. Since $u^*$ is positive, we can split it into $x^* \cdot v^*$
with $x^* \in \ipow{\cn}$ and $1 \leq v^* < \cn$ (so, $\lambda(u^*) = x^*$). To
obtain quantifier elimination, our goal is then to characterise, symbolically as
a finite set of polynomials~$\tau(\vec{y})$, the set of all possible values for
$x^*$. In this way, we will be able to eliminate the variable~$u$ by considering the
polynomials $\cn^{-1} \cdot \tau(\vec y)$, $\tau(\vec y)$ and $\cn \cdot \tau(\vec{y})$ representing the
integer powers of $\cn$ that are ``close'' to endpoints. The following lemma
provides a first step towards the required characterisation given in 
the subsequent \Cref{lemma:last-lambda-lemma}.

\begin{restatable}{lemma}{LemmaAuxOne}\label{aux-lemma:one}
  Fix $\cn > 1$.
  Let $r(x,\vec y) \coloneqq \sum_{i=0}^n p_i(\cn,\vec y) \cdot x^i$, where 
  each $p_i(z,\vec y)$ is an integer polynomial in variables $\vec y$ and $z$.
  Then, the formula
  \begin{equation*}
      \ipow{\cn}(x) \wedge r(x,\vec y)=0 \land \Big(\bigvee_{i=0}^n p_i(\cn,\vec y) \neq 0\Big) \implies 
      \bigvee_{\ell = 1}^m \theta_\ell (x, \vec y)\,,
  \end{equation*}
  is a tautology of $\exists\R(\ipow{\cn})$, where each $\theta_\ell$ is a formula 
  of the form either 
  \begin{align*} 
    & x^{k-j} = \textstyle\frac{\cn^s \cdot \lambda(-p_j(\cn,\vec y))}
    {\lambda(p_k(\cn,\vec y))} \land p_j(\cn,\vec y) < 0 \land p_k(\cn,\vec y) > 0 \qquad\text{or}\\
    & x^{k-j}= \textstyle\frac{\cn^s \cdot \lambda(p_j(\cn,\vec y))}{\lambda(-p_k(\cn,\vec y))} \land p_j(\cn,\vec y) > 0 \land p_k(\cn,\vec y) < 0\,,
  \end{align*}
  with $0\leq j < k \leq n$, $s \in [-g..g]$ with $g \coloneqq 1+\ceil{\log_{\cn}(n)}$,
  and $m\leq n^2\cdot \left(2 \cdot \ceil{\log_{\cn}(n)} + 3\right)$.
\end{restatable}

\begin{proof}
  The proof follows somewhat closely the arguments in~\cite[Lemmas 3.9 and
  3.10]{AvigadY07}. Observe that the lemma is trivially true for $n = 0$,
  as in this case the antecedent of the implication is false
  (from the formulae $r(x,v,\vec y)=0$ and 
  $p_0(\cn,\vec y) \neq 0$). Below, assume $n \geq 1$.

  Pick $x \in \R$ and $\vec y \in R$ making 
  the antecedent of the implication of the formula true, 
  that is, 
  we have $\ipow{\cn}(x)$, $r(x,\vec y) = 0$ and $p_i(\cn, \vec y) \neq 0$
  for some $i \in [0..n]$. We
  show that $x$ and $\vec y$ satisfy one of the formulae
  $\theta_1,\dots,\theta_m$.
  
  We can write $r(x,\vec y)$ as $p_k(\cn,\vec y) \cdot x^k + p_j(\cn, \vec y)
  \cdot x^j + r^*(x,\vec y)$ where $p_k(\cn, \vec y) \cdot x^k$ and $p_j(\cn, \vec y) \cdot x^j$ are
  respectively the largest and smallest monomial in $r(x,\vec y)$, and $r^*(x, \vec y)$
  is the sum of all the other monomials.
  Since we are assuming $r(x,\vec y) = 0$ and $p_i(\cn,\vec y) \neq 0$, we
  conclude that $p_k(\cn,\vec y) \cdot x^k > 0$ and $p_j(\cn, \vec y) \cdot x^j < 0$. This also entails
  that $k \neq j$. We have,
  \begin{align} 
    \frac{p_k(\cn,\vec y) \cdot x^k}{\cn}
    &< p_k(\cn,\vec y) \cdot x^k - r(x,\vec y) 
    &\tag*{since $\cn > 1$ and $r(x,\vec y) = 0$} \notag\\
    &= -p_j(\cn,\vec y) \cdot x^j - r^*(x,\vec y)
    &\tag*{by def.~of~$p_j$, $p_k$ and $r^*$} \notag\\
    &\leq -n \cdot p_j(\cn,\vec y) \cdot x^j
    &\text{by def.~of~$p_j$, $p_k$ and $r^*$.}
    \label{aux-lemma-1:eq1}
  \end{align}
  Observe that $\ipow{\cn}(x)$ implies $x > 0$, 
  and therefore $p_j(\cn,\vec y) \cdot x^j < 0$ implies $p_j(\cn,\vec y) < 0$.
  From \Cref{aux-lemma-1:eq1} we then obtain $\frac{p_k(\cn,\vec y)}{-p_j(\cn,\vec y)} x^{k-j} \leq n \cdot \cn$.
  Moreover, from $r(x,\vec y) = 0$ we have $-p_j(\cn,\vec y) \cdot x^j \leq n \cdot p_k(\cn,\vec y) \cdot x^k$, i.e., $\frac{1}{n} 
  \leq \frac{p_k(\cn,\vec y)}{-p_j(\cn, \vec y)} \cdot x^{k-j}$, and therefore 
  \[ 
    0 < \cn^{-\ceil{\log_{\cn}(n)}} \leq \cn^{-\log_\cn(n)} = \frac{1}{n} 
    \leq \frac{p_k(\cn,\vec y)}{-p_j(\cn, \vec y)} \cdot x^{k-j} \leq n \cdot \cn = 
    \cn^{1+\log_\cn(n)} \leq \cn^{1+\ceil{\log_{\cn}(n)}}\,.
  \]
  The above chain of inequalities shows that $\frac{-p_j(\cn,\vec y)}{p_k(\cn,\vec y)} \cdot
  \cn^{-\ceil{\log_{\cn}(n)}} \leq x^{k-j} \leq \frac{-p_j(\cn,\vec y)}{p_k(\cn,\vec y)} \cdot
  \cn^{\ceil{\log_{\cn}(n)}+1}$. 
  Since $\lambda$ is a monotonic function, this implies 
  \begin{align*}
    &\lambda\left(\frac{-p_j(\cn,\vec y)}{p_k(\cn,\vec y)} \cdot \cn^{-\ceil{\log_{\cn}(n)}}\right) 
    \leq \lambda(x^{k-j}) \leq \lambda\left(\frac{-p_j(\cn,\vec y)}{p_k(\cn,\vec y)} 
    \cdot \cn^{1+\ceil{\log_{\cn}(n)}}\right).
  \end{align*}
  By definition of $\lambda$, for every $a \in \R$ we have $\frac{a}{\cn} \leq \lambda(a) \leq a$. 
  Moreover, since $x$ is an integer power of $\cn$, $x^{k-j} = \lambda(x^{k-j})$, and therefore the above inequalities can entail 
  \begin{align*}
    &\frac{\lambda(-p_j(\cn,\vec y))}{\lambda(p_k(\cn,\vec y))} \cdot \cn^{-(1+\ceil{\log_{\cn}(n)})} 
    \leq x^{k-j} \leq \frac{\lambda(-p_j(\cn,\vec y))}{\lambda(p_k(\cn,\vec y))} 
    \cdot \cn^{\ceil{\log_{\cn}(n)}+1}
  \end{align*} 
  Let $g \coloneqq
  1+\ceil{\log_{\cn}(n)}$. 
  We conclude that $x^{k-j} = \cn^s \cdot
  \frac{\lambda(-p_j)}{\lambda(p_k)}$, for some integer $s \in [-g..g]$.
  To conclude the proof we analyse two cases, depending on whether $k-j > 0$ (recall: $k\neq j$).
  
  \begin{description}
    \item[case $k-j > 0$] 
      We have  $x^{k-j} = \cn^s \cdot
      \frac{\lambda(-p_j(\cn, \vec y))}{\lambda(p_k(\cn, \vec y))}$, with
      $p_k(\cn,\vec y) > 0$ and $p_j(\cn, \vec y) < 0$. We have thus obtained the first of the two forms in the statement of the lemma.
    \item[case $k-j < 0$] 
      We have $x^{j-k} = \cn^{-s} \cdot \frac{\lambda(p_k(\cn, \vec y))}{\lambda(-p_j(\cn, \vec y))}$
      with $p_j(\cn, \vec y) < 0$ and ${p_k(\cn, \vec y) > 0}$. This corresponds to the second of the two forms in the statement of the
      lemma. For convenience, in the statement we have swapped the symbols $j$ and $k$, and wrote $s$ instead of $-s$ (since both these integers belongs to $[-g..g]$).
      \qedhere
  \end{description}
\end{proof}

\begin{restatable}{lemma}{LemmaLastLambdaLemma}
  \label{lemma:last-lambda-lemma}
  Let $r(x,v,\vec y) \coloneqq \sum_{i=0}^n p_i(\cn,\vec y) \cdot (x \cdot v)^i$, where 
  each~$p_i$ is an integer polynomial,
  $M$ be the set of monomials $\vec y^{\vec \ell}$ occurring in some $p_i$, 
  and ${N \coloneqq \{\vec y^{\vec \ell_1 - \vec \ell_2} : \vec y^{\vec \ell_1},\vec y^{\vec \ell_2} \in M\}}$.~Then,%
  \vspace{-3pt}
  \begin{align*}
      \hspace{-10pt}
      \ipow{\cn}(x) 
      \land 1 \leq v < \cn 
      \land r(x,v,\vec y)=0 
      \land \big(\!\bigvee_{i=0}^n p_i(\cn,\vec y) \neq 0\big)
      \land \bigwedge_{\mathclap{y \textup{\,from\,} \vec y}} \ipow{\cn}(y)
      \ \models\ 
      \hspace{-2pt}
      \bigvee_{(j,g,\vec y^{\vec \ell}) \in F}
      \hspace{-5pt}
      x^j = \cn^{g} \cdot \vec y^{\vec \ell}
  \end{align*}
  holds (in the theory~$\exists\R(\ipow{\cn})$) for some finite set $F \subseteq [1..n] \times \Z \times N$.
  Moreover: 
  \begin{enumerate}[label=\Roman*., ref=\Roman*]
    \item\label{lemma:last-lambda-lemma:i1} If $\cn$ is a computable transcendental number, there is an algorithm computing $F$ from $r$. 
    \item\label{lemma:last-lambda-lemma:i2} If $\cn$ has a root barrier $\sigma(d,h) \coloneqq c \cdot
            (d+\ceil{\ln(h)})^k$, for some $c,k \in \N_{\geq 1}$, then,
            \begin{equation*}
              F \coloneqq [1..n] \times \left[-L..L\right] \times N,
              \qquad\text{where }
              L \coloneqq n\left(2^{4c}D\ceil{\ln(H)}\right)^{6\abs{M} \cdot k^{3\abs{M}}},
            \end{equation*}
              with $H \coloneqq \max\{8,\height(p_i) : i \in [1..n]\}$, 
              and
              $D \coloneqq \max\{\deg(\cn,p_i)+2 : i \in [0..n]\}$.
  \end{enumerate}

\end{restatable}

\begin{proof}
  The lemma is trivially true for $n = 0$,
  as in this case the premise of the entailment is equivalent to $\bot$
  (from the formulae $r(x,v,\vec y)=0$ and 
  $p_0(\cn,\vec y) \neq 0$). Below, assume~$n \geq 1$.

  We start by showing the existence of the finite set $F$ (first statement of the lemma).
  Assume the premise of the entailment 
  of the lemma,~i.e., 
  \begin{equation}
    \label{last-lambda-antecedent}
    \ipow{\cn}(x) 
      \land 1 \leq v < \cn 
      \land \Big(\bigwedge_{i=1}^d \ipow{\cn}(y_i)\Big)
      \land r(x,v,\vec y)=0 
      \land \Big(\bigvee_{i=0}^n p_i(\cn,\vec y) \neq 0\Big)\,,
  \end{equation}
  to be satisfied. 
  We view $r(x,v,\vec y)$ as a polynomial in the 
  variable $x$ with coefficients of the form $p_i(\cn,\vec y) \cdot v^i$.
  By applying~Lemma~\ref{aux-lemma:one}, 
  we deduce that the above formula
  entails a finite disjunction $\bigvee_{u=1}^m \theta_u(x,v,\vec y)$ 
  where the formulae $\theta_u(x,v,\vec y)$ are of the form 
  \begin{equation}\label{last-lambda-eq1}
    x^\mu=\frac{\cn^s\cdot \lambda(\pm p_j(\cn,\vec y)v^j)}
    {\lambda(\mp p_w(\cn,\vec y)v^{w})}
    \land \pm p_j(\cn,\vec y)v^j > 0 
    \land \mp p_w(\cn,\vec y)v^w > 0,
  \end{equation}
  where $\mu\in[1..n]$ and $j,w \in [0..n]$ with $j\neq w$.
  Moreover, the number $m$ of disjuncts $\theta_u$ is bounded 
  by $n^2\cdot (2 \cdot \ceil{\log_{\cn}(n)} + 3)$, 
  and $s \in [-(1+\ceil{\log_{\cn}(n)})..(1+\ceil{\log_{\cn}(n)})]$.

  By definition of $\lambda$, for every $a,b \in \R$, either $\lambda(a \cdot b) = \lambda(a) \cdot \lambda(b)$ or $\lambda(a \cdot b) = \cn \cdot \lambda(a) \cdot \lambda(b)$.
  Moreover, $v^j$ and $v^w$ are positive numbers, and therefore Formula~\eqref{last-lambda-eq1}
  is equivalent to 
  \begin{equation}\label{last-lambda-eq2}
    \bigvee_{t \in \{-1,0,1\}} x^\mu=\frac{\cn^{s+t}\cdot \lambda(\pm p_j(\cn,\vec y)) \cdot \lambda(v^j)}
    {\lambda(\mp p_w(\cn,\vec y)) \cdot \lambda(v^w)}
    \land \pm p_j(\cn,\vec y) > 0 
    \land \mp p_w(\cn,\vec y) > 0.
  \end{equation}
  Next, we bound the terms $\lambda(v^j)$ and $\lambda(v^w)$. 
  Since Formula~\eqref{last-lambda-antecedent} asserts $1 \leq v < \cn$,
  we have $\lambda(v^j) = \cn^{\alpha}$ for some $\alpha \in [0,j-1]$; 
  and similarly $\lambda(v^w) = \cn^{\beta}$ for some $\alpha \in [0,w-1]$.
  Given that $j$ and $w$ belong to $[0..n]$, 
  Formula~\eqref{last-lambda-eq2} (or, equivalently, Formula~\eqref{last-lambda-eq1}) then entails
  \begin{equation}\label{last-lambda-eq3}
    \bigvee_{t=-n}^n
    x^\mu=\frac{\cn^{s+t} \cdot \lambda(\pm p_j(\cn,\vec y))}
    {\lambda(\mp p_w(\cn,\vec y))}
    \land \pm p_j(\cn,\vec y) > 0 
    \land \mp p_w(\cn,\vec y) > 0.
  \end{equation}
  Let $p(\cn,\vec y)$ be a polynomial among $\pm p_j(\cn,\vec y)$ and $\mp p_w(\cn,\vec y)$.
  This polynomial can be seen as having variables in $\vec y$, and having as coefficients polynomial expressions in $\cn$, that is,
  \begin{equation*}
    p(\cn,\vec y)= \sum_{\ell=1}^{\abs{M}} q_{\ell}(\cn) \cdot \vec{y}^{\vec{e}_\ell}.
  \end{equation*}
  where each $\vec{y}^{\vec{e}_\ell}$ is a monomial from $M$.
  Since Formula~\eqref{last-lambda-antecedent} asserts that every variable in $\vec y$ 
  is an integer power of $\cn$, given $\ell \in [1,\abs{M}]$ we can introduce an integer variable $z_\ell$ and set $\cn^{z_\ell} = \vec{y}^{\vec{e}_\ell}$. That is, Formula~\eqref{last-lambda-antecedent} entails the following formula of $\R(\ipow{\cn})$
  \begin{equation}\label{last-lambda-eq4}
    p(\cn, \vec y) > 0 
    \iff \exists z_1 \dots z_{\abs{M}} \in \Z \,
      \Big( \sum_{\ell=1}^{\abs{M}} q_{\ell}(\cn) \cdot \cn^{z_\ell} > 0 \land \bigwedge_{\ell=1}^{\abs{M}} \cn^{z_\ell} = \vec{y}^{\vec{e}_\ell} \Big).
  \end{equation}
  We apply Lemma~\ref{lemma:lambda-close-to-variable} on $\sum_{\ell=1}^{\abs{M}} q_{\ell}(\cn) \cdot \cn^{z_\ell} > 0$: there is a finite set $G_p$ such that 
  \begin{equation*}
    \sum_{\ell=1}^{\abs{M}} q_{\ell}(\cn) \cdot \cn^{z_\ell} > 0
    \implies 
    \bigvee_{g \in G_p} \bigvee_{\ell=1}^{\abs{M}}\lambda\Big(\sum_{\ell=1}^{\abs{M}} q_{\ell}(\cn) \cdot \cn^{z_\ell}\Big)=
    \cn^{g} \cdot \cn^{z_\ell}.
  \end{equation*}
  Then, by Formula~\eqref{last-lambda-eq4}, substituting $\cn^{z_\ell}$ for $\vec{y}^{\vec{e}_\ell}$ 
  we obtain 
  \begin{equation}\label{last-lambda-eq5}
    p(\cn, \vec y) > 0 
    \implies \bigvee_{g \in G_p} \bigvee_{\vec y^{\vec \ell} \in M}\lambda\left(p(\cn,\vec y)\right)=
    \cn^{g} \cdot \vec y^{\vec \ell}.
  \end{equation}
  From Formulas~\eqref{last-lambda-eq3} and~\eqref{last-lambda-eq5}, 
  we conclude that Formula~\eqref{last-lambda-eq1} entails
  \begin{equation}\label{last-lambda-eq6}
    \bigvee_{t=-n}^n \bigvee_{(g_1,g_2) \in \left(G_{\pm p_j} \times G_{\mp p_w}\right)} 
    \bigvee_{\vec y^{\vec \ell_1}, \vec y^{\vec \ell_2} \in M}
    \vec x^\mu = \cn^{s+t+g_1-g_2} \cdot \vec y^{\vec \ell_1-\vec \ell_2}.
  \end{equation}
  Above, note that $\vec y^{\vec \ell_1-\vec \ell_2}$ belongs to the set $N$ in the statement of the lemma.
  Since Formula~\eqref{last-lambda-antecedent} entails a finite disjunction of formulae 
  of the form shown in~\eqref{last-lambda-eq1}, 
  and the disjunctions in Formula~\eqref{last-lambda-eq6} are over finite sets, 
  this completes the proof of the first statement of the lemma.
  In particular, one can take as $F$ the set 
  \[ 
    F \coloneqq [1..n] \times [-L..L] \times N
  \]
  where $L \coloneqq \max \{ 1 + \ceil{\log_{\cn}(n)} + n + 2 \abs{g'} : g' \in G_{\pm p_j} \text{ for some } j \in [0..n]\}$.

  We move to the second part of the lemma, which adds further assumptions
  on~$\cn$.

  \paragraph{\textit{Case: $\cn$ is a computable transcendental number \textup{(}Item~\eqref{lemma:last-lambda-lemma:i1}\textup{)}}}
  From~\Cref{lemma:lambda-close-to-variable} item \eqref{lemma:lambda-close-to-variable:i1}, we conclude that the sets $G_{\pm
  p_j}$ can be computed. Moreover, $\ceil{\log_{\cn}(n)}$ can be computed by
  iterating through the natural numbers, finding $\alpha \in \N$ such that
  $\cn^{\alpha-1} < n \leq \cn^{\alpha}$. Checking these inequalities can be
  done by opportunely iterating the algorithm for the sign evaluation problem for transcendental numbers already 
  discussed in the proof of~\Cref{lemma:lambda-close-to-variable}.
  \paragraph{\textit{Case: $\cn$ has a polynomial root barrier \textup{(}Item~\eqref{lemma:last-lambda-lemma:i2}\textup{)}}}
    Assume now $\cn$ to have a polynomial root barrier 
    $\sigma(d,h) \coloneqq c \cdot (d+\ceil{\ln(h)})^k$, 
    with $c,k \in \N_{\geq 1}$. 
    We provide an explicit upper bound to the set $L$ defined above, 
    so that replacing $L$ with this upper bound in the definition of $F$ 
    yield the last statement of the lemma.
    For this, it suffices to upper bound $\ceil{\log_{\cn}(n)}$ 
    as well as $\abs{g'}$, where $g' \in G_{\pm p_j}$ with $j \in [0..n]$.
  
    For the upper bound to $\ceil{\log_{\cn}(n)}$, 
    as done in the proof of Lemma~\ref{lemma:lambda-close-to-variable}, 
    we can consider the polynomial $x-1$ in order to obtain a lower bound on the number $\cn > 1$ 
    via the root barrier $\sigma$. We have $\cn \geq 1 + \frac{1}{e^c}$.
    Then, 
    \begin{align}
      \ceil{\log_{\cn}(n)} 
      & \leq \ceil{\frac{\ln(n)}{\ln(1 + \frac{1}{e^c})}} 
      & \tag*{changing the base of logarithm, and $\cn \geq 1 + \frac{1}{e^c}$} \notag\\
      & \leq \ceil{\ln(n) \cdot 2^{2c}} 
      & \tag*{since $\frac{1}{\ln(1 + \frac{1}{e^c})} \leq 2^{2c}$} \notag\\
      & \leq 2^{2c} \ceil{\ln(n)}. \label{last-lambda-eq7}
    \end{align}
    
    For the bound on the elements in $G_{\pm p_j}$, 
    recall that this set has been computed following 
    Lemma~\ref{lemma:lambda-close-to-variable}.
    The polynomial $\pm p_j$ 
    has the form $\sum_{\ell=1}^{\abs{M}} q_{\ell}(\cn) \cdot \vec{y}^{\vec{e}_\ell}$, where 
    $\height(q_\ell) \leq H$ and $\deg(q_{\ell}) \leq D$.
    Hence,
    by~Lemma~\ref{lemma:lambda-close-to-variable}, 
    $G_{\pm p_j}$ can be taken as the interval $[-L'..L']$ where $L' \coloneqq \left(2^{3c}D\ceil{\ln(H)}\right)^{6\abs{M} \cdot k^{3\abs{M}}}$.
    We can now conclude the proof:
    \begin{align*} 
      L &= \max \{ 1 + \ceil{\log_{\cn}(n)} + n + 2 \abs{g'} : g' \in G_{\pm p_j} \text{ for some } j \in [0..n]\}\\
      & \leq 1 + \ceil{\log_{\cn}(n)} + n + 2\left(2^{3c}D\ceil{\ln(H)}\right)^{6\abs{M} \cdot k^{3\abs{M}}}
      &\tag*{bound on $G_{\pm p_j}$}\\
      & \leq 1 + 2^{2c}\ceil{\ln(n)} + n + 2\left(2^{3c}D\ceil{\ln(H)}\right)^{6\abs{M} \cdot k^{3\abs{M}}}
      &\tag*{by Equation~\ref{last-lambda-eq7}}\\
      & \leq 2^{2c+1}n + 2\left(2^{3c}D\ceil{\ln(H)}\right)^{6\abs{M} \cdot k^{3\abs{M}}}
      &\tag*{since $c,n \geq 1$}\\
      & \leq n\left(2^{4c}D\ceil{\ln(H)}\right)^{6\abs{M} \cdot k^{3\abs{M}}}.
      &&\tag*{\qedhere}
    \end{align*}
\end{proof}


\noindent 
By relying on the characterisation, given in~\Cref{lemma:last-lambda-lemma}, of
the values that $\lambda(u^*)$ can take, where $u^* > 0$ is the root of some
polynomial, and by applying our previous observation that satisfiability can be
witnessed by picking elements of $\ipow{\cn}$ that are ``close'' to $u^*$ (i.e.,
the numbers $\cn^{-1} \cdot \lambda(u^*)$, $\lambda(u^*)$ or $\cn \cdot
\lambda(u^*)$), we obtain the following key lemma.

\begin{restatable}{lemma}{LemmaRelativiseQuantifiers}
  \label{lemma:relativise-quantifiers}
  Let  $\phi(u,\vec y)$ be a quantifier-free formula from $\exists \ipow{\cn}$.
  Then, $\exists u\,\phi$ is equivalent to
  \begin{equation}
    \tag{$\dagger$}
    \label{eq:lemma14}
    \bigvee_{\ell \in [-1..1]}\ 
    \bigvee_{q \in Q}\ 
    \bigvee_{(j,g,\vec y^{\vec \ell}) \in F_q}
    \exists u : u^{j} = \cn^{j \cdot \ell + g} \cdot \vec y^{\vec \ell} \land \phi
  \end{equation}
  where $Q$ is the set of all polynomials in $\phi$ featuring $u$, 
  plus the polynomial $u-1$,
  and each~$F_q$ is the set obtained by applying~\Cref{lemma:last-lambda-lemma}
  to $r(x,v,\vec y) \coloneqq q\sub{x \cdot v}{u}$, with~$x,v$~fresh~variables.
\end{restatable}

\begin{proof}
  The right-to-left implication is trivial. Let us show the left-to-right
  implication. Below, let $\psi(u,\vec y) \coloneqq \phi \land \ipow{\cn}(u)
  \land \bigwedge_{y \in \vec y} \ipow{\cn}(y)$. For simplicity of the
  argument, instead of the left-to-right implication in the statement, we
  consider the following formula~of~$\R(\ipow{\cn})$:%
  \begin{equation}
    (\exists u \, \psi)
    \implies 
    \bigvee_{\ell = -1}^{1}\ 
    \bigvee_{q \in Q}\ 
    \bigvee_{(j,g,\vec y^{\vec \ell_1}, \vec y^{\vec \ell_2}) \in F_q}
    \exists u \left( u^{j} = \cn^{j \cdot \ell + g} \cdot \vec y^{\vec \ell_1-\vec \ell_2} \land \psi \right).
    \label{lemma:relativise-quantifiers:eq0}
  \end{equation}
  Since in this implication all variables are constrained to be integer powers
  of $\cn$, this formula is equivalent to the left-to-right implication of the
  equivalence in the statement of the lemma. We show
  Formula~\eqref{lemma:relativise-quantifiers:eq0} by relying on a series of
  tautologies.

  \begin{claim}\label{lemma:relativise-quantifiers:claim1}
    Let $Q'$ be the set of all polynomials in $\phi$ featuring $u$.
    The following formula is a tautology of~$\R(\ipow{\cn})$:
    \begin{align*}
      (\exists u\, \psi)
      \implies 
      &\Big(\big(\forall u \,(u > 0 \implies \phi)\big) 
      \lor\\
      &\hspace{-1cm} \bigvee_{r \in Q'}
      \exists w \big(w > 0 \land r(w,\vec y) = 0 \land (\bigvee_{i=0}^n p_{r,i}(\cn,\vec y) \neq 0) 
      \land \exists u (w \cdot \cn^{-1} \leq u \leq w \cdot \cn \land \psi) \big)
      \Big),
    \end{align*}
    where $r \in Q'$ is of the form $r(x,\vec y) = \sum_{i=0}^n p_{r,i}(\cn,\vec y) \cdot x^i$.
  \end{claim}
  \begin{proof}
    Let $\vec y^*$ be real numbers that are a solution to the formula $(\exists u\, \psi) \land \lnot \forall u \,(u > 0 \implies \phi)$. 
    To prove the claim, it suffices to show that then $\vec y^*$ is a solution to the formula
    \begin{equation}\label{lemma:relativise-quantifiers:eqn1}
      \bigvee_{r \in Q'}
      \exists w \big(w > 0 \land r(w,\vec y) = 0 \land (\bigvee_{i=0}^n p_{r,i}(\cn,\vec y) \neq 0) 
      \land \exists u (w \cdot \cn^{-1} \leq u \leq w \cdot \cn \land \psi) \big).
    \end{equation}
    Let $S \coloneqq \{ u \in \R : \phi(u,\vec y^*) \land u > 0 \}$ be the set
    of positive real numbers satisfying $\phi$ with respect to the vector
    $\vec y^*$ we have picked. Since $\vec y^*$ satisfies $\lnot \forall u (u
    > 0 \implies \phi)$, we have $S \subsetneq \R_{>0}$. Since $S$ is the set
    of solutions of over $\R_{>0}$ of a formula in the language of Tarski
    arithmetic, it is a finite union~$\bigcup_{j \in J} I_j$ of disjoint
    (open, closed or half-open) intervals with endpoints in $\R \cup
    \{{+}\infty\}$. This follows directly from the fact that Tarski arithmetic
    is an o-minimal theory~\cite[Chapter 3.3]{marker2002model}. Without loss
    of generality, we can assume $\{I_j\}_{j \in J}$ to be a minimal family of
    intervals characterising $S$; in other words, we can assume that for every
    two distinct intervals $I_j$ and $I_k$, the set $I_j \cup I_k$ is not an
    interval. Since $\vec y^*$ satisfies $\exists u\, \psi$, there is $j \in
    J$ such that $I_j$ contains an integer power of $\cn$, $\cn^{i_j}$. The
    interval $I_j$ is of the form $(a,b)$, $[a,b)$, $(a,b]$ or $[a,b]$, for
    some $a \in \R_{>0}$ and $b \in \R_{>0} \cup \{{+}\infty\}$. We divide the
    proof in two cases, depending on whether $b = {+}\infty$.
    \begin{description}
      \item[case: $b \neq {+}\infty$] 
        There is an interval $(c,d)$ around $b$ such that $(c,b)$ and $(b,d)$
        are non-empty, $(c,d) \subseteq I_j$, and $(b,d) \cap S = \emptyset$.
        That is, the truth of the formula $\phi(u,\vec y^*) \land u > 0$
        changes around $b$. Since $\phi(u,\vec y)$ is a quantifier-free
        formula from $\exists \ipow{\cn}$, this means that the truth value of
        a polynomial inequality $r(u,\vec y^*) \sim 0$ changes around $b$,
        which in turn implies both $r(b,\vec y^*) = 0$ (since polynomials are
        continuous functions) and that $r(b,\vec y^*)$ is non-constant, i.e.,
        $\bigvee_{i=0}^n p_{r,i}(\cn,\vec y^*) \neq 0$. At this point we have
        established that $b > 0 \land r(b,\vec y^*) = 0 \land (\bigvee_{i=0}^n
        p_{r,i}(\cn,\vec y) \neq 0)$ holds, and hence to conclude that
        Formula~\eqref{lemma:relativise-quantifiers:eqn1} holds we must now
        show that there is $u^* \in \ipow{\cn}$ such that $(b \cdot \cn^{-1}
        \leq u^* \leq b \cdot \cn \land \psi(u^*,\vec y^*))$ also holds.
        Observe that, since $\vec y^*$ satisfies $\exists u \psi$, each entry
        in $\vec y^*$ is an integer power of $\cn$, and therefore it suffices
        to show that $u^* \in \ipow{\cn}$ such that $b \cdot \cn^{-1} \leq u^*
        \leq b \cdot \cn$ and $u^* \in I_j$. This follows from the case
        analysis below:
        \begin{description}
          \item[case: $b \in I_j$ and $b \in \ipow{\cn}$] In this case, $u^*
          = b$.
          \item[case: $b \not\in I_j$ and $b \in \ipow{\cn}$] We have
          $\lambda(b) = b$. Since we are assuming that $I_j$ contains an
          integer power of $\cn$, we must have that $\cn^{-1} \cdot b$, which
          is the largest integer power of $\cn$ that is strictly below the
          endpoint $b$, belongs to $I_j$. Hence, we can take $u^* = \cn^{-1}
          \cdot b$.
          \item[case: $b \not\in \ipow{\cn}$]
          We have $\lambda(b) < b$, and $\lambda(b)$ is the largest integer
          power of $\cn$ strictly below $b$. We have
          $\lambda(b) \in I_j$ and $b \cdot \cn^{-1} \leq \lambda(b)$, and so
          we can take $u^* = \lambda(b)$.
        \end{description}
      \item[case: $b = {+}\infty$] 
      In this case, instead of the right endpoint $b$ we consider the left
      endpoint $a$. Since $S$ if a strict subset of $\R_{> 0}$, we must have $a
      > 0$. By the same arguments as in the previous case, $\phi$ must feature
      a polynomial inequality $r(u,\vec y) \sim 0$ such that $r(a,\vec y^*) =
      0$. We thus have $a > 0 \land r(a,\vec y^*) = 0 \land (\bigvee_{i=0}^n
      p_{r,i}(\cn,\vec y) \neq 0)$, and to conclude that
      Formula~\eqref{lemma:relativise-quantifiers:eqn1} holds it suffices to
      show that there is $u^* \in \ipow{\cn}$ such that $a \cdot \cn^{-1} \leq
      u^* \leq a \cdot \cn$ and $u^* \in I_j$. This is shown with a case
      analysis that is analogous to the one above:
      \begin{description}
        \item[case: $a \in I_j$ and $a \in \ipow{\cn}$] In this case, $u^* =
        a$.
        \item[case: $a \not\in I_j$ and $a \in \ipow{\cn}$] We have
        $\lambda(a) = a$. Since we are assuming that $I_j$ contains an integer
        power of $\cn$, we must have that $a \cdot \cn$, which is the largest
        integer power of $\cn$ strictly above the endpoint $a$,
        belongs to $I_j$. Hence, we can take $u^* = a \cdot \cn$.
        \item[case: $a \not\in \ipow{\cn}$]
        We have $\lambda(a) < a$. In this case, the largest power of $\cn$
        that is strictly above the endpoint $a$ is $\lambda(a) \cdot \cn$. We
        have $\lambda(a) \cdot \cn \in I_j$ and $a < \lambda(a) \cdot \cn \leq
        a \cdot \cn$, and so we can take $u^* = \lambda(a) \cdot \cn$.
      \end{description}
    \end{description}
    In both the cases above, we have shown that $\vec y^*$ 
    is a solution to~Formula~\eqref{lemma:relativise-quantifiers:eqn1}.
  \end{proof}

  \begin{claim}\label{lemma:relativise-quantifiers:claim2}
    The following formula is a tautology of~$\R(\ipow{\cn})$:
    \[
      (\forall u (u > 0 \implies \psi)) 
      \implies 
      \exists w (w = 1 \land \exists u (w \cdot \cn^{-1} \leq u \leq w \cdot \cn \land \psi)).
    \]
  \end{claim}
  \begin{proof}
    First, observe that $\exists w (w = 1 \land \exists u (w \cdot \cn^{-1}
    \leq u \leq w \cdot \cn \land \psi))$ is trivially equivalent to $\exists
    u ( \cn^{-1} \leq u \leq \cn \land \psi)$. (The addition of the variable
    $w$ assigned to $1$ is convenient for the forthcoming arguments of the
    proof of~\Cref{lemma:relativise-quantifiers}.)

    Let $\vec y$ be real numbers satisfying the antecedent $(\forall u (u > 0
    \implies \psi))$ of the implication. Since $\cn > 1$, the non-empty
    interval $[\cn^{-1},\cn]$ is included in $\R_{>0}$. Therefore, from the
    antecedent of the implication we deduce that $\vec y$ satisfies $\exists u
    ( \cn^{-1} \leq u \leq \cn \land \psi)$.
  \end{proof}
  By Claims~\ref{lemma:relativise-quantifiers:claim1}
  and~\ref{lemma:relativise-quantifiers:claim2}, the following formula is a
  tautology of~$\R(\ipow{\cn})$: 
  \begin{align*}
    (\exists u\, \psi)
    \implies 
    \!\!\bigvee_{r \in Q}\!
    \exists w \Big( w > 0 \land r(w,\vec y) = 0 \land \exists u (w \cdot \cn^{-1} \leq u \leq w \cdot \cn \land \psi) \land \bigvee_{i=0}^n p_{r,i}(\cn,\vec y) \neq 0 \Big).
  \end{align*}
  Since every $w > 0$ can be uniquely decomposed into $x\cdot v$, with
  $x$ being an integer power of $\cn$ and $1 \leq v < \cn$, the above formula
  can be rewritten as follows:
  \begin{align*}
    (\exists u\, \psi)
    \implies 
    \bigvee_{r \in Q}
    \exists x \exists v \Big(& \ipow{\cn}(x) \land 1 \leq v < \cn \land r(x \cdot v,\vec y) = 0 \land (\bigvee_{i=0}^n p_{r,i}(\cn,\vec y) \neq 0)\\ 
    &{} \land \exists u ( (x \cdot v) \cdot \cn^{-1} \leq u \leq (x \cdot v) \cdot \cn \land \psi) \Big).
  \end{align*}
  Hence, by applying~\Cref{lemma:last-lambda-lemma}, we conclude the following
  formula is a tautology of $\R(\ipow{\cn})$: 
  \begin{align}
    (\exists u\, \psi)
    \implies 
    \bigvee_{r \in Q}
    \exists x \exists v \Big(& \ipow{\cn}(x) \land 1 \leq v < \cn \land 
    \big(\bigvee_{(j,g,\vec y^{\vec \ell}) \in F_r}
      x^{j} = \cn^{g} \cdot \vec y^{\vec \ell}      
    \big)
    \notag\\ 
    &{} \land \exists u ( (x \cdot v) \cdot \cn^{-1} \leq u \leq (x \cdot v) \cdot \cn \land \psi) \Big).
    \label{lemma:relativise-quantifiers:eq1}
  \end{align}

  We now simplify the inequalities $(x \cdot v) \cdot \cn^{-1} \leq u \leq (x \cdot v) \cdot \cn$:

  \begin{claim}\label{lemma:relativise-quantifiers:claim3}
    The following formula is a tautology of~$\exists\R(\ipow{\cn})$:
    \[
      (\ipow{\cn}(u) \land \ipow{\cn}(x) \land 1 \leq v < \cn 
      \land (x \cdot v) \cdot \cn^{-1} \leq u \leq (x \cdot v) \cdot \cn) 
      \implies 
      \bigvee\nolimits_{\ell \in \{-1,0,1\}}
      u = \cn^{\ell} \cdot x.
    \]
  \end{claim}
  \begin{proof}
    Let $(u,v,x)$ be three real numbers satisfying the antecedent of the
    implication. By properties of $\lambda$, $(x \cdot v) \cdot \cn^{-1} \leq
    u \leq (x \cdot v) \cdot \cn$ implies $\lambda(x \cdot v) \cdot \cn^{-1}
    \leq \lambda(u) \leq \lambda(x \cdot v) \cdot \cn$. Since the antecedent
    of the implication imposes $u$ and $x$ to be integer powers of $\cn$, and
    $1 \leq v < \cn$, we have $\lambda(u) = u$ and $\lambda(x \cdot v) = x$.
    We conclude that $x \cdot \cn^{-1} \leq u \leq x \cdot \cn$, or
    equivalently $u = \cn^{\ell} \cdot x$ for some $\ell \in [-1,1]$, as
    required.
  \end{proof}
  We apply Claim~\ref{lemma:relativise-quantifiers:claim3} to
  Equation~\ref{lemma:relativise-quantifiers:eq1}, obtaining the following
  tautology of~$\R(\ipow{\cn})$:
  \begin{align*}
    (\exists u\, \psi)
    \implies\!
    \bigvee_{r \in Q}
    \exists x \Big(& \ipow{\cn}(x) \land 
    \big(\!\!\!\bigvee_{(j,g,\vec y^{\vec \ell}) \in F_r}\!\!\!
      x^{j} = \cn^{g} \cdot \vec y^{\vec \ell}      
    \big) \land \exists u \big( \psi \land \bigvee_{\ell \in \{-1,0,1\}}
    u = \cn^{\ell} \cdot x \big) \Big).
  \end{align*}
  Lastly, in the above formula, we can exponentiate both sides of $u =
  \cn^{\ell} \cdot x$ by $j$ and eliminate~$x$. That is, the following entailment 
  with formulae from~$\exists\R(\ipow{\cn})$ holds:
  \begin{align*}
    \exists u\, \psi
    \models 
    \bigvee_{\ell = -1}^{1} 
    \bigvee_{q \in Q}
    \bigvee_{(j,g,\vec y^{\vec \ell}) \in F_q}
    \exists u 
    \big(
      u^{j} = \cn^{j \cdot \ell + g} \cdot \vec y^{\vec \ell}      
      \land \psi
    \big).
  \end{align*}
  This concludes the proof of \Cref{lemma:relativise-quantifiers}.
\end{proof}

To eliminate the variable $u$, we now 
consider each disjunct $\exists u \left( u^{j} = \cn^{k} \cdot \vec{y}^{\vec \ell}
\land \phi \right)$ from Formula~\eqref{eq:lemma14} and, roughly speaking, 
substitute $u$ with $\sqrt[j]{\cn^{k} \cdot \vec{y}^{\vec{\ell}}}$. 
We do not need however to introduce~$j$th roots, as shown in the following lemma.

\begin{restatable}{lemma}{LemmaRemoveU}
  \label{lemma:remove-u}
  Let $\phi(u,\vec y)$ be a quantifier-free formula from $\exists \ipow{\cn}$,
  with $\vec y = (y_1,\dots,y_n)$. Let $j \in \N_{\geq 1}$, $k \in \Z$ and
  $\vec{\ell} \coloneqq (\ell_1,\dots,\ell_n) \in \Z$. Then,
  $\exists \vec{y} \exists u : u^{j} = \cn^{k} \cdot \vec y^{\vec \ell} \land
  \phi$ is equivalent to
  \begin{equation*}
    \textstyle\bigvee_{\vec r \coloneqq (r_1,\dots,r_n) \in R}\
    \exists \vec z : \phi\sub{z_i^j \cdot \cn^{r_i}}{y_i : i \in [1..n]}\sub{\cn^{\frac{k + \vec \ell \cdot \vec r}{j}} \cdot \vec z^{\vec \ell}}{u},
  \end{equation*}
  where $R \coloneqq \big\{(r_1,\dots,r_n)\in[0..j-1]^n : j \text{ divides } k +
  \sum_{i=1}^n r_i \cdot \ell_i \big\}$, $\vec \ell \cdot \vec r \coloneqq
  \sum_{i=1}^n r_i \cdot \ell_i$, and $\vec z \coloneqq (z_1,\dots,z_n)$ is a
  vector of fresh variables.
\end{restatable}

\begin{proof} 
  We first prove the right-to-left direction of the lemma. Consider $\vec r\in
  R$ such that the sentence $\exists\vec z : \phi\sub{z_i^j \cdot
  \cn^{r_i}}{y_i : i \in [1..n]}\sub{\cn^{\frac{k + \vec \ell \cdot
  \vec{r}}{j}} \cdot \vec z^{\vec \ell}}{u}$ is a tautology of $\exists
  \ipow{\cn}$. The following sequence of implications (in the language of $\R(\ipow{\cn})$) establishes the
  right-to-left direction:
  \begin{align*}
      &\exists\vec z : \phi\sub{z_i^j \cdot \cn^{r_i}}{y_i : i \in [1..n]}\sub{\cn^{\frac{k + \vec \ell \cdot \vec r}{j}} \cdot \vec z^{\vec \ell}}{u}\\
      \implies{}& \exists \vec z \exists\vec y\exists u: \phi(u,\vec y)\land
      \big(\bigwedge_{i=1}^n y_i=z_i^j \cdot \cn^{r_i}\big) \land u=\cn^{\frac{k + \vec \ell \cdot \vec r}{j}} \cdot \vec z^{\vec \ell}
      &\hspace{-1cm}\tag*{def.~of~substitution}\\
      \implies{}&\exists \vec z \exists\vec y\exists u: \phi(u,\vec y)\land \big(\bigwedge_{i=1}^n y_i^{\ell_i}=z_i^{j\ell_i} \cdot \cn^{r_i\ell_i} \big)\land u^j=\cn^{k + \vec \ell \cdot \vec r} \cdot \vec z^{\vec \ell\cdot j}\\
      \implies{}&\exists \vec z \exists\vec y\exists u: \phi(u,\vec y)\land \big(\bigwedge_{i=1}^n y_i^{\ell_i}=z_i^{j\ell_i} \cdot \cn^{r_i\ell_i} \big)\land u^j=\cn^{k} \cdot \prod_{i=1}^n (z_i^{j\ell_i} \cdot \cn^{r_i\ell_i})\\
      \implies{}&\exists\vec y\exists u: \phi(u,\vec y)\land u^j=\cn^{k}\cdot y_1^{\ell_1}\cdot\cdots \cdot y_n^{\ell_n}.
  \end{align*}
  We move to the left-to-right direction. Suppose $\exists \vec y \exists u :
  u^{j} = \cn^{k} \cdot \vec y^{\vec \ell} \land \phi$ to be a tautology
  of~$\exists \ipow{\cn}$. For every $i \in [1..n]$, we have $y_i =
  \cn^{\alpha_i}$ for some $i \in \Z$. We consider the quotient $\beta_i \in
  \Z$ and remainder $r_i \in [0..j-1]$ of the integer division of $\alpha_i$
  modulo $j$, that is, $\alpha_i = \beta_i \cdot j + r_i$. Setting $z_i =
  \cn^{\beta_i}$, we have $y_i = z_i^j \cdot \cn^{r_i}$. Therefore, the
  following sentence is a tautology of~$\exists\ipow{\cn}$:
  \[
      \exists \vec y \exists u \exists \vec z \bigvee_{(r_1,\dots,r_n) \in [0..j-1]^n} u^{j} = \cn^{k} \cdot \vec y^{\vec \ell} \land \phi \land \bigwedge_{i=1}^n y_i = z_i^j \cdot \cn^{r_i}.
  \]
  By distributing existential quantifiers over disjunctions and eliminating
  $\vec y$ by performing the substitutions $\sub{z_i^j \cdot \cn^{r_i}}{y_i}$,
  we conclude that the following sentence is also tautological:
  \begin{equation}
      \label{lemma:remove-u:eq1}
      \bigvee_{(r_1,\dots,r_n) \in [0..j-1]} \exists u \exists \vec z : u^{j} = \cn^{k} \cdot \prod_{i=1}^n (z_i^{\ell_i j} \cdot \cn^{\ell_i r_i})  \land \phi\sub{z_i^j \cdot \cn^{r_i}}{y_i : i \in [1..n]}.
  \end{equation}
  Since all the $z_i$ and $u$ are powers of $\cn$, there are $\alpha,\beta \in
  \Z$ such that $\vec z^{\vec\ell\cdot j}=\cn^{\alpha \cdot j}$ and
  $u^j=\cn^{\beta \cdot j}$. Observe that then, in order for $u^{j} = \cn^{k}
  \cdot \prod_{i=1}^n z_i^{\ell_ij} \cdot \cn^{\ell_ir_i}$ to hold, we must
  have $\beta j = k + \alpha j + \vec \ell \cdot \vec r$. This implies that
  $j$ divides $k+ \vec \ell \cdot \vec r$. Therefore, we can update
  Formula~\eqref{lemma:remove-u:eq1} as follows: 
  \begin{itemize}
      \item instead of a disjunction over all elements in $[0..j-1]^n$,
      consider the set 
      \[
      R \coloneqq \big\{(r_1,\dots,r_n)\in[0..j-1]^n : j \text{ divides } k + \textstyle\sum_{i=1}^n r_i \cdot \ell_i\big\};
      \]
      \item in the disjunct corresponding to $\vec r \coloneqq (r_1,\dots,r_n)
      \in R$, replace $u^{j} = \cn^{k} \cdot \prod_{i=1}^n (z_i^{\ell_ij}
      \cdot \cn^{\ell_ir_i})$ with $u = \cn^{\frac{k+\vec \ell \cdot
      \vec{r}}{j}}\cdot \vec z^{\vec\ell}$.
  \end{itemize}
  We conclude that the following sentence is a tautology of $\exists
  \ipow{\cn}$: 
  \[
      \bigvee_{\vec r \coloneqq (r_1,\dots,r_n) \in R} \exists u \exists \vec z : u = \cn^{\frac{k+\vec \ell \cdot \vec r}{j}}\cdot \vec z^{\vec\ell} \land \phi\sub{z_i^j \cdot \cn^{r_i}}{y_i : i \in [1..n]}.
  \]
  From the sentence above, we eliminate $u$ from each disjunct corresponding
  to $\vec r \in R$ by performing the substitution
  $\sub{\cn^{\frac{k+\vec{\ell} \cdot \vec r}{j}}\cdot \vec z^{\vec\ell}}{u}$.
  In doing so, we obtain the formula in the statement of the lemma.
\end{proof}

%

By chaining~\Cref{lemma:relativise-quantifiers,lemma:remove-u}, one can eliminate all variables from a quantifier-free formula $\phi(\vec x)$, obtaining an equisatisfiable
formula with no variables.

\subsection{Quantifier relativisation (Proof of \Cref{theorem:small-model-property})}
\label{subsection:quantifier-relativisation}

Looking closely at how a quantifier-free formula $\phi(u_1,\dots,u_n)$ of $\exists\ipow{\cn}$ evolves as we chain~\Cref{lemma:relativise-quantifiers,lemma:remove-u} to eliminate all variables, we see that the resulting variable-free formula is a finite disjunction~$\bigvee_{i} \psi_i$ of formulae $\psi_i$ that are obtained from $\phi$ via a sequence of substitutions stemming from~\Cref{lemma:remove-u}. As an example, for a formula in three variables~$\phi(u_1,u_2,u_3)$, each $\psi_i$ is obtained by applying a sequence of substitutions of the form:
\begin{center}
  \vspace{-3pt}
  \begin{tikzpicture}
    \node[align=center] (l1) at (0,0) {
        \footnotesize{elimination of $u_1$}
      };
    \node[align=center] (l2) at (4,0) {
        \footnotesize{elimination of $z_1$}
      };
    \node[align=center] (l3) at (8,0) {
        \footnotesize{elimination of $z_3$}
      };

    \node 
      (s1) [above = 0.05cm of l1] 
      {$\begin{cases}
        u_1 = \cn^{k_1} \cdot z_1^{\ell_1} \cdot z_2^{\ell_2}\\ 
        u_2 = z_1^{j_1} \cdot \cn^{r_1}\\
        u_3 = z_2^{j_1} \cdot \cn^{r_2}
      \end{cases}$};

    \node 
      (s2) [above = 0.05cm of l2] 
      {$\begin{cases}
        z_1 = \cn^{k_2} \cdot z_3^{\ell_3}\\ 
        z_2 = z_3^{j_2} \cdot \cn^{r_3}
      \end{cases}$};

    \node 
      (s3) [above = 0.05cm of l3] 
      {$\begin{cases} 
        z_3 = \cn^{k_3}
      \end{cases}$};

    \node (t2) [left = 1cm of s2] {};
    \draw[->] (t2) -- (s2);

    \node (t3) [left = 1cm of s3] {};
    \draw[->] (t3) -- (s3);

  \end{tikzpicture}
  \vspace{-5pt}
\end{center}
We can ``backpropagate'' these substitutions to the initial variables
$u_1,\dots,u_n$, associating to each one of them an integer power of $\cn$. In the
above example, we obtain the system
\begin{center}
  $\begin{cases} u_1 = \cn^{k_1} \cdot (\cn^{k_2} \cdot
        (\cn^{k_3})^{\ell_3})^{\ell_1} \cdot ((\cn^{k_3})^{j_2} \cdot
        \cn^{r_3})^{\ell_2}\\ 
        u_2 = (\cn^{k_2} \cdot (\cn^{k_3})^{\ell_3})^{j_1} \cdot \cn^{r_1}\\
        u_3 = ((\cn^{k_3})^{j_2} \cdot \cn^{r_3})^{j_1} \cdot \cn^{r_2}
      \end{cases}$
\end{center}
By~\Cref{lemma:last-lambda-lemma,lemma:relativise-quantifiers,lemma:remove-u}, we can restrict the integers occurring 
as powers of $\cn$ in the resulting system of substitutions to a finite set.
Since the disjunction $\bigvee_i \psi_i$ is finite, 
this implies that,
under the hypothesis that $\cn$ is a computable number that is either transcendental or has a polynomial root barrier, 
it is possible to compute a finite set $P_{\phi} \subseteq \Z$ witnessing the satisfiability of $\phi$. That is, the sentence $\exists u_1\dots \exists u_n\, \phi$ is equivalent to
\[ 
  \exists u_1 \dots \exists u_n \textstyle\bigvee_{(g_1,\dots,g_n) \in (P_{\phi})^n} \textstyle\left(\phi \land \bigwedge_{i=1}^n u_i = \cn^{g_i}\right).
\]

By formalising this step, we obtain a proof of the small witness property:
\begin{proof}[Proof of~\Cref{theorem:small-model-property}]
  The proposition is clearly true for $n = 0$, hence below we assume $n \geq 1$.
  By repeatedly applying~\Cref{lemma:relativise-quantifiers,lemma:remove-u} we
  conclude that there is a sequence $S_0,\dots,S_{n-1}$ of finite sets of
  integers, and a sequence $\phi_0, \phi_1, \dots, \phi_n$ of
  \emph{equisatisfiable} quantifier-free formulae such that 
  \begin{enumerate}[A.]
    \item\label{pp6:itemA} for every $r \in [0..n]$, the variables occurring in $\phi_r$ are
    among $u_{r+1},\dots,u_{n}$, 
    \item\label{pp6:itemB} $\phi_0 = \psi$, and 
    \item\label{pp6:itemC} for all $r \in [0..n-1]$, $\phi_{r+1} = \phi_r\sub{u_i^{j_r} \cdot
          \cn^{f_{r,i}}}{u_i : i \in [r+2..n]}\sub{\cn^{g_r} \cdot
          u_{r+2}^{\ell_{r,r+2}} \cdot \ldots \cdot
          u_{n}^{\ell_{r,n}}}{u_{r+1}}$, for some integers
          $j_r,f_{r,i},g_r,\ell_{r,r+2},\dots,\ell_{r,n}$ taken from the
          set $S_r$.
  \end{enumerate}
  Above, observe that for convenience and differently from~\Cref{lemma:remove-u}
  we are reusing the variables $u_2,\dots,u_n$ instead of introducing fresh
  variables $\vec z$. Without loss of generality, we assume $S_r$ to always
  contain $0$ and $1$. In this way, if $u_{r+1}$ does not occur in $\phi_r$
  (e.g., because it has been ``accidentally'' eliminated together with a
  previous variable), then we can pick $j_r = 1$ and $f_{r,i} = 0$, for every $i
  \in [r+2..n]$, in order to obtain $\phi_{r+1} = \phi_r$.

  From~\Cref{lemma:last-lambda-lemma,lemma:relativise-quantifiers}, for every $r
  \in [0..n-1]$, we have:
  \begin{enumerate}
    \item\label{pp6:itemOne} If $\cn$ is a computable transcendental number,
    there is an algorithm computing $S_r$ from~$\psi_r$.
    \item\label{pp6:itemTwo} If $\cn$ has a root barrier $\sigma(d,h) \coloneqq
    c \cdot (d+\ceil{\ln(h)})^k$, for some $c,k \in \N_{\geq 1}$, then,
    \begin{center}
      $\begin{aligned}
      j_r &\in [1..\deg(u_{r+1},\phi_r)],\\
      f_{r,i} &\in [0..j_r-1] \ \text{ and } \ \abs{\ell_{r,i}} \leq \deg(u_i,\phi_r), \qquad\text{for every } i \in [r+2..n],\\
      \abs{g_r} &\leq \deg(u_{r+1},\phi_r) \cdot ((2^{4c} D_r \cdot \ceil{\ln(H_r)})^{6 M_r k^{3M_r}}\!\!+n \cdot \max\{\deg(u_i,\phi_r) : i \in [r+2..n]\}).
      \end{aligned}$
    \end{center}
    where $H_r \coloneqq \max\{8,\height(\phi_r)\}$, $D_r \coloneqq
    \deg(\cn,\phi_r)+2$, and $M_r$ is the maximum number of monomials occurring
    in a polynomial of $\phi_r$. Here, $\deg(u_i,\phi_r)$
    (resp.~$\deg(\cn,\phi_r)$) stands for the maximum degree that the variable
    $u_i$ (resp.~$\cn$) has in a polynomial occurring in $\phi_r$, \emph{which in this
    proof we always assume to be at least $1$} without loss of generality.
  \end{enumerate}

  \noindent 
  As explained in~\Cref{subsection:quantifier-relativisation}, we can
  ``backpropagate'' the substitutions performed to define the formulae
  $\phi_1,\dots,\phi_n$ (\Cref{pp6:itemC}) to obtain a solution for
  $\psi$. Formally, we consider the set of integers $\{d_{i,h} : i \in [1..n], h
  \in [0..i-1]\}$ given by the following recursive definition: 
  \begin{align}
    d_{i,i-1} &\coloneqq g_{n-i} + \sum_{h=0}^{i-2} (d_{i-1,h} \cdot \ell_{n-i,n-h}), \notag\\
    d_{i,h} &\coloneqq d_{i-1,h} \cdot j_{n-i} + f_{n-i,n-h},
    &\text{ for every } h \in [0..i-2].
    \label{eq:def-dih}
  \end{align}
  Observe that $d_{1,0} = g_{n-1}$, and that all integers $d_{i,h}$ are
  ultimately defined in terms of integers from the sets $S_0,\dots,S_{n-1}$. We
  prove the following claim:

  \begin{claim}\label{claim:quantifier-relativisation}
    Suppose $\psi$ to be satisfiable. Then,
    for every $i \in [0..n]$, the assignment 
    \[ 
      \begin{cases}
        u_{n-h} &= \cn^{d_{i,h}} \qquad\qquad\text{for every } h \in [0..i-1]
      \end{cases}
    \]
    is a solution of $\phi_{n-i}$.
  \end{claim}

  \begin{proof}
    The proof is by induction on $i$. 

    \begin{description}
      \item[base case: $i = 0$]
      By~\Cref{pp6:itemC}, the formula $\phi_n$ does not feature any variable,
      and, accordingly, the assignment in the claim is empty. Since $\phi_n$ is
      equisatisfiable with $\phi_0$, and $\phi_0 = \psi$ by~\Cref{pp6:itemC}, we
      conclude that $\phi_n$ is equivalent to $\top$.
      
      \item[induction step: $i \geq 1$] 
      By induction hypothesis, the assignment 
      \[ 
        \begin{cases}
          u_{n-h} &= \cn^{d_{i-1,h}} \qquad\qquad\text{for every } h \in [0..i-2]
        \end{cases}
      \]
      is a solution of $\phi_{n-(i-1)}$. By~\Cref{pp6:itemC}, we have 
      \[ 
        \phi_{n-(i-1)} = \]\[\quad \phi_{n-i}\sub{u_t^{j_{n-i}} \cdot
          \cn^{f_{{n-i},t}}}{u_t : t \in [{n-i}+2..n]}\sub{\cn^{g_{n-i}} \cdot
          u_{{n-i}+2}^{\ell_{{n-i},{n-i}+2}} \cdot \ldots \cdot
          u_{n}^{\ell_{{n-i},n}}}{u_{{n-i}+1}}.
      \]
      Therefore, the following assignment is a solution of $\phi_{n-i}$:
      \[ 
        \begin{cases}
          u_{n-(i-1)} &= \cn^{g_{n-i}} \cdot (\cn^{d_{i-1,i-2}})^{\ell_{{n-i},{n-i}+2}} \cdot \ldots \cdot (\cn^{d_{i-1,0}})^{\ell_{{n-i},n}}\\
          u_{n-(i-2)} &= (\cn^{d_{i-1,i-2}})^{j_{n-i}} \cdot \cn^{f_{n-i,n-i+2}}\\ 
          u_{n-(i-3)} &= (\cn^{d_{i-1,i-3}})^{j_{n-i}} \cdot \cn^{f_{n-i,n-i+3}}\\ 
          \vdots\\
          u_n &= (\cn^{d_{i-1,0}})^{j_{n-i}} \cdot \cn^{f_{n-i,n}}
        \end{cases}
      \]
      that is,
      \[ 
        \begin{cases}
          u_{n-(i-1)} &= \cn^{g_{n-i}+\Sigma_{h=0}^{i-2} (d_{i-1,h} \cdot \ell_{n-i,n-h})}\\
          u_{n-h} &= \cn^{d_{i-1,h} \cdot j_{n-i}+f_{n-i,n-h}} \qquad\qquad\text{for every } h \in [0..i-2]
        \end{cases}
      \]
      and the statement follows by definition of $d_{i,i-1},\dots,d_{i,0}$.
      \qedhere%
    \end{description}
  \end{proof}

  Let us move back to the proof of~\Cref{theorem:small-model-property}. Given
  the finite sets $S_0,\dots,S_{n-1}$, we can compute an upper bound $U \in \N$
  to the absolute value of the largest integer among $d_{n,0},\dots,d_{n,n-1}$.
  Let~$P_{\psi} \coloneqq [-U..U]$. By~\Cref{claim:quantifier-relativisation} and we
  conclude that, whenever satisfiable, the formula~$\psi$ has a solution in the set
  $\{(\cn^{j_1},\dots,\cn^{j_n}) : j_1, \dots, j_n \in P_{\psi} \}$. Now, thanks to
  \Cref{pp6:itemOne,pp6:itemTwo} above, the finite sets $S_0,\dots,S_{n-1}$ can
  be computed in both the cases where either $\cn$ is a computable
  transcendental number or $\cn$ is a number with a polynomial root barrier. The
  set~$P_{\psi}$ can thus be computed in both these cases, which implies the
  effectiveness of the procedure required
  by~\Cref{theorem:small-model-property}.

  To conclude the proof, we derive an upper bound on $U$ in the case where $\cn$
  is a number with a polynomial root barrier $\sigma(d,h) \coloneqq c \cdot (d +
  \ceil{\ln(h)})^k$ for some $c,k \in \N_{\geq 1}$. We start by expressing, for
  every $r \in [0..n-1]$, the bounds from~\Cref{pp6:itemTwo} in terms of
  parameters of $\phi_0$. Below, let $E_r \coloneqq \max\{\deg(u_i,\phi_r) : i
  \in [r+1..n]\}$.
  \begin{claim}\label{claim:ugly-bounds}
    For every $r \in [0..n-1]$, we have
    \begin{align*}
        M_r &\leq M_0, &E_r &\leq 4^{2^r-1} \cdot (E_0)^{2^r},
        &H_r &\leq 2^r \cdot H_0, &\deg(\cn,\phi_r) &\leq (G_r)^{I^r-1} \cdot \deg(\cn,\phi_0)^{I^r},
    \end{align*}
    where $G_r \coloneqq n \cdot 2^{6 \cdot 2^r}(E_0)^{3 \cdot 2^r} \big(2^{r+4c+2} \ceil{\ln(H_0)}\big)^{I}$ and $I \coloneqq 6M_0k^{3M_0}$.
  \end{claim}
  \begin{proof}
    For $r = 0$ the claim is trivially true. Below, let us assume the claim to
    be true for $r \in [0..n-2]$, and show that it then also holds for $r+1$.
    Recall that, by \Cref{pp6:itemC}, we have 
    \begin{equation}
      \label{eq:phirp1phir}
      \phi_{r+1} = \phi_r\sub{u_i^{j_r} \cdot
      \cn^{f_{r,i}}}{u_i : i \in [r+2..n]}\sub{\cn^{g_r} \cdot
      u_{r+2}^{\ell_{r,r+2}} \cdot \ldots \cdot
      u_{n}^{\ell_{r,n}}}{u_{r+1}}.
    \end{equation}
    Recall that the integers $\ell_{r,i}$ and $g_r$ might be negative, and thus
    the substitutions performed in~\Cref{eq:phirp1phir} may require to update
    the polynomials in the formula so that they do not contain negative degrees
    for $\cn$ and each $u_i$. As described
    in~\Cref{subsection:quantifier-elimination}, these updates do not change the
    number of monomials nor the height of the polynomials, but might double the
    degree of each variable and of $\cn$. Hence, of the four bounds in the
    statement, which we now consider separately, these updates only affects the
    cases of $E_{r+1}$ and $\deg(\cn,\phi_{r+1})$.
    \begin{description}
      \item[case: $M_{r+1}$] 
        The substitutions done to obtain $\phi_{r+1}$ from $\phi_r$ replace
        variables with monomials. These type of substitutions do not increase
        the number of monomials occurring in the polynomials of a formula. (They
        may however decrease, causing an increase in the height of the
        polynomials, see below.) Therefore, we have $M_{r+1} \leq M_r \leq M_0$.
      \item[case: $H_{r+1}$] 
        The substitutions $\sub{u_i^{j_r} \cdot \cn^{f_{r,i}}}{u_i : i \in
        [r+2..n]}$ do not increase the heights of the polynomials in the
        formula. Indeed, consider a polynomial of the form 
        \begin{equation}
          \label{eq:a-polynomial}
          p(\cn,\vec u) + a \cdot \cn^{e_1} \cdot \vec u^{\vec d_1} 
          + b \cdot \cn^{e_2} \cdot \vec u^{\vec d_2},
        \end{equation}
        where $\vec u = (u_{r+1},\dots,u_{n})$, and $\cn^{e_1} \cdot
        \vec{u}^{\vec{d}_1}$ and $\cn^{e_2} \cdot \vec{u}^{\vec{d}_2}$ are two
        syntactically distinct monomials (i.e., either $e_1 \neq e_2$ or
        $\vec{d}_1 \neq \vec d_2$). Given $i \in [r+2..n]$, consider the
        substitution $\sub{u_i^{j_r} \cdot \cn^{f_{r,i}}}{u_i}$. We have three
        cases: 
        \begin{itemize}
          \item If $u_i$ occurs with different powers in the two monomials
          $\cn^{e_1} \cdot \vec{u}^{\vec{d}_1}$ and $\cn^{e_2} \cdot
          \vec{u}^{\vec{d}_2}$, then it will still occur with different powers
          in the monomials $(\cn^{e_1} \cdot \vec{u}^{\vec d_1})\sub{u_i^{j_r}
          \cdot \cn^{f_{r,i}}}{u_i}$ and $(\cn^{e_2} \cdot
          \vec{u}^{\vec{d}_2})\sub{u_i^{j_r} \cdot \cn^{f_{r,i}}}{u_i}$ obtained
          after replacement.
          \item If $u_i$ occurs with the same power $\hat{d}$ in the two
          monomials, and $e_1 \neq e_2$, then, after replacement, $\cn$ occurs
          with different powers in the obtained monomials $e_1 + \hat{d} \cdot
          f_{r,i}$ and $e_2 + \hat{d} \cdot f_{r,i}$ respectively.
          \item If $u_i$ occurs with the same power in the two monomials, and
          $e_1 = e_2$, then there is a variable $u_{t}$ with $t \neq i$ that, in
          the two monomials, occurs with different powers, say $\hat{d}_1$ and
          $\hat{d}_2$. (Note: one among $\hat{d}_1$ or $\hat{d}_2$ may be $0$.)
          This variable is unchanged by the substitution $\sub{u_i^{j_r} \cdot
          \cn^{f_{r,i}}}{u_i}$, and thus in the resulting monomials $u_{t}$
          still occurs with powers $\hat{d}_1$ and $\hat{d}_2$.
        \end{itemize}
        We move to the substitution $\sub{\cn^{g_r} \cdot u_{r+2}^{\ell_{r,r+2}}
        \cdot \ldots \cdot u_{n}^{\ell_{r,n}}}{u_{r+1}}$, which may increase the
        height of polynomials. Consider a polynomial as
        in~\Cref{eq:a-polynomial}. Observe that if $u_{r+1}$ occurs with a
        non-zero power in both the monomials $\cn^{e_1} \cdot \vec u^{\vec d_1}$
        and $\cn^{e_2} \cdot \vec u^{\vec d_2}$, then the two monomials
        $(\cn^{e_1} \cdot \vec u^{\vec d_1})\sub{\cn^{g_r} \cdot
        u_{r+2}^{\ell_{r,r+2}} \cdot \ldots \cdot u_{n}^{\ell_{r,n}}}{u_{r+1}}$
        and $(\cn^{e_2} \cdot \vec u^{\vec d_2})\sub{\cn^{g_r} \cdot
        u_{r+2}^{\ell_{r,r+2}} \cdot \ldots \cdot u_{n}^{\ell_{r,n}}}{u_{r+1}}$
        obtained after replacement are still different (a formal proof of this
        fact follows similarly to the one we have just discussed for the
        substitution $\sub{u_i^{j_r} \cdot \cn^{f_{r,i}}}{u_i}$). The same
        holds true if $u_{r+1}$ does not occur in any of the two monomials. If
        instead $u_{r+1}$ occurs with a non-zero power only in one monomial, say
        $\cn^{e_2} \cdot \vec u^{\vec d_2}$, we might have 
        \[
          (\cn^{e_2} \cdot \vec u^{\vec d_2})\sub{\cn^{g_r} \cdot
          u_{r+2}^{\ell_{r,r+2}} \cdot \ldots \cdot
          u_{n}^{\ell_{r,n}}}{u_{r+1}}
          =\]\[
          (\cn^{e_1} \cdot \vec u^{\vec d_1})\sub{\cn^{g_r} \cdot
          u_{r+2}^{\ell_{r,r+2}} \cdot \ldots \cdot
          u_{n}^{\ell_{r,n}}}{u_{r+1}}
          = \cn^{e_1} \cdot \vec u^{\vec d_1}.
        \]
        Hence, after replacement, the coefficient of $\cn^{e_1} \cdot
        \vec{u}^{\vec{d}_1}$ updates from $a$ to $(a+b)$. Note that no
        further increase are possible. Indeed, suppose $p(\cn,\vec u)$ contains
        a third monomial $\cn^{e_3}\vec u^{\vec d_3}$ in which $u_{r+1}$ has a
        non-zero power. By the arguments above, we have 
        \[\quad
          (\cn^{e_2} \cdot \vec u^{\vec d_2})\sub{\cn^{g_r} \cdot
          u_{r+2}^{\ell_{r,r+2}} \cdot \ldots \cdot
          u_{n}^{\ell_{r,n}}}{u_{r+1}} 
          \neq 
          (\cn^{e_3} \cdot \vec u^{\vec d_3})\sub{\cn^{g_r} \cdot
          u_{r+2}^{\ell_{r,r+2}} \cdot \ldots \cdot
          u_{n}^{\ell_{r,n}}}{u_{r+1}},
        \]
        and therefore no other monomial from $p$ can be updated to $\cn^{e_1}
        \cdot \vec u^{\vec d_1}$ after replacement. Since $\abs{a},\abs{b} \leq
        H_r$, we have $\abs{a+b} \leq 2 \cdot H_r$. This shows $H_{r+1} \leq 2
        \cdot H_r \leq 2^{r+1} \cdot H_0$.
      \item[case: $E_{r+1}$] Consider $u_i$ with $i \in [r+2..n]$. We show
        $\deg(u_i,\phi_{r+1}) \leq 4 (E_r)^2$, which implies $E_{r+1} \leq 4
        (4^{2^r-1} (E_0)^{2^r})^{2} = 4^{2^{r+1}-1}(E_0)^{2^{r+1}}$, as
        required. Consider a monomial occurring in $\phi_r$ and let $d_1$ and
        $d_2$ be the non-negative integers occurring as powers of $u_i$ and
        $u_{r+1}$ in this monomial. The substitutions performed to obtain
        $\phi_{r+1}$ (\Cref{eq:phirp1phir}) update the power of $u_i$ in the
        monomial from $d_1$ to $d_1 \cdot j_r + d_2 \cdot \ell_{r,i}$.
        By~\Cref{pp6:itemTwo}, $j_r \in [1..\deg(u_{r+1},\phi_r)]$ and
        $\abs{\ell_{r,i}} \leq \deg(u_i,\phi_r)$, and therefore
        $j_r,\abs{\ell_{r,i}} \leq E_r$. We conclude that $\abs{d_1 \cdot j_r +
        d_2 \cdot \ell_{r,i}} \leq 2 \cdot (E_r)^2$. Lastly, we need to account
        for the updates performed to the formula in order remove the negative
        integers that occur as powers of the variables and of $\cn$. As already
        stated, in the worst case, these updates double the degree of each
        variable, and so $E_{r+1} \leq 4 (E_r)^2$.
      \item[case: $\deg(\cn,\phi_{r+1})$] 
        We start by reasoning similarly to the previous case. 
        Consider a monomial $\cn^{d} \cdot u_{r+1}^{d_{r+1}} \cdot \ldots \cdot u_{n}^{d_n}$
        occurring in $\phi_r$. 
        The substitutions performed to obtain $\phi_{r+1}$ 
        update the power of $\cn$ from $d$ to $d + g_r \cdot d_{r+1} + \sum_{i = r+2}^n f_{r,i} \cdot d_i$.
        Observe that 
        \begin{align*} 
          \qquad \abs{d + g_r \cdot d_{r+1} + \sum_{i = r+2}^n f_{r,i} \cdot d_i}
          \leq{}&
            \deg(\cn,\phi_r) + \Big(\abs{g_r} + \textstyle\sum_{i = r+2}^n \abs{f_{r,i}}\Big) \cdot E_r\\
          \leq{}& 
            \deg(\cn,\phi_r) + \Big(\abs{g_r} + \textstyle\sum_{i = r+2}^n E_r\Big) \cdot E_r 
            &\tag*{by~\Cref{pp6:itemTwo}.}
        \end{align*}
        Accounting for the updates performed to the formula in order to remove negative powers, 
        we conclude that $\deg(\cn,\phi_{r+1})$ is bounded by 
        $2 \cdot (\deg(\cn,\phi_r) + (\abs{g_r} +\!\!\sum_{i = r+2}^n E_r) \cdot E_r)$.
        We further analyse this quantity as follows:
        \allowdisplaybreaks{
        \begin{align*}
          & 2 \cdot (\deg(\cn,\phi_r) + (\abs{g_r} +\textstyle\sum_{i = r+2}^n E_r) \cdot E_r)\\
          \leq{}& 
            2\deg(\cn,\phi_r) + 
            2\Big(E_r ((2^{4c} D_r \ceil{\ln(H_r)})^{6 M_r k^{3M_r}}\!\!+n E_r)
            +\textstyle\sum_{i = r+2}^n E_r\Big) \cdot E_r 
            &\hspace{-0.2cm}\tag*{by~\Cref{pp6:itemTwo}}\\
          \leq{}& 
            2\deg(\cn,\phi_r) + 
            4n \cdot (E_r)^3 \cdot (2^{4c} D_r \cdot \ceil{\ln(H_r)})^{6 M_r k^{3M_r}}
          \\
          \leq{}& 
            2 \deg(\cn,\phi_r) + 
            4n \cdot (E_r)^3 \cdot (2^{4c} (\deg(\cn,\phi_r)+2) \cdot \ceil{\ln(H_r)})^{6 M_r k^{3M_r}}
          &\hspace{-0.2cm}\tag*{def.~of~$D_r$}\\
          \leq{}& 
            2\deg(\cn,\phi_r) + 
            4n \cdot (E_r)^3 (2^{4c+2} \deg(\cn,\phi_r) \cdot \ceil{\ln(H_r)})^{6 M_r k^{3M_r}}
          &\hspace{-0.5cm}\tag*{$\deg(\cn,\phi_r) \geq 1$}\\
          \leq{}& 
            5n \cdot (E_r)^3 (2^{4c+2} \cdot \ceil{\ln(H_r)})^{6 M_r k^{3M_r}} \deg(\cn,\phi_r)^{6 M_r k^{3M_r}}
          \\
          \leq{}& 
            5n (4^{2^r-1}(E_0)^{2^r})^3 (2^{4c+2} \ceil{\ln(2^r H_0)})^{6 M_0 k^{3M_0}} \deg(\cn,\phi_r)^{6 M_0 k^{3M_0}}
          \\
          &&\hspace{-2.5cm}\tag*{bounds on $E_r$, $H_r$ and $M_r$}\\
          \leq{}& 
            n 2^{6 \cdot 2^r}(E_0)^{3 \cdot 2^r} (2^{r+4c+2} \ceil{\ln(H_0)})^{6 M_0 k^{3M_0}} \deg(\cn,\phi_r)^{6 M_0 k^{3M_0}}
          \\
          \leq{}& 
          G_r \cdot \deg(\cn,\phi_r)^{I}
          &\hspace{-1cm}\tag*{def.~of~$G_r$ and $I$}
          \\
          \leq{}& 
          G_r \cdot ((G_r)^{I^r-1} \cdot \deg(\cn,\phi_0)^{I^r})^{I}
          &\hspace{-1.6cm}\tag*{bound on $\deg(\cn,\phi_r)$}
          \\
          \leq{}& 
          (G_{r+1})^{I^{r+1}-1} \cdot \deg(\cn,\phi_0)^{I^{r+1}}
          &\hspace{-4cm}\tag*{since $I \geq 2$ and $G_{r+1} \geq G_r$.}
        \end{align*}}
        This completes the proof of the claim.
        \qedhere
    \end{description}
  \end{proof}
  We use the bounds in~\Cref{claim:ugly-bounds} to also bound the quantities
  $j_r$, $f_{r,i}$, $\abs{\ell_{r,i}}$ and $\abs{g_r}$.
  \begin{claim}
    \label{claim:less-ugly-bounds}
    For every $r \in [0..n-1]$ and $i \in [r+2..n]$, we have
    \begin{align*}
      j_r,f_{r,i},\abs{\ell_{r,i}} &\leq 4^{2^r} (E_0)^{2^r}\!, \text{ and }
      \abs{g_r} \leq \big(n \cdot 2^{2^{r+4}+4c}(E_0)^{2^{r+3}} \ceil{\ln(H_0)} \cdot \deg(\cn,\phi_0) \big)^{(6M_0k^{3M_0})^{r+2}}.
    \end{align*}
  \end{claim}
  \begin{proof}
    By~\Cref{pp6:itemTwo}, 
    the numbers $j_r$, $f_{r,i}$ and $\abs{\ell_{r,i}}$ 
    are all bounded by $E_r$, 
    which in turn is bounded by $4^{2^r} (E_0)^{2^r}$ 
    (by~\Cref{claim:ugly-bounds}).
    Let us now consider $\abs{g_r}$.
    Observe that $\deg(\cn,\phi_r)$  and $\abs{g_r}$ are mutually dependant, 
    and in particular that in the proof of~\Cref{claim:ugly-bounds} 
    we have bounded $\deg(\cn,\phi_{r+1})$ with a long chain of manipulations establishing, among other inequalities, 
    \begin{align*}
        2 \cdot (\deg(\cn,\phi_r) + (\abs{g_r} + \sum\nolimits_{i = r+2}^n E_r) \cdot E_r) \ \leq \ (G_{r+1})^{I^{r+1}-1} \cdot \deg(\cn,\phi_0)^{I^{r+1}},
    \end{align*}
    where $G_{r+1} \coloneqq n \cdot 2^{6 \cdot 2^{r+1}}(E_0)^{3 \cdot 2^{r+1}} \big(2^{r+4c+3} \ceil{\ln(H_0)}\big)^{I}$ and $I \coloneqq 6M_0k^{3M_0}$.
    Since $\abs{g_r}$ is smaller than $(\deg(\cn,\phi_r) + (\abs{g_r} +\!\!\sum_{i = r+2}^n E_r) \cdot E_r)$, we conclude that 
    \begin{align*}
      \abs{g_r} 
        &\leq (G_{r+1})^{I^{r+1}-1} \cdot \deg(\cn,\phi_0)^{I^{r+1}}\\
        &\leq \Big(n \cdot 2^{6 \cdot 2^{r+1}}(E_0)^{3 \cdot 2^{r+1}} \big(2^{r+4c+3} \ceil{\ln(H_0)}\big)^{I} \cdot \deg(\cn,\phi_0) \Big)^{I^{r+1}}\\
        &\leq \Big(n \cdot 2^{6 \cdot 2^{r+1}+r+4c+3}(E_0)^{3 \cdot 2^{r+1}} \ceil{\ln(H_0)} \cdot \deg(\cn,\phi_0) \Big)^{I^{r+2}}\\
        \hspace{1.1cm}&\leq \Big(n \cdot 2^{2^{r+4}+4c}(E_0)^{2^{r+3}} \ceil{\ln(H_0)} \cdot \deg(\cn,\phi_0) \Big)^{(6M_0k^{3M_0})^{r+2}}.
        &\tag*{\qedhere}
    \end{align*}
  \end{proof}
  Next, we bound the integers $d_{i,h}$ from~\Cref{eq:def-dih}, with $i \in [1..n]$ and $h \in [0..i-1]$.
  \begin{claim}
    \label{claim:bounds-on-dih}
    For every $i \in [1..n]$ and $h \in [0..i-1]$ we have 
    \[ 
      \abs{d_{i,h}} \leq 2^{h} (2A)^{i-1}B,
    \]
    where $A \coloneqq 4^{2^n} (E_0)^{2^n}$ and $B \coloneqq \Big(n \cdot
    2^{2^{n+3}+4c}(E_0)^{2^{n+2}} \ceil{\ln(H_0)} \cdot \deg(\cn,\phi_0)
    \Big)^{(6M_0k^{3M_0})^{n+1}}$.
  \end{claim}
  \begin{proof}
    By~\Cref{claim:less-ugly-bounds}, for every $r \in [0..n-1]$ and $i \in
    [r+2..n]$, $j_r,f_{r,i},\abs{\ell_{r,i}} \leq A$ and $\abs{g_r} \leq
    B$. 
    Following the definition of $d_{i,h}$ given in~\Cref{eq:def-dih}, 
    for every $i \in [1..n]$ and $h \in [0..i-2]$, $\abs{d_{i,h}}$ is
    bounded by the positive integer $D_{i,h}$ recursively defined as
    follows. For every $i \in [1..n]$,
    \begin{align*}
      D_{i,i-1} &\coloneqq B + \sum\nolimits_{h=0}^{i-2} D_{i-1,h} \cdot A, \notag\\
      D_{i,h} &\coloneqq (D_{i-1,h} + 1) \cdot A,
      &\tag*{$\text{ for every } h \in [0..i-2].$}
    \end{align*}
    Observe that $D_{1,0} = B$ and that, more generally, every $D_{i,h}$ is greater or equal to $B$.
    Since $B \geq 1$, for $h \neq i-1$ we have $D_{i,h} \leq 2 \cdot D_{i-1,h} \cdot A$. 
    To complete the proof, we show $D_{i,h} \leq 2^{h} (2A)^{i-1}B$ 
    by induction on $i$. 
    \begin{description}
      \item[base case: $i = 1$] In this case we only need to consider $D_{1,0}$, which as already states is equal to $B$. The base case thus follows trivially.
      \item[induction step: $i \geq 2$] Let $h \in [0..i-2]$. We consider two cases, depending on whether $h = i-1$. If $h \neq i-1$, then by definition of $D_{i,h}$ we have $D_{i,h} \leq 2 \cdot D_{i-1,h} \cdot A$. Then, from the induction hypothesis, 
        \begin{align*} 
          D_{i,h} \leq 2 \big( 2^{h} (2A)^{i-2} B\big)A 
          \leq 2^{h} (2A)^{i-1} B.
        \end{align*}
        If $h = i-1$, then by definition of $D_{i,h}$ we have  $D_{i,i-1} = B + \sum_{h=0}^{i-2} D_{i-1,h} \cdot A$. By applying the induction hypothesis, we obtain:
        \begin{align*} 
          D_{i,i-1} 
          &\leq B + \sum\nolimits_{h=0}^{i-2} (2^{h} (2A)^{i-2} B) \cdot A \leq B + 2^{i-2}A^{i-1}B \cdot \sum\nolimits_{h=0}^{i-2} 2^{h}\\ 
          ~\hspace{2cm}&\leq B + 2^{i-2}A^{i-1}B \cdot 2^{i-1} 
          \leq 2^{i-1} (2A)^{i-1}B.  
          &\tag*{\qedhere}
        \end{align*}
    \end{description}
  \end{proof}
  Together, \Cref{claim:quantifier-relativisation}
  and~\Cref{claim:bounds-on-dih} show that, whenever satisfiable,
  $\psi(u_1,\dots,u_n)$ (that is,~$\phi_0$) has a solution assigning to each
  variable an integer power of $\cn$ of the form $\cn^\beta$ with $\abs{\beta}
  \leq 2^{2n}A^nB$, where $A$ and $B$ are defined as
  in~\Cref{claim:bounds-on-dih}.
  We conclude the
  proof by simplifying this bound to improve its readability, obtaining the one
  in the statement. 
  Recall that $H \coloneqq \max\{8,\height(\psi)\}$,
  $D \coloneqq \deg(\psi)+2$  
  (where $\deg(\psi)$ also account for the degree of $\cn$),
  $E_0 = \max\{\deg(u_i,\psi) : i
  \in [1..n]\}$ 
  and $M_0$ is the number of monomials in a polynomial of $\psi$, 
  which can be crudely bounded as $D^{n+1}$ (the monomials also contain $\cn$).
  \begin{align*}
    \abs{\beta} 
      &\leq 2^{2n}A^nB\\
      &\leq 2^{2n}\big(4^{2^n} (E_0)^{2^n}\big)^n  \Big(n \cdot
      2^{2^{n+3}+4c}(E_0)^{2^{n+2}} \ceil{\ln(H_0)} \cdot \deg(\cn,\phi_0)
      \Big)^{(6M_0k^{3M_0})^{n+1}}\\
      &\leq 2^{2n}\big(4^{2^n} D^{2^n}\big)^n  \Big(n \cdot
      2^{2^{n+3}+4c}D^{2^{n+2}} \ceil{\ln(H)} \cdot D
      \Big)^{(6D^{n+1}k^{3D^{n+1}})^{n+1}}\\
      &\leq \Big(
      2^{4c}D^{2n(2^n+1) + n2^n + 2^{n+3} + \log(n) + 2^{n+2}+1} \ceil{\ln(H)}
      \Big)^{(6D^{n+1}k^{3D^{n+1}})^{n+1}}
      &\hspace{-0.3cm}\tag*{$\text{as } D \geq 2$}\\
      &\leq \Big(
      2^{4c}D^{18n2^n} \ceil{\ln(H)}
      \Big)^{(6D^{n+1}k^{3D^{n+1}})^{n+1}}
      &\hspace{-0.3cm}\tag*{$\text{as } n \geq 1$}\\
      &\leq \Big(
      2^{c} \ceil{\ln(H)}
      \Big)^{72 n2^n\log(D) \cdot (6D^{n+1}k^{3D^{n+1}})^{n+1}},
  \end{align*}
  and the exponent in the last expression can be upper bounded as 
  \begin{align*}
    72 n2^n\log(D) \cdot (6D^{n+1}k^{3D^{n+1}})^{n+1} \leq{} 
    2^{13n}D^{5n^2}(k^{3D^{n+1}})^{n+1} 
    \leq D^{2^5n^2}k^{D^{8n}}.
  \end{align*}
  Therefore, one can set $U \coloneqq (2^c \ceil{\ln(H)})^{D^{2^5 n^2} k^{D^{8n}}}$ 
  when defining $P_{\psi} \coloneqq [-U..U]$. 
  This concludes the proof of~\Cref{theorem:small-model-property}.
\end{proof}

\section{Proof of Theorem~\ref{theorem:result-root-barrier}} 
\label{sec:poly-evaluation}

This section completes the proof of~\Cref{theorem:result-root-barrier}.
Recall that~\Cref{theorem:result-root-barrier}(\ref{theorem:result-root-barrier:point3}) was proven in~\Cref{subsection:correctness-algorithm}.
All remaining cases of~\Cref{theorem:result-root-barrier}  
are shown by discussing polynomial-time Turing machines and polynomial root barriers
of the numbers under analysis.
\Cref{theorem:result-root-barrier}(\ref{theorem:result-root-barrier:point0}) and \Cref{theorem:result-root-barrier}(\ref{theorem:result-root-barrier:point1})
are proven in~\Cref{subsection:poly-evaluation-algebraic}. 
\Cref{theorem:result-root-barrier}(\ref{theorem:result-root-barrier:point2}) 
is proven in~\Cref{subsection:poly-evaluation-transcendental}.

\subsection{Proof of \Cref{theorem:result-root-barrier}(\ref{theorem:result-root-barrier:point0}) and 
\Cref{theorem:result-root-barrier}(\ref{theorem:result-root-barrier:point1}):~$\cn$ is algebraic.} 
\label{subsection:poly-evaluation-algebraic}
Let $\cn$ be a fixed algebraic number represented by $(q,\ell,u)$.
It is well known that algebraic numbers have a polynomial root barrier where the integer 
$k$ from~\Cref{theorem:general-result-root-barrier} equals $1$:

\begin{restatable}[{\cite[Theorem~A.1]{Bugeaud04}}]{thmC}{TheoremAlgRootBarrier}\label{theorem:alg-root-barrier}
  Let $\alg \in \R$ be a zero
  of a non-zero integer polynomial~$q(x)$,
  and consider a non-constant integer polynomial $p(x)$.
  Then, either ${p(\alg) = 0}$ or 
  ${\ln \abs{p(\alg)} \geq - \deg(q) \cdot \big(\ln(\deg(p)+1) + \ln \height(p)\big)
  - \deg(p) \cdot \big( \ln(\deg(q)+1) + \ln \height(q) \big)}$.
\end{restatable}

\noindent 
We now show how to construct a polynomial-time Turing machine computing 
$\cn$. The idea of the proof is performing 
a dichotomy search to refine the interval~$[\ell,u]$.


\begin{restatable}{lemma}{LemmaApproxAlgebraicBody}
  \label{lemma:approx-algebraic}
  Let~$\alg$ be a fixed algebraic number represented by $(q,\ell,u)$.
  There is a polynomial-time Turing machine computing $\alg$.
\end{restatable}
\begin{proof}
  We prove the lemma by showing that there is an algorithm that given as input
  $L \in \N$ written in unary computes in time polynomial in $L$ two rational
  numbers $\ell'$ and $u'$ such that $(q,\ell',u')$ is a representation of
  $\alg$ and $0 \leq u'-\ell' \leq 2^{-L}$.
  
  Following what is written in the last paragraph of the preliminaries (\Cref{sec:preliminaries}), 
  we can assume without loss of generality that either $\ell = u$ or $\cn \in (\ell,u)$ 
  and $(\ell,u) \cap \Z = \emptyset$, 
  which implies $u-\ell < 1$. 
  The algorithm further refine the interval $[\ell,u]$ to an accuracy that depends on $L$ by performing a dichotomy search: 
  \begin{algorithmic}[1]
    \While{$u-\ell > 2^{-L}$}\label{alg:ref-alg:line1}
      \State $m \gets \frac{\ell+u}{2}$
      \label{alg:ref-alg:line2}
      \If{$q(m) = 0$} 
        \myreturn $(q,m,m)$
        \label{alg:ref-alg:line3} 
      \EndIf
      \If{$\exists x : \ell < x < m \land q(x) = 0$}
        $u \gets m$
        \label{alg:ref-alg:line4}
      \Else \ $\ell \gets m$
        \label{alg:ref-alg:line5}
      \EndIf
    \EndWhile
    \State \myreturn $(q,\ell,u)$
  \end{algorithmic}

  \paragraph{\textit{Correctness of the algorithm}} 
  At each iteration of the
  \textbf{while} loop, after setting $m = \frac{u-\ell}{2}$, one of the following three cases holds: $\cn = m$, $\ell <
  \cn < m$, or $m < \cn < u$. Lines~\ref{alg:ref-alg:line3}
  and~\ref{alg:ref-alg:line4} check which of the three cases holds, since $\cn$ is the only root of $q(x)$ in the interval $[\ell,u]$.
  We can implement the test in line~\ref{alg:ref-alg:line4} by relying, e.g., on
  the procedure form~\autoref{theorem:basu}.

  \paragraph{\textit{Running time of the algorithm}} 
  Observe that at each iteration of the \textbf{while} loop the distance between
  $\ell$ and $u$ is halved. Since initially $u-\ell < 1$, this means that the
  \textbf{while} loop iterates at most $L$
  times. Furthermore, the procedure form~\autoref{theorem:basu} runs in polynomial
  time when the input formula has a fixed number of variables.
  Therefore, in order to argue that 
  the procedure runs in polynomial time 
  it suffices to track the growth of the numbers $\ell$ and $u$ across $L$ iterations of the while loop.

  For simplicity, below we assume~$\ell$, $u$ and $m$ to be standard programming variables, all
  storing pairs of (not necessarily coprime) integers representing rational numbers.  
  We also let $(a,b)$ and $(c,d)$ be the content of the variables $\ell$ and
  $u$, respectively, at the beginning of the algorithm. These two pairs encode
  the rationals $\frac{a}{b}$ and $\frac{c}{d}$. We assume $a,c \in \Z$ and $b,d
  \in \N_{\geq 1}$.
  This assumption implies that, for a given $\ell = (\ell_1,\ell_2)$ and $u = (u_1,u_2)$, line~\ref{alg:ref-alg:line2}
  of the algorithm assigns to $m$ the pair $(m_1,m_2)$ where
  \begin{equation}
    \label{eq:m1m2}
    m_1 \coloneqq {\textstyle\frac{\lcm(\ell_2,u_2)}{u_2}u_1 - \frac{\lcm(\ell_2,u_2)}{\ell_2}\ell_1} 
    \qquad\text{ and }\qquad
    m_2 \coloneqq 2 \cdot \lcm(\ell_2,u_2),
  \end{equation}
  One can easily check that then $m = \frac{u-\ell}{2}$ after 
  the execution of line~\ref{alg:ref-alg:line2}.
  
  We show the following loop invariant: 
  \begin{itemize}
    \item[] 
      \textit{After the $M$th iteration of the \textbf{while} loop, the variables $\ell$ and $u$
      store pairs from the set $S_M \coloneqq \{ (v_1,v_2) \in \Z \times \N :
          \abs{v_1} \leq 2^{M^2} \lcm(b,d)^{M}\max(\abs{a},\abs{b}), \text{ and } 
          v_2 \text{ divides } 2^{M} \lcm(b,d)
        \}$.}
  \end{itemize}
  Observe that both $(a,b)$ and $(c,d)$ belongs to $S_0$, 
  hence the invariant holds after the $0$th iteration of the
  \textbf{while} loop. 
  Consider now the $(M+1)$th iteration of the \textbf{while} loop, with
  $M \geq 0$. Let $(\ell_1,\ell_2) \in S_M$ and $(u_1,u_2) \in S_M$ be the pairs assigned to
  $\ell$ and $u$, respectively, at the beginning of this iteration. If the test
  performed in line~\ref{alg:ref-alg:line3} is successful, then the algorithm
  returns and we do not have anything to prove. Below, assume the test in
  line~\ref{alg:ref-alg:line3} to be unsuccessful, so that the algorithm
  completes the $(M+1)$th iteration of the loop. Let $(m_1,m_2)$ be the pair of
  integers computed in line~\ref{alg:ref-alg:line2}.
  At the end of the iteration of the loop, one of the following two possibilities occurs: 
  \begin{itemize}
    \item $\ell$ stores $(\ell_1,\ell_2)$ and $u$ stores $(m_1,m_2)$ (this occurs if the assignment in line~\ref{alg:ref-alg:line4} is executed), 
    \item $\ell$ stores $(m_1,m_2)$ and $u$ stores $(u_1,u_2)$ (this occurs if the assignment in line~\ref{alg:ref-alg:line5} is executed). 
  \end{itemize}
  To conclude the proof we show that $(m_1,m_2) \in S_{M+1}$. By~\Cref{eq:m1m2}, 
  ${m_2 = 2 \cdot \lcm(\ell_2,u_2)}$ 
  which, by induction hypothesis, divides $2 \cdot \lcm(2^{M} \lcm(b,d),2^{M} \lcm(b,d))$, i.e., $2^{M+1}\lcm(b,d)$.
  For $m_1$ instead, by applying the induction hypothesis we find:
  \begin{align*}
    \abs{m_1} 
      &= \abs{{\textstyle\frac{\lcm(\ell_2,u_2)}{u_2}u_1 - \frac{\lcm(\ell_2,u_2)}{\ell_2}\ell_1}} \leq \abs{2^M \lcm(b,d) \cdot u_1} + \abs{2^M \lcm(b,d) \cdot \ell_1}\\
      &\leq  2 \cdot \abs{2^M \lcm(b,d) \cdot 2^{M^2} \lcm(b,d)^{M}\max(\abs{a},\abs{b})}\\
      &\leq 2^{(M+1)^2} \lcm(b,d)^{M+1}\max(\abs{a},\abs{b}).
      &&\tag*{\qedhere}
  \end{align*}

\end{proof}


By applying~\Cref{theorem:general-result-root-barrier}(\ref{theorem:general-result-root-barrier:point1}),
\autoref{theorem:alg-root-barrier} and \Cref{lemma:approx-algebraic},
we deduce that the satisfiability problem for $\exists\R(\ipow{\cn})$ 
is in \twoexptime. However, for algebraic numbers it is possible to 
obtain a better complexity result (\expspace) by slightly 
modifying~Steps~II and~III of~\Cref{algo:main-procedure}.

\begin{proof}[Proof of~\Cref{theorem:result-root-barrier}(\ref{theorem:result-root-barrier:point1})]
Let $\phi$ be a formula in input of~\Cref{algo:main-procedure}, and
$\psi(u_1,\dots,u_n)$ be the formula obtained from $\phi$ after executing
lines~\ref{algo:line1}--\ref{algo:line6}. In lines~\ref{algo:line7}
and~\ref{algo:line8}, guess the integers~$g_1,\dots,g_n$ in binary, instead of
unary. These numbers have at most $m$ bits where, 
by~\autoref{theorem:basu} and \Cref{theorem:small-model-property}, $m$ is exponential
in~$\size(\phi)$. Let $g_i = \pm_{i} \sum_{j=0}^{m-1} d_{i,j} 2^j$, with
$d_{i,j} \in \{0,1\}$ and $\pm_i \in \{+1,-1\}$, so that $\cn^{g_i} =
\prod_{j=0}^{m-1}\cn^{\pm_i d_{ij}2^j}$. Note that the formula 
\[ 
  \gamma(x_0,\dots,x_{m-1}) \coloneqq q(x_0) = 0 \land \ell \leq x_0 \leq u \land \textstyle\bigwedge_{i=1}^{m-1} x_i = (x_{i-1})^2
\]
has a unique solution: for every $j \in [0..m-1]$, $x_j$ must be equal to
$\cn^{2^j}$. The formula $\psi$ is therefore equisatisfiable with the formula
$\psi' \coloneqq \psi\sub{x_0}{\cn} \land \gamma \land \bigwedge_{i=1}^n u_i =
\prod_{j=0}^{m-1}x_j^{\pm_id_{ij}}$, which (after rewriting $u_i =
\prod_{j=0}^{m-1}x_j^{\pm_id_{ij}}$ into $u_i \prod_{j=0}^{m-1}x_j^{d_{ij}} = 1$
when $\pm_i = -1$) is a formula from the existential theory of the reals of size
exponential in $\size(\phi)$. Since the
satisfiability problem for the existential theory of the reals is in
\pspace~\cite{Canny88}, we conclude that checking whether $\psi'$ is satisfiable
can be done in~\expspace. Accounting for Steps I and II, we thus obtain a
procedure running in non-deterministic exponential space (because of the guesses
in lines~\ref{algo:line7} and~\ref{algo:line8}), which can be determinised by
Savitch's theorem~\cite{Savitch70}.
\end{proof}

Similarly, we can modify~Steps~II and~III of~\Cref{algo:main-procedure} to
obtain a \nexptime upper bound for the case of $\cn$ natural number. For this
result we use the notion of fewnomial, also found in the literature as sparse
polynomial or lacunary polynomial~\cite{Khovanskii91}. A \emph{fewnomial} is a
representation of polynomial in which coefficients and degrees of monomials are
both represented in binary. More precisely, the fewnomial representation of a
polynomial ${p(x) = \sum_{i=1}^{k} a_i x^{d_i}}$, where $d_1,\dots,d_k$ are
non-negative integers and $d_i>d_{i+1}$ for $i\in[1..k-1]$, is given by the
sequence $((a_1,d_1),\dots,(a_k,d_k))$, where all $a_i$ and $d_i$ are
represented in binary encoding. This representation is particularly interesting
in the case of polynomials with very few monomials, hence the name. Observe that
$\deg(p)$ can be exponential in the size of $((a_1,d_1),\dots,(a_k,d_k))$. For
our purpose, we need a polynomial time algorithm for evaluating the sign of a
fewnomial on a given natural number. Both~\cite{MyasnikovUW12} and~\cite{HitarthMS26} provide such an
algorithm, although in the context of more general objects. Below we provide a
simpler, standalone procedure tailored to fewnomials.%

\begin{proposition}\label{prop:sign-fewnomial}
  Fix $n \in \N_{\geq 1}$. 
  There is a polynomial time algorithm that given in input an integer polynomial $p(x)$ 
  in fewnomial representation, returns the sign of $p(n)$.
\end{proposition}
\begin{proof}
  Let $((a_1,d_1),\dots,(a_k,d_k))$ be the fewnomial representation of $p(x)$, 
  where $d_i > d_{i+1}$ for every~$i\in[1..k-1]$. The following algorithm computes the 
  sign or $p(n)$:
  \begin{algorithmic}[1]
    \State $P \gets ((a_1,d_1),\dots,(a_k,d_k))$\label{alg:fewnomial_eval:line0}
    \While{$P$ is a sequence $((b_1,g_1),\dots,(b_\ell,g_\ell))$ with $\ell > 1$}
    \label{alg:fewnomial_eval:line1}
    \If{$\abs{b_1} \cdot n^{g_1-g_2} >\sum_{i=2}^\ell \abs{b_i}$} 
        \myreturn the sign of $b_1$\label{alg:fewnomial_eval:line3}
      \Else \ $P \gets((b_1\cdot n^{g_1-g_2}+b_2,g_2),\dots, (b_\ell,g_\ell))$
      \label{alg:fewnomial_eval:line4}
      \EndIf
    \EndWhile
    \State \myreturn the sign of $b_1$\label{alg:fewnomial_eval:line5}
  \end{algorithmic}

  \paragraph{\textit{Correctness of the algorithm}} Observe that at any point
  during its execution, $P$ stores a fewnomial representation
  $((b_1,g_1),\dots,(b_\ell,g_\ell))$ of a polynomial $q$ such that $p(n) =
  q(n)$. Indeed, $P$ represents $p$ when it is initialized in
  line~\ref{alg:fewnomial_eval:line0},  
  and the update performed in line~\ref{alg:fewnomial_eval:line4} preserve this
  property. Furthermore, observe that at each iteration of the \textbf{while}
  loop, the length of~$P$ shrinks by one. After~$k-1$ iterations in which the
  test of line~\ref{alg:fewnomial_eval:line3} fails, the loop exists; when this
  happens $P$ is a monomial, and the sign of $p(n)$ is given by the sign of the
  only coefficient $b_1$ (line~\ref{alg:fewnomial_eval:line5}). Therefore, to
  conclude that the algorithm is correct, it suffices to show that whenever the
  guard of the \textbf{if} statement in line~\ref{alg:fewnomial_eval:line3}
  holds, the sign of $q(n)$ corresponds to the sign of $b_1$. This is easy to
  see: if $\abs{b_1} \cdot n^{g_1-g_2} >\sum_{i=2}^\ell \abs{b_i}$, then $b_1
  \cdot n^{g_1}$ is larger, in absolute value, than $(\sum_{i=2}^\ell \abs{b_i})
  \cdot n^{g_2}$. But since the degrees $g_1,\dots,g_\ell$ are given in
  descending order this means $\abs{b_1} \cdot n^{g_1} > \abs{\sum_{i=2}^\ell
  \abs{b_i} \cdot n^{g_i}}$, and so the signs of
  $q(n)$ and $b_1$ coincide.


  \paragraph{\textit{Running time of the algorithm}}
  We prove the following invariant:
  \begin{center}
    \begin{minipage}{0.9\linewidth}
      \textit{Throughout the procedure, all sequences $((b_1,g_1),\dots,(b_\ell,g_\ell))$ 
      stored in $P$ are such that $\abs{b_1} \leq 2 \cdot \sum_{i=1}^k\abs{a_i}$, 
      $\sum_{j=2}^\ell \abs{b_j} \leq \sum_{i=1}^k\abs{a_i}$
      and $g_1,\dots,g_\ell$ are among $d_1,\dots,d_k$.}
    \end{minipage}
  \end{center}
  The invariant is clearly true in line~\ref{alg:fewnomial_eval:line0}, 
  and so it suffices to show that if $((b_1,g_1),\dots,(b_\ell,g_\ell))$ 
  satisfies the properties in the invariant, 
  then so does the sequence $((b_1\cdot n^{g_1-g_2}+b_2,g_2),\dots, (b_\ell,g_\ell))$ 
  computed in line~\ref{alg:fewnomial_eval:line3}. 
  The part of the invariant that concerns the degrees $g_1,\dots,g_\ell$ is 
  trivial to check. For the magnitude of $b_1,\dots,b_\ell$, 
  it suffices to bound $\abs{b_1\cdot n^{g_1-g_2}+b_2}$. 
  Line~\ref{alg:fewnomial_eval:line3} is only executed when 
  $\abs{b_1} \cdot n^{g_1-g_2} \leq \sum_{i=2}^\ell \abs{b_i}$, and so
  \begin{equation*}
    \abs{b_1 n^{g_1-g_2}+b_2}
    \leq \abs{b_1 n^{g_1-g_2}} + \abs{b_2} 
    \leq \sum\nolimits_{i=2}^\ell \abs{b_i} + \abs{b_2}\leq 2 \cdot \sum\nolimits_{i=1}^k\abs{a_i}.
  \end{equation*}
  Hence, all numbers computed by the algorithm have polynomial bit
  size. Since the \textbf{while} loop terminates after at most $k-1$ iterations,
  to conclude that the algorithm runs in polynomial time 
  it now suffices to discuss how to perform the test $\abs{b_1} \cdot n^{g_1-g_2} >\sum_{i=2}^\ell \abs{b_i}$ in polynomial time. We translate $\sum_{i=2}^\ell \abs{b_i}$ in base-$n$ encoding.
  Let $L$ be the length of this encoding, which is roughly $\ceil{\log_n(\sum_{i=2}^\ell \abs{b_i})}$. 
  If $L < g_1-g_n$, then the test passes if and only if $b_1 \neq 0$. 
  Otherwise, $n^{g_1-g_2}$ is smaller than $\sum_{i=2}^\ell \abs{b_i}$, and so we can compute $g_1-g_2$ in unary encoding (in polynomial time), translate also $\abs{b_1}$ in base-$n$, and perform the comparison directly.
\end{proof}

\begin{proof}[Proof of~\Cref{theorem:result-root-barrier}(\ref{theorem:result-root-barrier:point0})]
  Let $\phi$ be a formula from $\exists\R(n^\Z)$ with $n\in \N_{\geq 1}$ fixed. 
  If $n = 1$ then we can simply replace every occurrence of the predicate $n^{\Z}(x)$ 
  with $x = 1$ and rely on the procedure in~\autoref{theorem:basu} to check
  the satisfiability of $\phi$. For $n \geq 2$, the algorithm works as follows.
  We execute lines~\ref{algo:line1}--\ref{algo:line6} of~\Cref{algo:main-procedure}.
  By~\autoref{theorem:basu} and~\Cref{theorem:small-model-property} the resulting formula $\psi(u_1,\dots,u_k)$ is such that
  its degree and height are, respectively, exponential and doubly exponential in size($\phi$).
  In exponential time in size($\phi$), we convert all polynomials in $\psi$ to fewnomials (in particular, updating their degrees to a binary representation).
  We perform the guesses from lines~\ref{algo:line7} and~\ref{algo:line8} in binary
  and substitute each $u_i$ with $x^{g_i}$, where $x$ is a fresh variable, again 
  keeping polynomials encoded in fewnomial representation. 
  After these substitutions, we obtain a univariate formula $\gamma(x)$ 
  of size exponential respect to the size of $\phi$, 
  and in which polynomials are encoded as fewnomials.
  By \Cref{prop:sign-fewnomial}, we can evaluate the truth of each inequality in $\gamma(n)$ 
  in time exponential in size($\phi)$. 
  Again in exponential time, we can then evaluate the Boolean structure of $\gamma$, 
  finding whether $n$ is a solution. 
  Overall, the algorithm runs in non-deterministic exponential time (because of the guesses
  in lines~\ref{algo:line7} and~\ref{algo:line8}). 
  Its correctness follows directly from the correctness of \Cref{algo:main-procedure}.
\end{proof}

\subsection{Proof of \Cref{theorem:result-root-barrier}(\ref{theorem:result-root-barrier:point2}):~$\cn$ among some classical transcendental numbers}
\label{subsection:poly-evaluation-transcendental}

By~\Cref{theorem:general-result-root-barrier}(\ref{theorem:general-result-root-barrier:point2}),
it suffices to show that all bases considered in~\Cref{theorem:result-root-barrier}(\ref{theorem:result-root-barrier:point2}) \textbf{(1)}~are
polynomial-time computable, and \textbf{(2)}~have a
polynomial root barrier. Below we summarise how (and where) these two properties are proven for each base appearing in~\Cref{theorem:result-root-barrier}(\ref{theorem:result-root-barrier:point2}):%
\begin{description}
  \item[case: $\cn = \pi$]~\\
  \textit{Polynomial-time Turing machine:} By~\autoref{lemma:poly-time-pi}.\\
  \textit{Polynomial root barrier:} See~\Cref{table:transcendence-degrees}. 
  \item[case: $\cn = e^\pi$]~\\
  \textit{Polynomial-time Turing machine:} By~\autoref{lemma:poly-time-pi} and~\Cref{lemma:poly-time-exp}.\\
  \textit{Polynomial root barrier:} See~\Cref{table:transcendence-degrees}.
  \item[case: $\cn = e^\eta$]~\\
  \textit{Polynomial-time Turing machine:} By~\Cref{lemma:approx-algebraic} and~\Cref{lemma:poly-time-exp}.\\
  \textit{Polynomial root barrier:} See~\Cref{table:transcendence-degrees}.
  \item[case: $\cn = \ln(\alg)$ with $\alg > 0$]~\\
  \textit{Polynomial-time Turing machine:} By~\Cref{lemma:approx-algebraic} and~\Cref{lemma:poly-time-log}.\\
  \textit{Polynomial root barrier:} See~\Cref{table:transcendence-degrees}. 
  \item[case: $\cn = \alg^\eta$ with $\alg > 0$]~\\
  \textit{Polynomial-time Turing machine:} Consider $e^{\eta \cdot \ln(\alg)}$, and construct the Turing machine by applying~\Cref{lemma:approx-algebraic}, \Cref{lemma:turing-machine-products} and~\Cref{lemma:poly-time-exp}.\\
  \textit{Polynomial root barrier:} Use~\Cref{lemma:check-rationality} to check if $\eta$ is rational. If it is, apply~\Cref{lemma:representation-power-of-algebraic} to obtain a representation of the algebraic number $\alg^\eta$, 
  and~\autoref{theorem:alg-root-barrier} to obtain a root barrier for it. If instead $\eta$ is irrational, use~\Cref{table:transcendence-degrees}.
  \item[case: $\cn = \frac{\ln(\alg)}{\ln(\beta)}$ with $\alg,\beta > 0$ (and $\beta \neq 1$)]~\\
  \textit{Polynomial-time Turing machine:} Use~\Cref{lemma:approx-algebraic} and~\Cref{lemma:poly-time-log} and~\Cref{lemma:turing-machine-reciprocal}.\\
  \textit{Polynomial root barrier:} Use~\Cref{lemma:mult-independence-rationality} and \Cref{remark:multiplicative-independence}.
\end{description}

\subsubsection{Turing machines.}
For $\pi$, we rely on the Turing machine from~\cite{Bailey1997OnTR}:

\begin{thmC}[\cite{Bailey1997OnTR}]\label{lemma:poly-time-pi}
 One can construct a polynomial-time Turing machine computing $\pi$.
\end{thmC}

\noindent 
For the remaining cases, we first need a way of composing approximations:
\begin{lemma}\label{lemma:auxiliary-for-approximations}
Let $\delta(x)$ be an integer polynomial, and $p \colon \R \to \R$ be a function
such that $\delta(x) \cdot p(x)$ equals an integer polynomial $q(x)$. Let $d,h
\in \N$ such that $d \geq \max(\deg(\delta),\deg(q))$ and
$h \geq \max(\height(\delta),\height(q))$. Let $r \in \R$ and $K \geq 1$ such that
$\max(\abs{r},\abs{p(r)}) \leq K$. Consider two natural numbers $L$ and $M$
satisfying $M \geq L + \log(h+1) + (2 d +1)(\log(K+1))$. Then, for every $r^* \in \R$
satisfying $\delta(r^*) \geq 1$, if $\abs{r-r^*} \leq 2^{-M}$ then
$\abs{p(r)-p(r^*)} \leq 2^{-L}$. 
\end{lemma}
\begin{proof}
  By applying~\Cref{lemma:approx-univ-polynomial} to both $\delta$ and $q$, 
  from $\abs{r-r^*} \leq 2^{-M}$ 
  we derive 
  \begin{align*}
    \abs{\delta(r)-\delta(r^*)} \leq 2^{-N}
    \text{ \ and \ }
    \abs{q(r) - q(r^*)} \leq 2^{-N},
  \end{align*}
  where $N \coloneqq L + \log(K+1)$. 
  We show that these two inequalities imply $\abs{p(r)-p(r^*)} \leq 2^{-L}$.
  Define $\epsilon \coloneqq \delta(r)-\delta(r^*)$. The following chain of implications holds 
  \begin{align*}
    &\abs{q(r) - q(r^*)} \leq 2^{-N}\\
    \implies{}&\abs{\delta(r) \cdot p(r) - \delta(r^*) \cdot p(r^*)} \leq 2^{-N}
    &\tag*{by hypotheses}\\
    \implies{}&\abs{(\delta(r^*)+\epsilon) \cdot p(r) - \delta(r^*) \cdot p(r^*)} \leq 2^{-N}
    &\tag*{by def.~of $\epsilon$}\\
    \implies{}&\abs{\delta(r^*)(p(r)-p(r^*))} \leq 2^{-N} + \abs{\epsilon \cdot p(r)}\\
    \implies{}&\abs{\delta(r^*)(p(r)-p(r^*))} \leq 2^{-N} + 2^{-N}\abs{p(r)}
    &\tag*{bound on $\epsilon$}\\
    \implies{}&\abs{\delta(r^*)(p(r)-p(r^*))} \leq 2^{-N}(K+1)
    &\tag*{bound on $\abs{p(r)}$}\\
    \implies{}&\abs{(p(r)-p(r^*))} \leq \frac{2^{-N}(K+1)}{\abs{\delta(r^*)}}\\
    \implies{}&\abs{(p(r)-p(r^*))} \leq 2^{-N}(K+1)
    &\tag*{since $\delta(r^*) \geq 1$.}\\
    \implies{}&\abs{(p(r)-p(r^*))} \leq 2^{-L} 
    &\tag*{from the def.~of $N$\qedhere}
  \end{align*}
\end{proof}

To obtain approximations of~$e^r$, we rely on truncations of standard power series.
\begin{lemma}\label{lemma:approx-exp}
  Let $r \in \R$ and $k \geq 1$ with $\abs{r} \leq k$, 
  and let~$t_n(x) \coloneqq \sum_{j=0}^n \frac{x^j}{j!}$. 
  For every $L,M \in \N$ satisfying $M \geq L + 8 k^2$, we have $\abs{e^r - t_M(r)} \leq 2^{-L}$.
\end{lemma}
\begin{proof}
  Following the identity $e^ x = \sum_{j=0}^\infty \frac{x^j}{j!}$, 
  whose right-hand side is the Maclaurin series for $e^x$ (see, e.g.,~\cite[Equation~4.2.19]{Olver10}), we have
  \begin{equation*}\label{inequality-exp-taylor}
    \begin{aligned}
      \abs{e^r-t_M(r)}&=\textstyle\abs{\sum_{j=M+1}^\infty \frac{r^j}{j!}} = 
      \abs{\sum_{j=0}^\infty \frac{r^{M+1+j}}{(M+1+j)!} } = 
      \abs{r^{M+1}\sum^\infty_{j=0}\frac{r^j}{(M+1+j)!}} = \\
      &=\textstyle\abs{\frac{r^{M+1}}{(M+1)!}\sum_{j=0}^\infty
      \frac{r^j}{(M+1+j)\cdot{\dots}\cdot(M+2)}}\leq
      \frac{k^{M+1}}{(M+1)!}e^k\,,
    \end{aligned}
  \end{equation*}
  where in the last inequalities we used $\abs{r} \leq k$. 
  Let us show that the hypothesis $M \geq L + 8k^2$ in the claim implies $\frac{k^{M+1}}{(M+1)!}e^k\leq \frac{1}{2^L}$, 
  concluding the proof. 
  Note that $M \geq 2^3k^2$ implies $\log(\frac{M}{2}) \geq 2 + 2 \log(k)$, and so $\frac{\log(\frac{M}{2})}{2} - \log(k) \geq 1$.
  We have the following chain of implications:
  {\allowdisplaybreaks%
  \begin{align*}
    &M\geq L+2^3k^2\\
    \implies{}&M\geq L+\log(k)+k\log(e) \quad\text{and}\quad M\geq 2^3k^2
    & \tag*{since $2^3k^2 \geq \log(k) + k\log(e)$}\\
    \implies{}&M\Big( \frac{\log(\frac{M}{2})}{2}-\log(k) \Big) \geq
    L+\log(k)+k\log(e)\\
    \implies{}&\frac{M}{2}\log\left( \frac{M}{2} \right)\geq 
    L+M\log(k) + \log(k) + k\log(e)\\
    \implies{}&\log((M+1)!)\geq L+(M+1)\log(k) + k\log(e)
    &\tag*{{since }$(M+1)!\geq \frac{M}{2}^{\frac{M}{2}}$}\\
    \implies{}&(M+1)!\geq 2^Lk^{M+1}e^k\\
    \implies{}& \frac{k^{M+1}}{(M+1)!}e^k\leq \frac{1}{2^L}.
    &&\tag*{\qedhere}
  \end{align*}}
\end{proof}

\Cref{lemma:approx-exp} gives us a way to handle the cases $e^\pi$ and $e^\eta$ with $\eta$ algebraic.

\begin{restatable}{lemma}{LemmaPolyTimeTMExp}
  \label{lemma:poly-time-exp}
  Given a polynomial-time Turing machine computing $r \in \R$
  one can construct a polynomial-time Turing machine computing $e^r$.
\end{restatable}

\begin{proof}
  Let $T$ be the polynomial-time Turing machine computing $r$.
  Following~\Cref{lemma:approx-exp}, for $n \in \N$ we define $t_n(x) \coloneqq
  \sum_{j=0}^{n} \frac{x^j}{j!}$, which we see as a polynomial with rational coefficients encoded in binary (as a pair of integers).
  The pseudocode of the Turing machine for computing $e^r$ is the following:
  \begin{algorithmic}[1]
    \Require A natural number $n$ written in unary.
    \Ensure A rational $b$ (given as a pair of integers written in binary) such that $\abs{e^r-b} \leq 2^{-n}$.%
    \State \textbf{let} $J \coloneqq \abs{T_0}+1$ 
    \Comment{recall: $T_0$ computed in constant time, and $\abs{r} \leq \abs{T_0}+1$}
    \label{approxexp:line1}
    \State \textbf{let} $M \coloneqq n + 1 + 8\ceil{J}^2$
    \label{approxexp:line2}
    \State \textbf{let} $N \coloneqq n +1 + 9M^2 (\ceil{\log(J)}+1)$
    \label{approxexp:line3}
    \State \myreturn evaluation of $t_M(T_{N})$
    \label{approxexp:line4}
    \Comment{$T_N$ computed in polynomial-time in $N$}
  \end{algorithmic}
  Below, we prove that this algorithm computes $e^r$ and that it runs in
  time $\poly(n)$.

  \paragraph{\textit{Correctness of the algorithm}}
  From~\Cref{lemma:approx-exp}, we have $\abs{e^r-t_M(r)} \leq \frac{1}{2^{n+1}}$, 
  where $M$ is the value defined in line~\ref{approxexp:line2}.
  Below, we apply~\Cref{lemma:auxiliary-for-approximations} in order to conclude that 
  $\abs{t_M(r)-t_M(T_N)} \leq \frac{1}{2^{n+1}}$.
  Observe that this concludes the proof of correctness, since we get: 
  \[ 
    \abs{e^r-t_M(T_N)} = \abs{e^r-t_M(r)+t_M(r)-t_M(T_N)} 
    \leq \abs{e^r-t_M(r)}+ \abs{t_M(r)-t_M(T_N)} 
    \leq  2^{-n}.%
  \]

  Let $\delta(x) \coloneqq M!$ ($\delta$ is a constant integer polynomial) and
  $q(x) \coloneqq \sum_{j=0}^M ((j+1) \cdot {\dots} \cdot M \cdot x^j)$. Observe
  that $\delta(x) \cdot t_M(x)$ equals $q(x)$, and that $\delta(T_N) \geq 1$. We
  have $\max(\deg(\delta),\deg(q)) \leq~M$ and $\max(\height(\delta),\height(q))
  \leq M!\,$. Lastly, let us define $K \coloneqq 3J^M$, so that
  $\max(\abs{r},\abs{t_M(r)}) \leq K$. (Indeed, observe that $\abs{t_M(r)} =
  \abs{\sum_{j=0}^M \frac{r^j}{j!}} \leq \abs{r}^M \sum_{j=0}^M \frac{1}{j!}
  \leq e \abs{r}^M \leq 3J^M$.) By~\Cref{lemma:auxiliary-for-approximations},
  $\abs{t_M(r)-t_M(T_N)} \leq \frac{1}{2^{n+1}}$ holds as soon as $\abs{r - T_N}
  \leq \frac{1}{2^{L}}$, where $L$ is any integer satisfying $L \geq  n+1 +
  \log(M!+1) + (2 M +1)(\log(K+1))$. Since $\abs{r - T_N} \leq \frac{1}{2^{N}}$,
  it then suffices to show that $N$ defined in line~\ref{approxexp:line3}
  corresponds to such an integer $L$:
  {\allowdisplaybreaks%
  \begin{align*}
    &n+1 + \log(M!+1) + (2 M +1)(\log(K+1))\\
    \leq{}& n+1 + \log(M!+1) + (2 M +1)(\log(4J^M))
      &\tag*{by def.~of $K$}\\
    \leq{}& n+1 + M^2 + (2 M +1)(M \log(J)+2)
      &\hspace{-16pt}\tag*{since $M \geq 1$, and so $\log(M!+1) \leq M^2$}\\
    \leq{}& n +1 + 9M^2 (\ceil{\log(J)}+1) = N.
  \end{align*}}

  \paragraph{\textit{Running time of the algorithm}}
  Line~\ref{approxexp:line1} does not depend on the input $n$ and thus its
  computation takes constant time. Lines~\ref{approxexp:line2}
  and~\ref{approxexp:line3} compute in polynomial time the numbers $M$ and $N$,
  which are written in unary and have size in~$O(n^2)$. 
  
  To conclude the proof, we show that the computation performed in
  line~\ref{approxexp:line4} takes time polynomial in $n$. This line first
  computes the number $T_N$, which can be done in time $\poly(n)$ because of the
  $O(n^2)$ bound on $N$. This also means that $T_N$ is of the form
  $\frac{\ell_1}{\ell_2}$ where $\ell_1,\ell_2$ are integers written in binary
  with bit size polynomial in $n$. The last step of the procedure is to evaluate
  the expression $\sum_{j=0}^{M} \frac{(\ell_1)^j}{(\ell_2)^j \cdot j!}$, which
  equals the rational number~$\frac{b_1}{b_2}$ where
  \begin{align*}
    b_1 \coloneqq \sum\nolimits_{j=0}^M (\ell_1)^j \cdot
        (\ell_2)^{M-j} \cdot (j+1) \cdot {\dots} \cdot M,
    \qquad
    b_2 \coloneqq (\ell_2)^M \cdot M!\,.
  \end{align*}
  The algorithm returns~$\frac{b_1}{b_2}$. Clearly,
  computing $b_1$ and $b_2$ using the above expressions can be done in
  polynomial time in $n$, just following the arithmetic operations. This means
  that also line~\ref{approxexp:line4} takes time polynomial in $n$,
  concluding the proof. 
\end{proof}

Let us move to the case of logarithms. Proceding similarly to the case $\cn = e^r$,
we establish a lemma on approximations of $\ln(r)$ by truncation of standard power series.

\begin{lemma}\label{lemma:approx-log}
  Let $r > 0$, 
  and let~$t_n(x) \coloneqq 2 \cdot \sum_{j=0}^n \big(\frac{1}{2j+1} \big(\frac{x-1}{x+1}\big)^{2j+1}\big)$. Consider $L,M \in \N$. 
  If~$r = 1$ or $M \geq (L+\log\abs{\ln(r)}) \big({-}2\log\abs{\frac{r-1}{r+1}}\big)^{-1}$, then $\abs{\ln(r) - t_M(r)} \leq 2^{-L}$.
\end{lemma}
\begin{proof}
  If $r = 1$, observe that $\ln(r) = 0$ 
  and $t_n(r) = 0$ for every $n \in \N$, 
  so $\abs{\ln(r)-t_M(r)} = 0$ and the statement trivially follows. 

  Below, assume $r \neq 1$.
  We follow the identity $\ln(x) = 2 \sum_{j=0}^\infty \big(\frac{1}{2j+1} \big(\frac{x-1}{x+1}\big)^{2j+1}\big)$, which holds for every $x > 0$, see~\cite[Equation~4.6.4]{Olver10}.
  We have:
  {\allowdisplaybreaks
  \begin{align*}
          &\abs{\ln(r)-t_M(r)}\\
      ={} &\abs{\sum_{j=M+1}^\infty\frac{1}{2j+1}
      \left( \frac{r-1}{r+1} \right)^{2j+1}}\\
      ={} &\abs{\sum_{j=0}^\infty\frac{1}{2j+2M+3}
      \left( \frac{r-1}{r+1} \right)^{2j+2M+3}}\\ 
      \leq{}&\abs{\left( \frac{r-1}{r+1} \right)^{2M+2}
      \sum_{j=0}^\infty\frac{1}{2j+1}\left( \frac{r-1}{r+1} \right)^{2j+1}}
      &\begin{minipage}{0.35\linewidth}
          note: $\frac{r-1}{r+1}, \Big( \frac{r-1}{r+1} \Big)^{3},\Big( \frac{r-1}{r+1} \Big)^{5},\dots$\\
          all have the same sign
      \end{minipage}\\ 
      \leq{}&
      \frac{1}{2}\abs{\frac{r-1}{r+1}}^{2M+2}\abs{\ln(r)}.
  \end{align*}}
  Let us now show that the hypothesis 
  $M \geq (L+\log\abs{\ln(r)}) \big({-}2\log\abs{\frac{r-1}{r+1}}\big)^{-1}$ in the statement of the lemma implies $\frac{1}{2}\abs{\frac{r-1}{r+1}}^{2M+2}\abs{\ln(r)} \leq \frac{1}{2^L}$, concluding the proof. Below, note that $r>0$ and $r \neq 1$ imply 
  that $\abs{\frac{r-1}{r+1}} \in (0,1)$, so $\log\abs{\frac{r-1}{r+1}} < 0$.
  \allowdisplaybreaks{
  \begin{align*}
    &M\geq(L+\log(\abs{\ln(r)}))\left( -2\log\abs{\frac{r-1}{r+1}} \right)^{-1}\\
    \implies{}&
    2M+2\geq \frac{L+\log\abs{\ln(r)}}{-\log\abs{\frac{r-1}{r+1}}}\\
    \implies{}&(2M+2)\log\abs{\frac{r-1}{r+1}}+\log\abs{\ln(r)}\leq -L
    &\tag*{since $\textstyle\log\abs{\frac{r-1}{r+1}} < 0$}\\
    \implies{}&\abs{\frac{r-1}{r+1}}^{2M+2}\abs{\ln(r)}\leq 2^{-L}\\
    \implies{}&\frac{1}{2}\abs{\frac{r-1}{r+1}}^{2M+2}\abs{\ln(r)}\leq 2^{-L}.
    &&\tag*{\qedhere}
  \end{align*}
  }
\end{proof}



The existence of polynomial time Turing machines for $\ln(r)$ follows:

\begin{restatable}{lemma}{LemmaPolyTimeTMLog}
  \label{lemma:poly-time-log}
  Given a polynomial-time Turing machine computing $r \in \R$, 
  if $r > 0$, one can construct a polynomial-time Turing machine computing $\ln(r)$.
\end{restatable}

\begin{proof}
  Let $T$ be the polynomial-time Turing machine computing $r > 0$.
  Following~\Cref{lemma:approx-log}, for $n \in \N$ we define $t_n(x) \coloneqq
  2 \cdot \sum_{j=0}^n \big(\frac{1}{2j+1} \big(\frac{x-1}{x+1}\big)^{2j+1}\big)$, 
  in which we see the rational numbers $\frac{1}{2j+1}$ as encoded in binary (as a pair of integers).
  The pseudocode of the Turing machine for computing $\ln(r)$ is the following:
  \begin{algorithmic}[1]
    \Require A natural number $n$ written in unary.
    \Ensure A rational $b$ (given as a pair of integers written in binary) s.t.~$\abs{\ln(r)-b} \leq 2^{-n}$.%
    \State \textbf{let} $k$ be the smallest natural number 
    such that $\frac{1}{2^k} < T_k$.
    \label{approxln:line1}
    \Statex \Comment{recall: $k$ and $T_k$ are computed in constant time}
    \State \textbf{let} $L \coloneqq T_k-\frac{1}{2^k}$
    \label{approxln:line2}
    \State \textbf{let} $U \coloneqq T_k+\frac{1}{2^k}$ 
    \Comment{note: $0 < L \leq r \leq U$}
    \label{approxln:line3}
    \State \textbf{let} $Z_1 \coloneqq \ceil{\max({\abs{\ln(L)}},{\abs{\ln(U)}})}$ 
    \label{approxln:line4}
      \Comment{$Z_1$ is a positive integer}
    \State \textbf{let} $Z_2 \coloneqq 1+\min(-\abs{\frac{L-1}{L+1}},-\abs{\frac{U-1}{U+1}})$ 
      \Comment{$Z_2$ is a positive rational number}
    \label{approxln:line5}
    \State \textbf{let} $M \coloneqq \ceil{\frac{n+1+Z_1}{2 \cdot Z_2}}$
      \Comment{$M$ is a positive integer written in unary}
    \label{approxln:line6}
    \State \textbf{let} $N \coloneqq n + 2 + 15M \cdot \ceil{\log(U + 4M)}$
      \Comment{$N$ is a positive integer written in unary}
    \label{approxln:line7}
    \State \myreturn evaluation of $t_M(\abs{T_{N}})$
    \label{approxln:line8}
    \Comment{$\abs{T_N}$ computed in polynomial-time in $N$}
  \end{algorithmic}
  Below, we prove that this algorithm computes $\ln(r)$ and that it runs in
  polynomial time with respect to the input $n$.

  \paragraph{\textit{Correctness of the algorithm}}
  We start with three observations:
  \begin{itemize}
    \item the number $k$ computed in line~\ref{approxln:line1} exists, since
      $\lim_{n \to \infty} T_k = r > 0$ whereas $\lim_{n \to \infty}
      \frac{1}{2^k} = 0$.  
    \item the values $Z_1$ and $Z_2$ are properly defined and positive, because $U > L > 0$, which in turns implies that also $M$ and $N$ are properly defined. 
    To prove that $Z_2 > 0$, observe that for every $y \geq 0$ we have $\abs{\frac{y-1}{y+1}} < 1$, hence $1 - \abs{\frac{y-1}{y+1}} > 0$.
    \item The Turing machine that on input $n$ returns $\abs{T_n}$ is a machine
    running in polynomial time and computing $r$. The latter property follows
    from the fact that $r > 0$ and therefore, for every $n \in \N$, if $T_n < 0$
    we get a better accuracy by considering $\abs{T_n}$ instead. 
    Note that this machine is used in line~\ref{approxln:line8}.
  \end{itemize}
  Below, we show \textbf{(1)} that $\abs{\ln(r)-t_M(r)} \leq \frac{1}{2^{n+1}}$,
  where $M$ is the value defined in line~\ref{approxln:line6}, 
  and \textbf{(2)} that $\abs{t_M(r)-t_M(\abs{T_N})} \leq \frac{1}{2^{n+1}}$, 
  where $N$ is the value defined in line~\ref{approxln:line7}.
  Note that this concludes the proof of correctness, since we get: 
  \begin{align*}
    \abs{\ln(r)-t_M(\abs{T_N})} &= \abs{\ln(r)-t_M(r)+t_M(r)-t_M(\abs{T_N})}\\
    &\leq \abs{\ln(r)-t_M(r)}+ \abs{t_M(r)-t_M(\abs{T_N})} 
    \leq  2^{-n}.%
  \end{align*}
  \begin{enumerate}[label=\textbf{(\arabic*)}]
    \item \textbf{proof of $\abs{\ln(r)-t_M(r)} \leq \frac{1}{2^{n+1}}$.}
      We apply~\Cref{lemma:approx-log}. If $r = 1$, the inequality we want to prove trivially holds. Otherwise, when $r \neq 1$, 
      this inequality holds as soon as $M \geq
      (n+1+\log\abs{\ln(r)}) \big({-}2\log\abs{\frac{r-1}{r+1}}\big)^{-1}$.
      Following the definition of $M$ from line~\ref{approxln:line6}, 
      it suffices then to show that 
      $\ceil{\frac{n+1+Z_1}{2 \cdot Z_2}} \geq \frac{n+1+\log\abs{\ln(r)}}{{-}2\log\abs{\frac{r-1}{r+1}}}\,$,
      which we do by establishing that $Z_1 \geq \log\abs{\ln(r)}$ and $Z_2 \leq -\log\abs{\frac{r-1}{r+1}}$. 
      
      Let us start with $Z_1 \geq \log\abs{\ln(r)}$. Recall that ${0 < L \leq r \leq U}$ and that $Z_1$ is defined in line~\ref{approxln:line4} 
      as~$\ceil{\max(\abs{\ln(L)},\abs{\ln(U)})}$.
      Since $\log(x) \leq x$ for every $x > 0$, it suffices to show
      $\max(\abs{\ln(L)},\abs{\ln(U)}) \geq \abs{\ln(r)}$. This is immediate.
      If $r < 1$, then $0 < L \leq r$ implies $\abs{\ln(L)} \geq
      \abs{\ln(r)}$. Otherwise, if $r > 1$, then $r \leq U$ implies
      $\abs{\ln(U)} \geq \abs{\ln(r)}$.

      Let us show $Z_2 \leq -\log\abs{\frac{r-1}{r+1}}$. Note that $\abs{\frac{r-1}{r+1}} \in (0,1)$, since $r > 0$ and $r \neq 1$; hence $-\log\abs{\frac{r-1}{r+1}} > 0$. 
      Since $\frac{r-1}{r+1} \in (0,1)$, we have $-\log\abs{\frac{r-1}{r+1}} >
      -\abs{\frac{r-1}{r+1}} + 1$. By definition of $Z_2$ in
      line~\ref{approxln:line5}, it suffices to prove $-\abs{\frac{r-1}{r+1}}
      \geq \min(-\abs{\frac{L-1}{L+1}},-\abs{\frac{U-1}{U+1}})$, or,
      equivalently, $\abs{\frac{r-1}{r+1}} \leq
      \max(\abs{\frac{L-1}{L+1}},\abs{\frac{U-1}{U+1}})$. Recall that $0 < L
      \leq r \leq U$. The first derivative $f'$ of the function $f(x)
      \coloneqq \abs{\frac{x-1}{x+1}}$ is $f'(x) = \frac{2(x-1)}{(x+1)^3
      \abs{\frac{x-1}{x+1}}}$. Observe that for $x \in (0,1)$, $f'$ is always
      negative, whereas for $x > 1$, $f'$ is always positive. Therefore, if $r
      < 1$ we have $\abs{\frac{r-1}{r+1}} \leq \abs{\frac{L-1}{L+1}}$, whereas
      for $r > 1$ we have $\abs{\frac{r-1}{r+1}} \leq \abs{\frac{U-1}{U+1}}$.

    \item \textbf{proof of $\abs{t_M(r)-t_M(\abs{T_N})} \leq \frac{1}{2^{n+1}}$.}
      Recall that $t_M(x) = 2 \cdot \sum_{j=0}^M \big(\frac{1}{2j+1} \big(\frac{x-1}{x+1}\big)^{2j+1}\big)$.
      With the aim of applying~\Cref{lemma:auxiliary-for-approximations}, let us define:
      \begin{align*}
        \delta(x) &\coloneqq (x+1)^{2M+1} \prod_{j=0}^M (2j+1)\,,\\
        q(x) &\coloneqq 2 \cdot \sum_{j=0}^M \Big( (x-1)^{2j+1} (x+1)^{2(M-j)} \prod_{\substack{k=0\\k \neq j}}^M (2k+1) \Big)\,.
      \end{align*}
      Note that $\delta(x) \cdot t_M(x)$ is equivalent to $q(x)$, and that
      $\delta(\abs{T_N}) \geq 1$ (since $\abs{T_N} \geq 0$), as required by the
      lemma. Moreover, note that $\delta$ and $q$ can be rewritten as integer
      polynomials by simply expanding products such as $(x+1)^{2M+1}$ and
      $(x-1)^{2j+1}$. We analyse the degree and heights of $\delta$ and $q$ in
      this expanded form (as integer polynomials). The computation of the degree
      is straightforward, and yields $\max(\deg(\delta),\deg(q)) \leq d \coloneqq 2M+1$.
      (Note that $(x-1)^{2j+1} (x+1)^{2(M-j)}$ in the definition of $q(x)$
      expands to a polynomial in degree $2j+1+2(M-j) = 2M+1$.)
      For the height, we show that $\max(\height(\delta),\height(q)) \leq h \coloneqq (3M)^{8M}$. 
      Recall that given $m \in \N$ and $a,b \in \R$, 
      we have $(a+b)^m = \sum_{j=0}^m \binom{d}{j}a^{(m-j)}b^j$,
      which as a corollary also shows $\binom{d}{j} \leq 2^m$
      (by setting $a=b=1$). 
      Therefore (recall: $M \geq 1$), 
      \begin{align*}
        \height(\delta) 
        &\leq 2^{2M+1} \textstyle\prod_{j=0}^M (2j+1)
        \leq 2^{2M+1} (2M+1)^M \leq (3M)^{8M}.
      \end{align*}
      Similarly, for the summand $q_j(x) \coloneqq (x-1)^{2j+1} (x+1)^{2(M-j)} \prod_{\substack{k=0\\k \neq j}}^M (2k+1)$ in the definition of $q(x)$ we have
      $\height(q_j) 
        \leq 2^{2M+1} 2^{2M} \textstyle\prod_{j=0}^M (2j+1) \leq 2^{4M+1} (2M+1)^M$,
      and therefore $\height(q) \leq 2(M+1)2^{4M+1} (2M+1)^M 
      \leq (3M)^{8M}$.

      Let $K \coloneqq \max(U,2(M+1))$. 
      Note that $K \geq 1$ and $\max(\abs{r},\abs{t_M(r)}) \leq K$, 
      since $0 < r < U$ and 
      $\abs{t_M(r)} = \abs{2 \cdot \sum_{j=0}^M \big(\frac{1}{2j+1} \big(\frac{r-1}{r+1}\big)^{2j+1}\big)} \leq 2 \cdot \sum_{j=0}^M \frac{1}{2j+1} \leq 2(M+1)$, 
      because $\frac{r-1}{r+1} \in (-1,1)$.
      Following the fact that $\abs{r-\abs{T_N}} \leq \frac{1}{2^N}$,
      by applying~\Cref{lemma:auxiliary-for-approximations} with respect 
      to the above-defined objects $\delta(x)$, $q(x)$, $d$, $h$ and $K$, 
      and conclude that $\abs{t_M(r)-t_M(\abs{T_N})} \leq \frac{1}{2^{n+1}}$ 
      holds as soon as $N \geq n+1+ \log(h+1) + (2 d +1)(\log(K+1))$.
      From the definition of $N$ in line~\ref{approxln:line7}, 
      it thus suffices to show $\log(h+1) + (2 d +1)(\log(K+1)) \leq 1 + 15M \cdot \ceil{\log(U+4M)}$.
      This inequality indeed holds (recall: $U > 0$ and $M \geq 1$): 
      \begin{align*}
          &\log(h+1) + (2 d +1)(\log(K+1))\\
        \leq{}& \log((3M)^{8M}+1) + (2 (2M+1) +1)(\log(\max(U,2(M+1))+1))\\
        \leq{}& 1 + 8M \cdot \log(3M) + 7M \cdot \log(U+4M)\\ 
        \leq{}& 1 + 15M \cdot \ceil{\log(U+4M)}.
      \end{align*}
  \end{enumerate}


  \paragraph{\textit{Running time of the algorithm}}
  Lines~\ref{approxln:line1}--\ref{approxln:line5} do not 
  depend on the input $n$, and therefore the computation of $k$, $L$, $U$, $Z_1$ and $Z_2$ 
  takes constant time. Line~\ref{approxln:line6} computes in polynomial time in $n$ 
  the number $M$, which is written in unary and has size $O(n)$.
  Similarly, line~\ref{approxln:line7} computes in polynomial time in $n$ 
  the number $N$, which is written in unary and has size~$O(n \log n)$.  
  
  To conclude the proof, we show that the computation done in
  line~\ref{approxln:line8} takes time polynomial in $n$. The arguments are
  analogous to the one used at the end of the proof
  of~\Cref{lemma:poly-time-exp}. First,
  line~\ref{approxln:line8} compute the number $\abs{T_N}$; this can be done in
  time $\poly(n)$ because of the $O(n \log n)$ bound on $N$. This also means
  that $\abs{T_N}$ is of the form $\frac{\ell_1}{\ell_2}$ where $\ell_1,\ell_2$
  are non-negative integers written in binary with bit size polynomial in $n$,
  and $\ell_2 \geq 1$. The last step of the algorithm is to evaluate the
  expression $2 \cdot \sum_{j=0}^M \big(\frac{1}{2j+1}
  \big(\frac{\frac{\ell_1}{\ell_2}-1}{\frac{\ell_1}{\ell_2}+1}\big)^{2j+1}\big)$,
  which equals the rational number $\frac{b_1}{b_2}$, where 
  \begin{align*}
    b_1 \coloneqq \sum_{j=0}^M \Big( (\ell_1+\ell_2)^{2(M-j)}(\ell_1-\ell_2)^{2j+1}\prod_{\substack{k=0\\k\neq j}}^M(2k+1) \Big),\qquad
    b_2 \coloneqq (\ell_1+\ell_2)^{2M+1} \prod_{j=0}^M (2j+1).
  \end{align*}
  The algorithm returns~$\frac{b_1}{b_2}$. Clearly,
  computing $b_1$ and $b_2$ using the above expressions can be done in
  polynomial time in $n$, just following the arithmetic operations. 
\end{proof}


Thanks to Lemmas \ref{lemma:approx-algebraic}, \ref{lemma:poly-time-exp} 
and \ref{lemma:poly-time-log} and Lemmas \ref{lemma:turing-machine-reciprocal} and \ref{lemma:turing-machine-products} from the preliminaries,
we can now construct polynomial-time Turing machine for all the numbers in 
\Cref{theorem:result-root-barrier}(\ref{theorem:result-root-barrier:point2}). As an example, to construct the Turing machine for~$\frac{\ln(\alg)}{\ln(\beta)}$ 
we construct machines for the following sequence of numbers: $\alg$ and $\beta$ (applying~\Cref{lemma:approx-algebraic}), 
$\ln(\alg)$ and $\ln(\beta)$ (\Cref{lemma:poly-time-log}), $\frac{1}{\ln(\beta)}$ (\Cref{lemma:turing-machine-reciprocal})
and $\frac{1}{\ln(\beta)} \cdot \ln(\alg)$ (\Cref{lemma:turing-machine-products}). 
For $\alg^\eta$, we follow the operations in $e^{\eta \cdot \ln(\alg)}$.

\subsubsection{Root barriers.}
We now study root barriers for the numbers in~\Cref{theorem:result-root-barrier}(\ref{theorem:result-root-barrier:point2}).
In the context of transcendental numbers, root barriers are usually called
\emph{transcendence measures}. Several fundamental results in number theory
concern deriving a transcendence measure for ``illustrious'' numbers, such as
Euler's $e$, $\pi$, or logarithms of algebraic
numbers~\cite{Popken29,Mahler32,Waldschmidt78}. The results we are interested in
are summarised in~\Cref{table:transcendence-degrees}, which is taken almost verbatim
from~\cite[Fig.~1 and Corollary~4.2]{Waldschmidt78}. All transcendence
measures in the table are \emph{polynomial} root barriers. Note that in the cases of
$\alg^\eta$ and $\frac{\ln \alg}{\ln \beta}$, the transcendence measures hold
under further assumptions, which are given in the caption of the table. 
To complete the proof of \Cref{theorem:result-root-barrier}(\ref{theorem:result-root-barrier:point2}), it suffices removing these assumptions.

\paragraph{Case: $\cn = \alg^\eta$.}
In this case, \Cref{table:transcendence-degrees} assumes~$\eta$ to be irrational. 
The following lemma shows that the rationality of an algebraic number $\eta$
represented by~$(q,\ell,u)$ can be checked, and whenever $\eta$ is rational, 
there is a way to compute it as a pair of integers.

\begin{table} 
  \begin{center}
  \def\arraystretch{1.15}
    \begin{tabular}{cll} 
      Number &  Transcendence measure from~\cite{Waldschmidt78} &  Simplified bound ($\alg,\beta,\eta$ fixed)\\[2pt]
      \midrule
      \rule{0pt}{1.1\normalbaselineskip}
      $\pi$ & $2^{40} d (\ln h + d \ln d)(1 + \ln d)$
      & $O(d^2 (\ln d)^2 \ln h)$\\
      $e^\pi$ & $2^{60} d^2 (\ln h + \ln d)(\ln \ln h + \ln d)(1 + \ln d)$
      & $O(d^2 (\ln d)^3 (\ln h) (\ln \ln h))$\\
      $e^\eta$ & $c_\eta \cdot d^2(\ln h + \ln d)\big(\frac{\ln \ln h + \ln d}{\ln \ln h + \ln \max(1,\ln d)}\big)^2$
      & $O(d^2 (\ln d)^3 (\ln h)(\ln \ln h)^2)$\\ 
      $\alg^\eta$ 
      & $c_{\alg,\eta} \cdot d^3 (\ln h + \ln d)\frac{\ln \ln h + \ln d}{(1+\ln d)^2}$
      & $O(d^3 (\ln d)^2(\ln h)(\ln \ln h))$\\
      $\ln \alg$ & $c_\alg \cdot d^2 \frac{\ln h + d \ln d}{1+\ln d}$ & $O(d^3 (\ln d) \ln h)$\\ 
      $\frac{\ln \alg}{\ln \beta}$ 
      & $c_{\alg,\beta} \cdot d^3 \frac{\ln h + d \ln d}{(1+\ln d)^2}$
      & $O(d^4 (\ln d) \ln h)$\\
    \end{tabular}
    \vspace{10pt}
  \end{center}
\caption{Transcendence measures for some classical real numbers. 
For convenience only, the table assumes $h \geq 16$ (so that $\ln \ln h \geq 1$; replace $h$ by $h+15$ to avoid this assumption).
The numbers $\alg > 0$, $\beta > 0$ and $\eta$ are fixed algebraic numbers, with $\beta \neq 1$.
The integers $c_{\eta}$, $c_{\alg,\eta}$, $c_{\alg}$ and $c_{\alg,\beta}$ are constants that depend on, and can be computed from, polynomials representing $\alg$, $\beta$ and $\eta$.
In the case of $\alg^\eta$, $\eta$ is assumed to be irrational.
In the last line of the table, $\frac{\ln \alpha}{\ln \beta}$ is assumed to be irrational.}
\label{table:transcendence-degrees}%
\end{table}


\begin{restatable}{lemma}{LemmaCheckRationality}
  \label{lemma:check-rationality}
  There is an algorithm deciding whether an input algebraic number $\eta$ 
  represented by $(q,\ell,u)$ is rational.
  When $\eta$ is rational, the algorithm returns $m,n \in \Q$ such that 
  $\eta = \frac{m}{n}$.
\end{restatable}
\begin{proof}
  By relying on the LLL-based algorithm from~\cite{lenstra1982}, 
  we can compute (in fact, in polynomial time)
  a decomposition of the univariate polynomial $q$ into 
  irreducible polynomials (below, factors) with rational coefficients. 
  Let $E$ be the (finite) set of those factors having degree~$1$. 
  Since $\eta$ is a root of $q$, we have that $\eta$ is rational if and only if it is a root of a polynomial in~$E$.
  Every element of $E$ is a linear polynomial of the form $n \cdot x - m$, where $n,m \in \Q$, having root $\frac{m}{n}$. 
  Recall that $\eta$ is the only root of $q$ in the interval $[\ell,u]$, and therefore, in order to check whether $\eta$ is rational, 
  it suffices to check whether there is $(n \cdot x - m) \in E$
  such that $\ell \leq \frac{m}{n} \leq u$. If the answer is positive, $\eta = \frac{m}{n}$. Otherwise, $\eta$ is irrational. 
\end{proof}

When $\eta$ is irrational the polynomial root barrier for $\alg^\eta$ is given
in~\Cref{table:transcendence-degrees}. Otherwise, $\eta = \frac{m}{n}$ and the
number $\alg^{\frac{m}{n}}$ is algebraic. In this case, rely on the following
lemma to construct a representation of $\alg^{\frac{m}{n}}$, and then derive a
polynomial root barrier by applying~\autoref{theorem:alg-root-barrier}.

\begin{restatable}{lemma}{LemmaRepresentationPowerOfAlgebraic}
  \label{lemma:representation-power-of-algebraic}
  There is an algorithm that given a rational $r$ and an algebraic number $\alg > 0$ 
  represented by $(q,\ell,u)$, computes a representation $(q',\ell',u')$ 
  of the algebraic number $\alg^r$.
\end{restatable}
\begin{proof}
  Let $r = \frac{m}{n}$ with $m \in \Z$ and $n \geq 1$, and let $q(x) =
  \sum_{i=0}^d a_i \cdot x^i$, with $\deg(q) = d$, and $h \coloneqq \height(q)$.
  Since we are not interested in the runtime of this algorithm, we can apply the
  procedure explained in the last paragraph of the preliminaries
  (\Cref{sec:preliminaries}) to impose that (in addition to $\alg$ being the
  only root of $q$ in the interval $[\ell,u]$) either $\ell = u$ or $\alg \in
  (\ell,u)$ and $(\ell,u) \cap \Z = \emptyset$ holds. Since $\alg > 0$, by
  applying~\autoref{theorem:alg-root-barrier} to the polynomial~$x$ we derive $\alg
  \geq 2^{-d}(h(d+1))^{-1}$, and so we can update $\ell$ and $u$ to be both
  strictly positive.
  
  First, let us reduce the problem to the case $m \geq 1$. If $m = 0$ or then
  $\alg^0 = 1$ and we can simply return $(x-1,1,1)$. To handle the case $m < 0$,
  we remark that $\alg^{-1}$ is a root of the Laurent polynomial $\sum_{i = 0}^d
  a_i \cdot x^{-i}$, and thus also of $x^d \cdot \sum_{i=0}^d a_i \cdot x^{-i}$.
  So, the polynomial $q''(x) \coloneqq \sum_{i=0}^d a_i \cdot x^{d-i}$ is such
  that $(q'',u^{-1},\ell^{-1})$ represents $\alg^{-1}$ (note that no
  root~$\beta$ of $q''$ that is distinct from $\alg^{-1}$ can lie in the
  interval $[u^{-1},\ell^{-1}]$, else $\beta^{-1} \neq \alg$ would lie in
  $[\ell,u]$). We can then compute the representation of $\alg^r$ starting from
  $(q'',u^{-1},\ell^{-1})$, and considering the positive rational $-r$ instead
  of $r$.

  Below, assume $m,n \geq 1$. We start by computing a polynomial $Q(x)$ 
  having $\alg^m$ as a root.
  Since $q(\alg)=0$, for every $j \in \N$, we can express $\alg^j$ as a
  rational linear combination $\mu(j)$ of the terms $1,\alg,\dots,\alg^{d-1}$: 
  \[ 
    \mu(j) \coloneqq  
      \begin{cases}
        \alg^j &\text{if $j \in [0..d-1]$}\\
        \sum_{i=0}^{d-1}\frac{-a_i}{a_d} \alg^{i} &\text{if $j = d$}\\
        b_{d-1}\mu(d) + \sum_{i=0}^{d-2}b_i\alg^{i+1} &\text{if $j > d$, where $\mu(j-1) = \sum_{i=0}^{d-1} b_i \alg^i$}.
      \end{cases}
  \]
  (Note that the last line in the definition of $\mu(j)$ is obtained by
  multiplying $\mu(j-1)$ by $\alg$, to then replace $\alg^d$, which is the only
  monomial with degree above $d-1$, by $\mu(d)$.)

  We can represent the polynomial $\mu(j) = \sum_{i=0}^{d-1} b_i \alg^i$ as the
  vector $(b_0,\dots,b_{d-1}) \in \Q^d$. Consider now the family of polynomials
  $\mu(0),\mu(m),\mu(2m),\dots,\mu(i \cdot m),\dots,\mu(d \cdot m)$. These
  correspond to a set of $d+1$ vectors in $\Q^d$, and therefore they are
  rationally dependent: there is a non-zero vector $(k_0,\dots,k_{d}) \in
  \Q^{d+1}$ such that 
  \[ 
    k_0 \cdot \mu(0) + k_1 \cdot \mu(m) + \dots + k_d \cdot \mu(d \cdot m) = 0.
  \]
  Since $\mu(j) = \alg^j$ for all $j \in \N$, we then conclude that
  $\sum_{j=0}^d k_j \alg^{j \cdot m} = 0$. Let $g$ be the least common multiple
  of the denominators of the rational numbers $k_0,\dots,k_{d}$, and define
  $\hat{k}_j = g \cdot k_j$ for all $j \in [0..d]$. Then, $\alg^m$ is a root of
  the non-zero integer polynomial $Q(x) \coloneqq \sum_{j=0}^d \hat{k}_j \cdot
  x^{j}$.

  We can now take $q'(x) \coloneqq Q(x^n)$ in order to obtain a polynomial
  having $\alg^{\frac{m}{n}}$ as a root.


  Now we move on to the problem of isolating $\alg^{\frac{m}{n}}$ from all other
  roots of $q'(x)$ by opportunely defining a separating interval $[\ell',u']$
  where $\ell'$, $u'\in\Q$.

  If $q'$ has degree $1$, then $\alg^{\frac{m}{n}}$ is its only root and it is
  rational. Finding an interval is in this case trivial: given $q'(x) = b \cdot
  x - a$, we have $\alg^{\frac{m}{n}} = \frac{a}{b}$ and so we can take $\ell' =
  u' = \frac{a}{b}$. Hence, below, let us assume $\deg(q') \geq 2$.
  To compute $\ell'$ and $u'$ we use the following result. 

  \begin{claim}
    \label{claim:mean-value}
    Let $0<\ell\leq u$ be rational numbers. Consider a function $f(x)$
    that is both increasing and continuously differentiable in the
    interval~$[\ell,u]$. Let $\delta > 0$ be an upper bound to the maximum of
    its derivative over $[\ell,u]$. If $\abs{u-\ell}\leq \frac{D}{\delta}$, then
    $\abs{f(\ell)-f(u)}\leq D$.
  \end{claim}
  \begin{proof}
    Since $f(x)$ is continuously differentiable over $[\ell,u]$, 
    by the mean value theorem we have $\frac{f(u)-f(\ell)}{u-\ell} \leq \delta$. 
    Moreover, since $f(x)$ is increasing inside $[\ell,u]$, then 
    $\frac{f(u)-f(\ell)}{u-\ell} = \frac{\abs{f(u)-f(\ell)}}{\abs{u-\ell}}$.
    We conclude that $\abs{f(u)-f(\ell)}\leq \delta\cdot\abs{u-\ell} \leq \delta \cdot  \frac{D}{\delta} \leq D$.
  \end{proof}


  Below, let $h' \coloneqq \height(q')$ and $\deg(q') \coloneqq d'$.
  By applying~\cite[Theorem~A.2]{Bugeaud04}, any two distinct roots $\alg_1$ and
  $\alg_2$ of $q'$ satisfy:
  \begin{equation}
    \label{eq:close-roots}
    \abs{\alg_1 - \alg_2} > D \coloneqq 2^{-d'-1} (d')^{-4d'} (h')^{-2d'}.
  \end{equation}
  Let $\delta \coloneqq \max_{x\in\{\ell,u\}}\{r\cdot x^{r-1}\}$, which is
  maximum of the derivative of $f(x) \coloneqq x^r$ in the interval $[\ell,u]$.
  Let us apply the algorithm in~\Cref{lemma:approx-algebraic} in order to refine
  the interval $[\ell,u]$ containing $\alg$ so that we achieve $\abs{\ell-u}\leq
  \frac{D}{2\delta}$.

  Note that, since $r > 0$, the function $f$ is increasing and continuously
  differentiable in $[\ell,u]$, from $\alg \in [\ell,u]$ we have $\alg^r\in
  [\ell^r,u^r]$. Moreover, by~\Cref{claim:mean-value}, 
  we have $\abs{u^r-\ell^r} \leq \frac{D}{2}$.
  From~\Cref{eq:close-roots}, 
  we conclude that $\alg^r$ is the only root in the interval $[\ell^r,u^r]$.

  Note that, in general, $\ell^r$ and $u^r$ are not rational numbers, 
  hence we cannot use $(q',\ell^r,u^r)$ in order to represent $\alg^r$. 
  Instead, we now compute two rational numbers $\ell'<\ell^r$ and $u'>u^r$ 
  such that $\alg^r \in [\ell',u']$ and, crucially, $\abs{u'-\ell'}\leq D$. 
  Again, by~\Cref{eq:close-roots}, we will conclude that 
  $\alg^r$ is the only root of $q'$ in $[\ell',u']$, 
  and therefore $(q',\ell',u')$ represents $\alg^r$.

  In order to compute $\ell'$ and $u'$, we rely on two Turing machines $T$ and
  $T'$ computing $\ell^r$ and $u^r$, respectively. To construct these machines,
  we simply apply~\Cref{lemma:turing-machine-products}, \Cref{lemma:poly-time-exp}
  and~\Cref{lemma:poly-time-log}, seeing $\ell^r$ as $e^{r \cdot
  \ln(\ell)}$ and $u^r$ as $e^{r \cdot \ln(u)}$ (note that $\ell,u > 0$, hence
  the two logarithms are well-defined). Since $\ell^r$ and $u^r$ are positive,
  w.l.o.g.~we can assume the outputs of $T$ and $T'$ to be always non-negative.
  Indeed, to force this condition on, e.g., $T$, we can consider a new Turing
  machine that on input $n \in \N$ returns $\abs{T_n}$; this new Turing machine
  still computes $\ell^r$.
  Let $M \coloneqq -\floor{\log(D)}$, and observe that $M \geq 1$, since $D \in (0,1)$.
  
  We are now ready to define the rationals $\ell'$ and $u'$: 
  \[ 
    \ell'\coloneqq T_{M+3}-\frac{1}{2^{M+3}}
    \qquad\text{and}\qquad 
    u'\coloneqq T'_{M+3}+\frac{1}{2^{M+3}}.
  \]
  Recall that $\abs{\ell^r-T_{M+3}} \leq \frac{1}{2^{M+3}}$, and similarly
  $\abs{u^r-T_{M+3}} \leq \frac{1}{2^{M+3}}$. Therefore, $\ell' \leq \ell^r \leq u^r \leq u'$, which in turn implies that $\alg^r \in [\ell',u']$. Moreover,
  we also conclude that $\ell^r-\frac{1}{2^{M+2}} \leq \ell'$ and $u'\leq
  u^r+\frac{1}{2^{M+2}}$.
  At last, let us show that $\abs{u'-\ell'} \leq D$:
  \begin{align*}
    \abs{u'-\ell'}&\leq \abs{u^r+\frac{1}{2^{M+2}}-\Big(\ell^r-\frac{1}{2^{M+2}}\Big)}
    &\quad\tag*{since $\ell^r-\frac{1}{2^{M+2}} \leq \ell' \leq u'\leq
  u^r+\frac{1}{2^{M+2}}$}\\
    &\leq\abs{u^r-\ell^r}+\frac{1}{2^{M+1}}\\
    &\leq \frac{D}{2} + \frac{1}{2^{-\floor{\log(D)}+1}}
    &\tag*{by def.~of~$M$ and since~$\abs{u^r-\ell^r} \leq \frac{D}{2}$}\\
    &\leq \frac{D}{2} + \frac{D}{2} \leq D.
    &&\tag*{\qedhere}
  \end{align*}
\end{proof}

\paragraph{Case: $\cn = \frac{\ln(\alg)}{\ln(\beta)}$.}
In this case,~\Cref{table:transcendence-degrees} assumes $\cn$ to be irrational.
Since $\cn$ is positive, $\alg,\beta \not\in \{0,1\}$. The following lemma
illustrates how $\cn = \frac{m}{n}$ if and only if~${\alg^n \beta^{-m} = 1}$.

\begin{restatable}{lemma}{LemmaMultIndependeceRationality}
  \label{lemma:mult-independence-rationality}
  Let $\alg$ and $\beta$ be two algebraic numbers different from $0$ and $1$. Then, $\alpha$ and $\beta$ are multiplicatively dependent 
  if and only if $\frac{\ln(\alg)}{\ln(\beta)}$ is rational.
\end{restatable}
\begin{proof}
  Let $n,m \in \Z$. With either $n$ or $m$ distinct from zero.
  We have 
  \begin{align*}
    \alg^n = \beta^m 
    \iff \ln(\alg^n) = \ln(\beta^m)
    \iff n\ln(\alg) = m\ln(\beta)
    \iff \frac{\ln(\alg)}{\ln(\beta)} = \frac{m}{n},
  \end{align*}
  where we note that one of the two sides of the equality $n\ln(\alg) =
  m\ln(\beta)$ must be non-zero (because $n$ or $m$ are non-zero, and
  $\alg,\beta \neq 1$) which makes non-zero also the other side.
\end{proof}


We can then compute a polynomial root barrier for $\cn =\frac{\log(\alpha)}{\log(\beta)}$
as described below:

\begin{remark}\label{remark:multiplicative-independence}
From celebrated result of Masser~\cite{Masser88}, 
the set ${\{(m,n) \in \Z^2 : \alg^n \beta^{-m} = 1\}}$ is a
finitely-generated integer lattice for which we can explicitly construct a
basis~$K$ (see~\cite{CaiLZ00} for a polynomial-time procedure). If $K =
\{(0,0)\}$, then $\cn$ is irrational and its polynomial
root barrier is given in~\Cref{table:transcendence-degrees}. Otherwise, since
$\alg,\beta \not \in \{0,1\}$, there is~$(m,n) \in K$ with $n \neq 0$, and
$\cn = \frac{m}{n}$. Then, a polynomial
root barrier can be derived by~\autoref{theorem:alg-root-barrier}.
\end{remark}

The proof of \Cref{theorem:result-root-barrier}(\ref{theorem:result-root-barrier:point2}) is now complete.

\section{An application: the entropic risk threshold problem}
\label{sec:application-entropic-risk}

We apply some of the machinery developed
for~$\exists\R(\ipow{\cn})$ to remove the appeal to Schanuel's conjecture from
the decidability proof of the entropic risk threshold problem
for stochastic games from~\cite{BaierCMP23}.
Briefly, a \emph{(turn-based) stochastic game} is a tuple $\mathsf{G} = (S_{\max},S_{\min},A,\Delta)$ where $S_{\max}$ and $S_{\min}$ are disjoint finite set of \emph{states} controlled by two players, $A$ is a function from states to a finite set of \emph{actions}, 
and $\Delta$ is a function taking as input a state $s$ and an action from $A(s)$, 
and returning a \emph{probability distribution} on the set of states. Below, we write $\Delta(s,a,s')$ for the probability associated to $s'$ in $\Delta(s,a)$, 
and set $S \coloneqq S_{\max} \cup S_{\min}$.

Starting from an initial state $\hat{s}$, a play of the game produces an
infinite sequence of states~$\rho=s_1s_2s_3\dots$ (a path), to which we
associate the \emph{total reward} $\sum_{i=1}^{\infty} r(s_i)$, where $r \colon
S \to \R_{\geq 0}$ is a given \emph{reward function}. 
A classical problem is to determine the strategy for one of the
players that optimises (minimises or maximises) its expected total reward.
Instead of expectation, the \emph{entropic risk} yields the normalised logarithm
of the average of the function $b^{-\eta X}$, where the \emph{base} $b > 1$
and the \emph{risk aversion factor} $\eta > 0$ are real numbers, and $X$ is a
random variable ranging over total rewards. We refer the reader
to~\cite{BaierCMP23} for motivations behind this notion, as well as all formal
definitions.

Fix a base~$b > 1$ and a risk aversion factor $\eta \in \R$.
The \emph{entropic risk threshold problem} $\ERisk[b^{-\eta}]$ asks to
determine if the entropic risk is above a threshold $t$. The inputs are
a stochastic game $\mathsf{G}$ having rational probabilities~$\Delta(s,a,s')$,
an initial state $\hat{s}$, a reward function $r \colon S \to \Q_{\geq 0}$ and a
threshold $t \in \Q$. In~\cite{BaierCMP23}, this problem is proven to be in \pspace
for $b$ and $\eta$ rationals, and decidable subject to Schanuel's conjecture
if~$b = e$ and $\eta \in \Q$ (both results also hold when $b$ and $\eta$
are not fixed). We improve upon the latter result, by establishing the following
theorem (that assumes having representations of $\alpha$ and~$\eta$):

\begin{theorem}
    \label{thoerem:Erisk}
    Both $\ERisk[e^{-\eta}]$ and $\ERisk[\alpha^{-\eta}]$ are in~\exptime for every fixed algebraic numbers $\alpha,\eta$.
    When $\alpha,\eta$ are not fixed but part of the input, 
    these problems~are~decidable.
\end{theorem}

\begin{proof}[Proof sketch]
Ultimately, in~\cite{BaierCMP23} the authors show that the problem $\ERisk[b^{-\eta}]$ is reducible in polynomial time to the problem of checking the satisfiability of a system of constraints of the following form (see~\cite[Equation~7]{BaierCMP23} for an equivalent formula):
\vspace{-6pt}
\begin{equation}%
    \label{formula:BCMP}%
    v(\hat{s}) \leq (b^{-\eta})^{t}
    \land\!\! \bigwedge_{s \in T} v(s) = d_s
    \land\!\! \bigwedge_{s \in S} v(s) = \oplus_{a \in A(s)} \Big((b^{-\eta})^{r(s)} \sum_{s' \in S} \Delta(s,a,s') \cdot v(s')\Big),
\end{equation} 
where $T$ is some subset of the states $S$ of the game, $d_s \in \{0,1\}$, and in the notation $\oplus_{a \in
A(s)}$ the symbol $\oplus$ stands for the functions $\min$ or
$\max$, depending on which of the two players controls~$s$.
The formula has one variable $v(s)$
for every $s \in S$, ranging over $\R$.

Since $z = \max(x,y)$ is equivalent to $z \geq x
\land z \geq y \land (z = x \lor z = y)$, and $z = \min(x,y)$ is
equivalent to $z \leq x \land z \leq y \land (z = x \lor z = y)$, 
except for the rationality of the exponents $t$ and $r(s)$ (which we handle below),
Formula~\ref{formula:BCMP} belongs to $\exists \R((b^{-\eta})^{\mathbb{Z}})$.

Fix $b > 1$ to be either $e$ or algebraic, 
and $\eta > 0$ to be algebraic. 
Assume to have access to representations for these algebraic numbers, 
so that if $\eta$ is represented by $(q(x),\ell,u)$, then $-\eta$ is represented by $(q(-x),-u,-\ell)$.
Consider the problem of checking whether a formula $\phi$ of the form given by~Formula~\ref{formula:BCMP} is satisfiable. 
Since $\phi$ does not feature predicates~$(b^{-\eta})^{\mathbb{Z}}$, but only the constant~$b^{-\eta}$, 
instead of~\Cref{algo:main-procedure} we can run the following simplified procedure:%
\begin{enumerate}[I.]
    \item \emph{Update all exponents $t$ and $r(s)$ of $\phi$ to be over $\N$
    and written in unary.} \textbf{(1)}~Compute the l.c.m.~$d \geq 1$ of the
    denominators of these exponents. \textbf{(2)}~Rewrite every
    term~$(b^{-\eta})^{\frac{p}{q}}$, where $\frac{p}{q}$ is one such
    exponent, into $(b^{\frac{-\eta}{d}})^{\frac{p \cdot d}{q}}$. Note that
    $\frac{p \cdot d}{q} \in \Z$. \textbf{(3)}~Rewrite $\phi$ into
    $\phi\sub{x}{b^{\frac{-\eta}{d}}} \land  x^d = b^{-\eta} \land x \geq
    0$, with $x$ fresh variable. \textbf{(4)}~Opportunely multiply both sides of
    inequalities by integer powers of~$x$ to make all exponents range over~$\N$.
    \textbf{(5)}~Change to a unary encoding for the exponents by adding further
    variables, as done in the proof
    of~\Cref{theorem:result-root-barrier}(\ref{theorem:result-root-barrier:point1})
    (\Cref{sec:poly-evaluation}). This step takes polynomial time
    in~$\size(\phi)$.
    \item \emph{Eliminate $x$ and all variables $v(s)$ with~$s \in S$.} This is
    done by appealing to~\autoref{theorem:basu}, treating~$b^{-\eta}$ as a free
    variable. The result is a Boolean combination $\psi$ of polynomial
    inequalities over $b^{-\eta}$. This step runs in time exponential in
    $\size(\phi)$.
    \item \emph{Evaluate $\psi$.} Call~\Cref{algo:sign-evaluation} on each
    inequality, to then return $\top$ or $\bot$ according to the Boolean
    structure of $\psi$. Since we can construct a polynomial-time Turing machine
    for $b^{-\eta}$ (\Cref{sec:poly-evaluation}),
    by~\Cref{lemma:runtime-sign-poly-root-barrier} this step takes polynomial
    time in $\size(\psi)$. 
    \qedhere
\end{enumerate}
\end{proof}

\section{Conclusion and future directions}
\label{sec:conclusion}

With the goal of identifying unconditionally decidable fragments or variants of
$\R(e^x)$, we have studied the complexity of the theory~$\exists \R(\ipow{\cn})$
for different choices of~${\cn > 0}$. Particularly valuable turned out to be the
introduction of root barriers~(\Cref{definition:intro:root-barrier}): by relying
on this notion, we have established that~$\exists \R(\ipow{\cn})$ is in
\expspace if $\cn$ is algebraic, and in \threeexptime for natural choices of
$\cn$ among the transcendental numbers, such as $e$ and $\pi$.

A first natural question is how far are we from the exact complexity of these
existential theories, considering that the only known lower bound  is
inherited from the existential theory of the reals, which lies in
\pspace~\cite{Canny88}. While we have no answer to this question, we remark that
strengthening the hypotheses on $\cn$ may lead to better complexity bounds.
For example, we have shown that our \expspace result for algebraic
numbers improves to \nexptime when $\cn$ is a natural number.


We have presented natural examples of bases $\cn$ having polynomial root
barriers. 
More exotic instances are known: setting $\cn
= q(\pi,\Gamma(\frac{1}{4}))$, where $q$ is an integer polynomial and~$\Gamma$
is Euler's Gamma function, results in one such base. This follows from a theorem
by Bruiltet~\cite[Theorem B$^\prime$]{Bruiltet02} on the algebraic independence
of $\pi$ and $\Gamma(\frac{1}{4})$.
This leads to a second natural question: are there $a,b \in \R$ satisfying
$\ipow{a} \cap \ipow{b} = \{1\}$ for which the existential theory of the reals
enriched with both the predicates $\ipow{a}$ and $\ipow{b}$ is decidable? The
undecidability proof of the full FO theory proven in~\cite{blanchard23} relies on three quantifier~alternations. 

\bibliographystyle{alphaurl}
\bibliography{bibliography}




\end{document}